\documentclass[preprint,12pt]{article}

\usepackage{amsmath}
\usepackage{graphicx}
\usepackage{enumerate}
\usepackage{natbib}
\usepackage{typearea}

\usepackage{amsfonts, amsmath, amssymb, amsthm, dsfont}
\usepackage{amsthm}			% Eigene Umgebungen
\usepackage{thumbpdf}	% Seitenvorschau in PDF-Dokumenten
\usepackage{natbib}
\usepackage{booktabs}
\usepackage{multirow}
\usepackage{todonotes}

\usepackage{pgf, tikz}
\usepackage{tikz-cd}		% Kommutative Diagramme
\usepackage{bbm}
\allowdisplaybreaks

\theoremstyle{definition}
\newtheorem{definition}{Definition}[section]

\newtheorem*{assumption*}{Assumption}

\newtheorem*{condition*}{Condition}
\newtheorem{example}{Example}
\theoremstyle{plain}
\newtheorem{theorem}[definition]{Theorem}
\newtheorem{proposition}[definition]{Proposition}
\newtheorem{lemma}[definition]{Lemma}

\theoremstyle{remark}
\newtheorem{remark}{Remark}

\newcommand{\crossref}[2]{#2}
\newcommand{\crosseqref}[2]{#2}

\renewcommand{\crossref}[2]{\ref{#1}}
\renewcommand{\crosseqref}[2]{\eqref{1}}
\usepackage{graphicx}
\usepackage{hyperref}
\usepackage[modulo]{lineno}

\newcommand{\N}{\mathbb{N}}
\newcommand{\Z}{\mathbb{Z}}

\newcommand{\R}{\mathbb{R}}
\newcommand{\E}{\mathbb{E}}

\newcommand{\F}{\mathcal{F}}
\newcommand{\pconv}{\xrightarrow{P}}

\newcommand{\Var}{\operatorname{Var}}
\newcommand{\Cov}{\operatorname{Cov}}
\newcommand{\Cor}{\operatorname{Cor}}
\newcommand{\Kurt}{\operatorname{Kurt}}

\newcommand{\deq}{\overset{d}{=}}

\newcommand{\wconv}{\Rightarrow}

\newcommand{\beps}{\boldsymbol{\epsilon}}

\newif\ifhideproofs
\ifhideproofs
  \usepackage{environ}
  \NewEnviron{hide}{}

\fi

\title{Functional estimation and change detection for nonstationary time series\footnote{This is an Accepted Manuscript of an article published by Taylor \& Francis in the \emph{Journal of the American Statistical Association} on September 27, 2021, available at: https://www.tandfonline.com/10.1080/01621459.2021.1969239.}}
\author{Fabian Mies\\ RWTH Aachen University}
\date{}
\begin{document}

\maketitle

\begin{abstract}
	Tests for structural breaks in time series should ideally be sensitive to breaks in the parameter of interest, while being robust to nuisance changes.
	Statistical analysis thus needs to allow for some form of nonstationarity under the null hypothesis of no change.
	In this paper, estimators for integrated parameters of locally stationary time series are constructed and a corresponding functional central limit theorem is established, enabling change-point inference for a broad class of parameters under mild assumptions.
	The proposed framework covers all parameters which may be expressed as nonlinear functions of moments, for example kurtosis, autocorrelation, and coefficients in a linear regression model.
	To perform feasible inference based on the derived limit distribution, a bootstrap variant is proposed and its consistency is established.
	The methodology is illustrated by means of a simulation study and by an application to high-frequency asset prices.
	
	\textbf{Keywords:} gradual change; locally stationary process; $p$-variation; bootstrap inference
\end{abstract}

\section{Introduction}

While statistical theory has historically been concerned with the study of data which is identically distributed, or at least stationary, this temporal homogeneity is often violated in practice.
For example, economic time series are typically heteroscedastic, which is especially prominent for returns of financial assets.
Statistical inference which does not account for nonstationarity may thus lead to wrong conclusions.
However, the treatment of nonstationary models poses methodological challenges. 
Specifying the temporal variation nonparametrically, inference is often based on   asymptotic methods, and it is in general not clear how to define a suitable limit in a nonstationary setting.
As a remedy, \cite{Dahlhaus1997} suggested to regard the time index as a fraction of the sample size, which allows for infill asymptotics.
This perspective leads to the concept of locally stationary time series, which has since been extended and generalized, see e.g.\ \cite{Wu2011,Dahlhaus2017,Truquet2019}.

In this paper, we study a multivariate, locally stationary time series model given by the causal representation
\begin{align}
	X_{t,n} = G_n(\tfrac{t}{n}, \beps_t )\in\R^d,\quad t=1,\ldots, n, \label{eqn:X-intro}
\end{align}
where $\beps_t = (\epsilon_t, \epsilon_{t-1},\ldots)\in\R^\infty$ for iid random variables $\epsilon_i$.
For each fixed $u\in[0,1]$ and $n\in\N$, the sequence $(G_n(u, \beps_t))_{t\in\Z}$ is stationary.
By letting the kernel depend on the fraction $\frac{t}{n}$, the model explicitly accounts for nonstationarity.
We assume that the kernel $G_n$ converges to a limiting kernel $G$ as $n\to\infty$ in $L_q(P)$, and that $u\mapsto G_n(u,\beps_t)$ admits some form of regularity.
The nonlinear model \eqref{eqn:X-intro} has been introduced by \cite{Zhou2009}, and later been refined by \cite{Zhou2013}.
In contrast to the existing studies, the regularity assumptions imposed on the mapping $u\mapsto G_n(u,\beps_t)$ in this paper are much less restrictive, as we only require finiteness of some $p$-variation instead of Hölder continuity; see the detailed discussion in Section \ref{sec:model}.

A statistician might be interested in various properties of the time series $X_{t,n}$.
We denote the quantity of interest by a local parameter $\theta_u, u\in[0,1]$, or $\theta_u^n$, which is a functional of the law of $G(u,\beps_t)$, resp.\ $G_n(u,\beps_t)$.
The estimator proposed in this paper is applicable for any parameter of the form $\theta^n_u = f(\mu^n_u)$, where $\mu^n_u = \E G_n(u,\beps_t)$, and $f$ is a sufficiently smooth function.
Upon replacing $X_{t,n}$ by $Y_{t,n}=h(X_{t,n})$ for a function $h$, we may also study parameters $\theta^n_u = f(\E[ h(G_n(u,\beps_t))])$, i.e.\ any parameter which may be expressed as a function of moments of $X_{t,n}$.
This framework is rather general, and contains for example the variance, kurtosis, and autocorrelations at fixed lag.

Depending on the application, different functionals of the temporal variation $u\mapsto\theta_u$ might be of interest, e.g.\ its temporal average $\int_0^1 \theta_u\, du$ as studied by \cite{Potiron2018}, its maximum value $\sup_{u\in[0,1]} \|\theta_u\|$, or its value $\theta_{u_0}$ at some $u_0\in[0,1]$.
For instance, estimation of $\theta_{u_0}$ is studied by \cite{Cui2020}, who estimate the autocovariance function of a locally stationary time series via polynomial smoothing, and by \cite{Dahlhaus2019}.
The latter example is a nonparametric problem, and thus in general suffers from slow rates of convergence, depending on the regularity of $u\mapsto \theta_u$.
In this paper, we tackle the temporal variation by studying the integrated parameter
\begin{align*}
	u\mapsto\Theta(u) = \int_0^u \theta_v\, dv,
\end{align*}
and we propose a corresponding estimator $M_n(u)$ of $\Theta(u)$. 
The function $\Theta(u)$ contains all information about the local parameter $\theta_u$ and is thus a nonparametric object as well.
However, from a statistical perspective, it is attractive to formulate hypotheses in terms of $\Theta(u)$ since the latter may be estimated at a parametric rate $\sqrt{n}$, as demonstrated by our functional central limit theorem presented in Section \ref{sec:estimator}.
We note that the idea of recovering a $\sqrt{n}$ rate of convergence via integration has also been employed in other areas of nonparametric estimation, e.g.\ when performing inference for integrated squared density derivatives in an iid setting \citep{Hall1987, Bickel1988}, for convolutions of nonparametrically specified densities \citep{Schick2004}, and for quadratic integrals of derivatives of a regression function \citep{Huang1999}.
The previous references consider estimation of a single integrated quantity, but the functional estimation of a local parameter $\theta_v$ via its integral function $u\mapsto\Theta(u)$ is also common practice in high-frequency econometrics, when estimating integrated volatility or integrated nonlinear functionals of volatility; see \cite{Jacod2013}, \cite{ait2014high}, and the references therein.
In the context of nonstationary time series, this is applied, e.g., by \cite{Dahlhaus2009}, who considers linear functionals of the time-varying spectral density.

A general method to estimate integrated parameters has been suggested by \cite{Potiron2018}. They construct  block-wise estimators $\hat{\theta}_{i,n}$ and average them to obtain an estimator of $\Theta(1)$. 
While this approach could be adapted to the functional estimation of $u\mapsto\Theta(u)$, the verification of their assumptions for the estimators $\hat{\theta}_{i,n}$ entails additional analytical effort for each special case.
In particular, they require strong conditions on the bias of $\hat{\theta}_{i,n}$, and present explicit debiasing procedures for specific examples. 
In contrast, our proposed estimator $M_n(u)$ is based on a linearization procedure around a nonparametric pilot estimator $\hat{\mu}_{t,n}$, and may be regarded as a generic approach for removing leading bias terms.
We suggest a pilot estimator based on local smoothing, but our results are deliberately formulated under much weaker conditions, requiring only assumptions on the rate of convergence of $\hat{\mu}_{t,n}$.
Many modern approaches to filtering and regression are based on statistical learning theory, which typically yields satisfactory rates of convergence, but does not lend itself to statistical inference.
Our linearized estimator $M_n(u)$ thus enables rigorous asymptotic inference based on these pilot estimates. 
Details are presented in Section \ref{sec:estimator}.

Our estimator is particularly useful to test for change-points in the local parameter.
Here, the null hypothesis is that the parameter $\theta_u$ is the same for all $u\in[0,1]$, which may be formulated equivalently as
\begin{align}
	H_0: \theta_u \equiv \theta_0 \quad\Leftrightarrow\quad H_0: \Theta(u) \equiv u \Theta(1). \label{eqn:H0-cp-intro}
\end{align} 
Analysis of this hypothesis of structural stability has a long history in statistics, see \cite{Aue2013} for a recent review.
Early studies were concerned with the stability of the mean \citep{Page1954,Page1955}.
The methodology has since been extended, and there exist procedures to test for, e.g., the stability of variances \citep{Gao2019}, regression coefficients \citep{Horvath1995}, or autocovariances \citep{Berkes2009,Killick2013,Preuss2015}. 
Our approach provides a unifying framework to study these problems for parameters $\theta_u$ which may be written as a function of nonlinear moments.
Besides the mentioned examples, this also includes novel change-point tests which have not been studied previously, e.g.\ a test for the temporal stability of kurtosis.
The proposed change-point tests are robust against various model misspecifications, e.g.\ nonstationarity and nuisance changes under the null hypothesis. 
In particular, our test only monitors changes in the parameter $\theta_u=f(\mu_u)$, but not in the parameter $\mu_u$ itself. 
For example, our statistic is sensitive to changes in the variance, but robust to changes in the unknown and time-varying mean value.
The change-point tests based on our estimator $M_n(u)$ are discussed in greater detail in Section \ref{sec:cp} below.

There are parameters of interest which may not be expressed as a function of finitely many moments, for example quantiles of the marginal distribution, or functionals of the local spectral measure as considered by \cite{Dahlhaus2009}.
A very general framework is presented by \cite{Shao2010}, where an arbitrary functional of the time series' distribution is studied. 
The assumptions therein are formulated in terms of the influence function corresponding to the statistical functional of interest. 
Their verification is far from simple and basically amounts to proving claims very similar to the steps we take in this article.
In contrast, we believe that the conditions imposed in the present paper are conveniently verified for a broad range of practical problems. 
The framework of \cite{Shao2010} is also adopted by \cite{Dette2018b} and applied to the monitoring of quantiles, as well as by \cite{Gosmann2019}.
Note that these authors study the stationary case only, while we allow for nonstationarity.
It might be of interest to extend the methodology introduced in the present paper to more general functionals.
We leave this question for future work.

The outline of this paper is as follows. 
After defining the model in Section \ref{sec:model}, we describe the functional estimator $M_n(u)$ in Section \ref{sec:estimator} and present our asymptotic results., we give the rigorous definition of our model and present the technical assumptions.
The application of our results to change-point problems is discussed in Section \ref{sec:cp}, and the finite sample properties of our proposed procedure are assessed via simulations study in Section \ref{sec:MC}. 
Our methodology is illustrated by an application to financial data in Section \ref{sec:empirical}.
Appendix \crossref{sec:discussion}{A} in the supplement contains additional remarks, and Appendix \crossref{sec:supp-CP}{B} contains further examples and simulation results for change point testing.
All technical proofs are deferred to Appendix \crossref{sec:proofs}{C}.

\subsubsection*{Notation}
For a function $f:\R^d\to\R$, we denote $Df(x) = (\partial_{x_1},\ldots, \partial_{x_d}) f(x) \in \R^{1\times d}$ the first order differential, and by $D^2f(x) \in\R^{d\times d}$ the Hessian matrix.
For two sequences $a_n, b_n$, we write $a_n\ll b_n$ if $a_n/b_n \to 0$ as $n\to\infty$.
For $x\in\R$, we denote $(x)_+ = \max(x, 0)$.
Weak convergence of probability measures and random elements is denoted by $\wconv$.
For a vector $x\in\R^d$, the Euclidean norm is denoted by $\|x\|$, and for a matrix $A\in\R^{d\times d}$, we denote by $\|A\|=\|A\|_{op}$ the operator norm, i.e.\ $\|A\| = \sup_{x\neq 0} \|Ax\|/\|x\|$.
The spectral radius of a matrix $A$ is denoted as $\rho(A)=\max(|\lambda_1|,\ldots |\lambda_d|)$, where the $\lambda_i$ are the complex eigenvalues of $A$.
The transpose of a matrix $A$ is denoted by $A^T$.
For a random vector $X$, we denote by $\|X\|_{L_q} = (\E\|X\|^q)^\frac{1}{q}$ the $L_q$ norm, for $q\geq 1$. 
The notation $\|G_n\|_{p-var}$ is introduced in Section \ref{sec:model}.

\section{Model}\label{sec:model}

Let $X_{t,n}, t=1,\ldots, n$ be a triangular array of random variables which is causal in the sense that $X_{t,n} = G_n (t/n,\beps_t)$ for a sequence of functions $G_n:\R\times\R^\infty\to\R^d$.
Here, we denote
\begin{align*}
	\beps_t = (\epsilon_t, \epsilon_{t-1},\ldots)\in\R^\infty,
\end{align*}
where the $\epsilon_i$ are iid random variables.
 The functions $G_n$ are assumed to be measurable, where we endow the sequence space $\R^\infty$ with the projection $\sigma$-Algebra, see \cite[Example 1.2]{billingsley1999convergence}.

By using a sequence of functions $G_n$, we will be able to apply our results to investigate the power of the proposed tests against local alternatives.
Furthermore, letting the kernel $G_n$ depend on $n$ allows us to account for potential discretization errors, see Example \ref{ex:VAR} below.
We assume that the kernel $G_n$ tends towards a limiting kernel $G:\R\times\R^\infty\to\R^d$, in the sense that 
\begin{align}
	\sup_{u\in[0,1]}\|G_n(u,\beps_0) - G(u,\beps_0)\|_{L_q} \to 0,\qquad n\to\infty, \tag{A.1}\label{eqn:GnG}
\end{align}
for some $q>2$.
Note that we do not require a rate of convergence in \eqref{eqn:GnG}.
Furthermore, we require the function $u\mapsto G_n(u, \beps_0) \in L_q(P)$ to satisfy some regularity conditions. 
A very mild condition can be formulated in terms of $p$-variation for some $p\geq 1$, defined as
\begin{align*}
	\|G_n\|_{p-var} = \sup_{\substack{0=u_0 < u_1 <\ldots < u_m=1\\ m\in\N}} \left( \sum_{i=1}^m \|G_n(u_i, \beps_t) - G_n(u_{i-1}, \beps_t)\|_{L_q}^p  \right)^{\frac{1}{p}}. 
\end{align*}
The latter definition is independent of the chosen $t$, since $\beps_t\sim \beps_0$.
The $p$-variation of the limiting kernel is denoted analogously as $\|G\|_{p-var}$.
We assume that 
\begin{align}
	\sup_n \sup_{u\in[0,1]} \|G_n(u, \beps_0)\|_{L_q} + \sup_n \|G_n\|_{p-var} \leq C_G<\infty, \tag{A.2}\label{eqn:pvar-finite}
\end{align}
and $p\geq 1$ will be further specified if necessary.
Note that $\|G_n\|_{p-var}\leq \|G_n\|_{r-var}$ for $1\leq r \leq p$, hence assumption \eqref{eqn:pvar-finite} is stronger for smaller $p$. 

The assumption of finite $p$-variation is less restrictive than the piecewise-locally-stationary (PLS) framework suggested by \cite{Zhou2013}, which amounts to requiring piecewise Lipschitz continuity with finitely many breakpoints.
The PLS framework has been applied for change-point analysis by \cite{Dette2018a}, \cite{Dette2018}, among others. 
In contrast, finite $p$-variation still allows for infinitely many discontinuities of the mapping $u\mapsto G_n(u,\beps_0)\in L_q(P)$.
On the other hand, if the latter mapping is Hölder continuous with exponent $\beta$, then $\|G_n\|_{p-var}<\infty$ for $p\geq \frac{1}{\beta}$.
Thus, assumption \eqref{eqn:pvar-finite} is more general than requiring Hölder continuity, and it combines classical smoothness conditions as well as discontinuities in a single framework.
The special case of bounded $1$-variation has been considered by \cite{Dahlhaus2009a} for linear processes.
Our framework also contains the model of \cite{Dahlhaus2017} as a special case, see Section \crossref{sec:dahlhaus}{A.1} in the supplement.

The third and last assumption imposed on the causal kernel is uniform ergodicity.
We employ the physical dependence measure introduced by \cite{Wu2005}.
To define the dependence measure, we introduce an iid copy $\epsilon_i^*$ of the $\epsilon_i$
Denote for $t\in\Z, j\in\Z_{\geq 0}$,
\begin{align*}
	\beps^*_{t,j} & = (\epsilon_t, \ldots, \epsilon_{t-j+1}, \epsilon^*_{t-j}, \epsilon^*_{t-j-1},\ldots) \in \R^\infty, \quad\\
	\tilde{\beps}_{t,j} & = (\epsilon_t, \ldots, \epsilon_{t-j+1}, \epsilon^*_{t-j}, \epsilon_{t-j-1},\ldots) \in \R^\infty.
\end{align*}
We assume that there exists a value $\rho\in(0,1)$ such that, for all $j,n\in\N, u\in[0,1]$ 
\begin{align}
	\|G_n(u, \beps_t) - G_n(u,\tilde{\beps}_{t,j})\|_{L_q} \leq C_G \rho^j. \tag{A.3}\label{eqn:ergodic}
\end{align}
This implies that 
\begin{align*}
	\|G_n(u, \beps_t) - G_n(u,\beps^*_{t,j})\|_{L_q} 
	&\leq \frac{C_G}{1-\rho} \rho^j,
\end{align*}
see Proposition \crossref{prop:ergodic}{C.1} in the appendix.

This set of assumptions suffices to establish a functional central limit theorem for the partial sums of the $X_{t,n}$.
An analogous limit theorem has been proven by \cite{Zhou2013} under the more restrictive PLS assumption.

%CLT FOR PARTIAL SUMS
\begin{theorem}\label{thm:CLT-linear}
Let \eqref{eqn:GnG}, \eqref{eqn:pvar-finite}, and \eqref{eqn:ergodic} hold, for some $q>2$.
Suppose furthermore that $d=1$, and denote by 
\begin{align*}
	\sigma^2(u) = \sum_{h=-\infty}^\infty \Cov \left[  G(u, \beps_h), G(u, \beps_0)  \right],
\end{align*}
the local long-run-variance.
Then, as $n\to\infty$,
\begin{align*}
	\frac{1}{\sqrt{n}}\sum_{t=1}^{\lfloor nv\rfloor}\left[ X_{t,n} - \E X_{t,n} \right] \wconv B\left(\int_0^v \sigma^2(u)\, du\right),
\end{align*}
where $B(v),v\geq 0$ is a standard Brownian motion.
The weak convergence holds in the Skorokhod space $D[0,1]$.
\end{theorem}

By virtue of the Cramer-Wold device, Theorem \ref{thm:CLT-linear} can be extended to the multivariate case to yield weak convergence in the space $(D[0,1])^d$ endowed with the product topology.
Note that this topology is different from the Skorokhod topology on the space $D^d[0,1]$ of cadlag functions $g:[0,1]\to\R^d$, see \cite[VI.1.23]{Jacod2003}. 

We conclude this section by giving an example of a process $X_{t,n}$ which satisfies assumptions \eqref{eqn:GnG}-\eqref{eqn:ergodic}, highlighting that the imposed conditions are rather weak.

\begin{example}\label{ex:VAR}
	Consider the time-varying vector autoregressive (tvVAR) process in $d$ dimensions which satisfies
	\begin{align*}
		X_{t,n} &= A(\tfrac{t}{n}) \left[X_{t-1,n}-\mu(\tfrac{t-1}{n})\right] + B(\tfrac{t}{n}) \epsilon_t + \mu(\tfrac{t}{n}),\\
		X_{0,n} &= \sum_{i=0}^\infty A(0)^{i} B(0) \epsilon_{-i} + \mu(0),
	\end{align*}
	where $\epsilon_t$ is a sequence of iid, $d$-dimensional random vectors with finite moments of all orders, $\mu:[0,1]\to\R^d$, and $A,B:[0,1]\to\R^{d\times d}$ are matrix valued functions. 
	Note that higher-order autoregressive processes may also be studied in this framework by stacking the lagged values of the process, effectively increasing the dimension $d$.
	For example, an autoregression of order two may be described in terms of the state vector $Y_{t,n} = (X_{t,n}, X_{t-1,n})$, taking values in $\R^{2d}$.
	Time-varying autoregressive models are classical examples of locally-stationary time series.
	Early investigations of these models include \cite{rao1970} and \cite{grenier1983}, and more recent contributions are due to \cite{moulines2005recursive}, \cite{Dahlhaus2009a}, and \cite{Giraud2015}, among others.
	
	We extend $A,B$ to the domain $(-\infty, 1]$ by setting $A(u)=A(0)$ and $B(u)=B(0)$ for $u<0$.
	The autoregressive process may be cast into our framework $X_{t,n} = G_n(\frac{t}{n},\beps_t)$ as
	\begin{align*}
		G_n(u,\beps_t) = \sum_{i=0}^\infty \left[ \prod_{j=1}^i A\left(u-\tfrac{j-1}{n}\right) \right]  B\left(u-\tfrac{i}{n}\right) \epsilon_{t-i} + \mu(\tfrac{t}{n}).
	\end{align*}
	This infinite sum is well-defined if we suppose that $\sup_{u} \|B(u)\|, \sup_u \|A(u)\|< \infty$, that the spectral radius $\rho(A(u))$ of $A(u), u\in[0,1]$, is at most $\rho(A(u))\leq\rho_0<1$, and that the function $u\mapsto A(u)$ admits some minimal regularity.
	In particular, if the function $u\mapsto A(u)$ has bounded $p$-variation for some $p\geq 1$, then $\|\prod_{j=1}^i A\left(u-\tfrac{j-1}{n}\right)\|\leq C \rho^{i}$ for any $\rho\in(\rho_0,1)$, see Lemma \crossref{lem:eigenvalue-product}{C.8} in the appendix.
	Hence, assumption \eqref{eqn:ergodic} is satisfied.
	Furthermore, if the functions $\mu$, $A$, and $B$ are left-continuous, we have $\|G_n(u,\beps_0)-G(u,\beps_0)\|_{L_q}\to 0$ by dominated convergence for any $q>2$, with limiting kernel
	\begin{align*}
		G(u,\beps_t) = \sum_{i=0}^\infty A(u)^i B(u) \epsilon_{t-i} + \mu(u).
	\end{align*}
	In particular, condition \eqref{eqn:GnG} holds.
	Finally, Proposition \crossref{prop:VAR}{C.9} in the appendix shows that assumption \eqref{eqn:pvar-finite} holds if $\|A\|_{p-var}+\|B\|_{p-var}+\|\mu\|_{p-var}<\infty$. 
	Moreover, in the more regular case that $\mu$, $A$, and $B$ are $\beta$-Hölder continuous, the kernel $G_n(u,\beps_0)$ is also $\beta$-Hölder continuous in $L_q(P)$, uniformly in $n$, and thus the same holds for $u\mapsto \mu^n_u = \E G_n(u,\beps_0)$.
	The latter property is required to apply Proposition \ref{prop:NW} discussed in the following section.
\end{example}

\section{Estimating integrated parameters}\label{sec:estimator}

For the locally stationary model introduced in Section \ref{sec:model}, we denote the local moments by
\begin{align*}
	\mu_u^n = \E G_n(u,\beps_0)\in\R^d,\qquad \mu_u = \lim_{n\to\infty} \mu_u^n=\E G(u,\beps_0) \in \R^d,\qquad u\in[0,1].
\end{align*}
Inference for the function $\mu_u^n$ resp.\ $\theta_u^n = f(\mu_u^n)$ can be performed in various ways. 
For example, one might assume a parametric form for the mapping $u\mapsto \mu_u^n$.
Here, we are interested in testing nonparametric hypotheses imposed on the function $\mu_u$. 
To this end, instead of treating the moment function directly, we suggest to consider its integral.
In particular, for a nonlinear function $f:\R^d\to\R$, we study the integrated quantity \begin{align*}
	F_n(u) = \int_0^u f(\mu^n_v)\, dv, \qquad u\in[0,1].
\end{align*}

Note that most hypotheses on $\theta_u^n=f(\mu^n_u)$ may be reformulated in terms of $F_n(u)$.
In Section \ref{sec:cp}, we will study change-point detection, where the null hypothesis $H_0: f(\mu^n_u) \equiv f(\mu^n_0)$ for $u\in[0,1]$ is equivalent to the hypothesis $H_0: F_n(u)\equiv u F_n(1)$. 
It is also possible to study the hypothesis $H_0: f(\mu^n_u) \equiv 0$, which is equivalent to $H_0: F_n(u)\equiv 0$. 
Another appealing aspect of studying the integrated quantity $F_n(u)$ instead of $f(\mu^n_u)$ is that the former may be estimated at a parametric rate $\sqrt{n}$, even in a nonparametric framework, see our results below. 
Though this might be surprising at first sight, note that a similar phenomenon occurs in the classical statistical setting of iid observations, where the empirical distribution function is $\sqrt{n}$ consistent, while nonparametric density estimation only allows for slower, nonparametric rates of convergences.

A straightforward way to estimate $F_n(u)$ is to consider the partial sum process
\begin{align}
	\tilde{M}_n(u) = \frac{1}{n} \sum_{t=\tau}^{\lfloor un\rfloor} f(\hat{\mu}_{t,n}),  \label{eqn:Mn-naive}
\end{align}
where $\hat{\mu}_{t,n}$ is a nonparametric estimator of $\mu_{\frac{t}{n}}$, and $\tau=\tau_n\geq 1$ may be introduced to alleviate boundary issues.
This approach has two shortcomings.
First, its distributional properties strongly depend on the specific estimator $\hat{\mu}_{t,n}$.
Thus, additional theoretical effort is required when adapting corresponding inferential procedures to new situations.
Second, for nonlinear $f$, the estimator $\tilde{M}_n(u)$ may be bias-dominated, impeding statistical inference.

In particular, a Taylor expansion yields
\begin{align}
	\frac{1}{n} \sum_{t=\tau}^{\lfloor un\rfloor} \left[f(\hat{\mu}_{t,n}) - f(\mu_{\frac{t}{n}})\right]
	&= \frac{1}{n} \sum_{t=\tau}^{\lfloor un\rfloor} Df(\mu_\frac{t}{n})(\hat{\mu}_{t,n} - \mu_\frac{t}{n}) + \mathcal{O}\left(\frac{1}{n} \sum_{t=\tau}^n \|\hat{\mu}_{t,n} - \mu_\frac{t}{n}\|^2\right). \label{eqn:Mn-naive-bias}
\end{align}
For most estimators $\hat{\mu}_{t,n}$, the bias of this expression is not smaller than $1/\sqrt{n}$.
If one considers, for example, a local average with bandwidth $k_n$, and assumes $u\mapsto \mu_u$ to be Lipschitz continuous, then the bias of $\tilde{M}_n(u)$ is of order $\mathcal{O}(\frac{k}{n} + \frac{1}{k})$, which is at best $\mathcal{O}(1/\sqrt{n})$.
In a related situation, \cite{Jacod2013} suggest to solve this problem by undersmoothing, i.e.\ choosing $k_n\ll \sqrt{n}$, and correcting the quadratic bias term in \eqref{eqn:Mn-naive-bias} explicitly; see also \cite[Section 4.2]{Potiron2018}.
However, the latter quadratic bias term strongly depends on the specific estimator $\hat{\mu}_{t,n}$, so that the approach of \cite{Jacod2013} may not be easily transfered to different problems.

As a generic approach for asymptotically unbiased estimation of integrated functionals, we propose the linearized partial sum estimator 
\begin{align}
	M_n(u) = \frac{1}{n} \sum_{t=\tau+L}^{\lfloor nu\rfloor} \left[ f(\hat{\mu}_{t-L,n}) + D f(\hat{\mu}_{t-L,n}) (X_{t,n}-\hat{\mu}_{t-L, n}) \right], \quad u\in[0,1]. \label{eqn:Mn-def}
\end{align}
for some initial offset $\tau=\tau_n\geq 1$, lag $L = L_n\to\infty$, $L_n\ll \tau_n$, to be specified later, and assuming $f$ to be sufficiently smooth.
The pilot estimator $\hat{\mu}_{t,n}$ needs to satisfy minimal high-level assumptions formulated below. 
Then, a Taylor expansion of $f$ readily yields, for some $\tilde{\mu}_{t,n}$ between $\hat{\mu}_{t-L,n}$ and $\mu^n_{t/n}$,
\begin{align*}
	M_n(u) 
	&= \frac{1}{n} \sum_{t=\tau_n+L_n}^{\lfloor nu \rfloor} f\left( \mu_{\frac{t}{n}}^n \right) 
	       + \frac{1}{n} \sum_{t=\tau_n+L_n}^{\lfloor nu \rfloor} \left( \mu^n_{\frac{t}{n}} - \hat{\mu}_{t-L,n} \right)^T\frac{D^2f\left(\tilde{\mu}_{t,n}\right)}{2} \left( \mu^n_{\frac{t}{n}} - \hat{\mu}_{t-L,n} \right)\\
	&\quad + \frac{1}{n} \sum_{t=\tau_n+L_n}^{\lfloor nu \rfloor} Df(\hat{\mu}_{t-L,n}) \left( X_{t,n}-\mu^n_{\frac{t}{n}} \right) \\
	&= I_n^1(u) + I_n^2(u) + I_n^3(u).
\end{align*} 
It can be shown that $I_n^1(u)-F_n(u)$ converges to zero sufficiently fast. 
Moreover, if $f$ is sufficiently regular, and the estimator $\hat{\mu}_{t,n}$ is good enough, then $\sqrt{n}I_n^2(u)\to 0$, whereas $\sqrt{n}I_n^3(u)$ is asymptotically unbiased and tends towards a Gaussian process.

We require the function $f$ to have bounded first and second derivatives in a neighborhood of the path of $\mu_v$.
Formally, introduce the convex set $\mathcal{M} = \text{conv} \{ \mu_u : u\in[0,1] \}\subset\R^d$, and its $\delta$-neighborhood $\mathcal{M}^\delta = \{ m \in\R^d: \|m-\tilde{m}\|<\delta\text{ for some } \tilde{m}\in\mathcal{M} \}$.
We require that there exists a $\delta\in(0,\infty]$ and a $C_f < \infty$ such that 
\begin{align}
	|f(x)| + \|D f(x)\| + \| D^2 f(x)\| \leq C_f, \quad x\in\mathcal{M}^\delta. \tag{A.4}\label{eqn:ass-f}
\end{align}
By restricting the boundedness assumption to the set $\mathcal{M}^\delta$, we may also consider functions of the form $f(x,y)= x/y$, if $y\geq c > 0$ on the set $\mathcal{M}^\delta$. 
Without localization, \eqref{eqn:ass-f} would not hold for this choice of $f$.

Regarding the pilot estimator, we require that $\hat{\mu}_{t,n}$ is measurable w.r.t.\ the past innovations $\beps_{t}$, i.e.\ $\hat{\mu}_{t,n}$ is a (potentially nonlinear) filter. 
By additionally introducing the lag $L$, the measurability condition on $\hat{\mu}_{t,n}$ serves to de-bias the term $I_n^3(u)$. 
In particular, as $L_n\to\infty$, the random vectors $Df(\hat{\mu}_{t-L,n})$ and $X_{t,n}$ decouple by virtue of \eqref{eqn:ergodic}.
Furthermore, we require the estimator $\hat{\mu}_{t,n}$ to be consistent in the sense that 
\begin{align}
	\sum_{t=\tau_n}^n \|\hat{\mu}_{t,n} - \mu^n_{\frac{t}{n}}\|^2  &= o_P(\sqrt{n}), \tag{A.5}\label{eqn:ass-mu-1}\\
	P\left( \hat{\mu}_{t,n}\in\mathcal{M}_\delta,\; t=\tau_n,\ldots, n\right) &= 1+o(1) \tag{A.6}\label{eqn:ass-mu-2}.
\end{align} 
The initial offset $\tau_n$ allows to circumvent boundary issues of the estimator $\hat{\mu}_{t,n}$.
The assumptions \eqref{eqn:ass-mu-1} and \eqref{eqn:ass-mu-2} are rather mild.
Property \eqref{eqn:ass-mu-2} requires some weak form of uniform consistency of the estimator.
If $f$ is globally smooth, i.e.\ $\delta=\infty$, then \eqref{eqn:ass-mu-2} is vacuous, and $\tau_n=1$ is a valid choice. 
Property \eqref{eqn:ass-mu-1} is a requirement on the rate of convergence of $\hat{\mu}_{t,n}$,  in a form routinely studied in nonparametric statistics.
The latter assumption is discussed in detail in Section \crossref{sec:rates}{A.2} of the supplement.
If we are willing to impose some additional smoothness conditions, a suitable nonparametric estimator may be obtained by local averaging.
In particular, we define
\begin{align*}
	\hat{\mu}_{t,n}^{NW} = \frac{1}{k_n\wedge t} \sum_{i=(t-k_n)\vee 1}^{t} X_{i,n},
\end{align*}
for a sequence $k_n\to\infty, k_n\ll n$. 
Note that $\hat{\mu}^{NW}_{t,n}$ can be interpreted as a one-sided kernel smoother of Nadaraya-Watson type.

\begin{proposition}\label{prop:NW}
	Let \eqref{eqn:ergodic} hold for some $q>2$, and choose $k_n\gg n^\frac{2}{q}$, then $\hat{\mu}_{t,n}^{NW}$ satisfies \eqref{eqn:ass-mu-2} for any offset sequence $\tau_n\to\infty$.
	Suppose that for all $n$, we have $\mu^n_u = \mu^{n,1}_{u} + \mu^{n,2}_{u}$, such that $\|\mu^{n,2}_{u}\|_{p-var}\leq a_n$ for $p\in[1,2)$, and such that $u\mapsto \mu^{n,1}_{u}$ is $\beta$-Hölder continuous with Hölder constant $C$, then
	\begin{align*}
		\frac{1}{n} \sum_{t=1}^n \left\|\hat{\mu}_{t,n}^{NW} - \mu^n_{\frac{t}{n}}\right\|^2 = \mathcal{O}_P \left( \frac{a_n k_n}{n}+ \left( \frac{k_n}{n} \right)^{2\beta} + \frac{\log n}{k_n} \right). 
	\end{align*} 
	In particular, $\hat{\mu}^{NW}_{t,n}$ satisfies \eqref{eqn:ass-mu-1} if $a_n k_n \ll\sqrt{n}$, and $\sqrt{n}\log(n) \ll k_n \ll n^{\frac{4\beta-1}{4\beta}}$.
	The latter condition is feasible for any $\beta>\frac{1}{2}$.
\end{proposition}

By allowing the vanishing discontinuous part $\mu^{n,2}$ in Proposition \ref{prop:NW}, we may account for potential discretization errors.

For the local smoother $\hat{\mu}_{t,n}^{NW}$, as well as for any other estimator $\hat{\mu}_{t,n}$ satisfying our assumptions \eqref{eqn:ass-mu-1} and \eqref{eqn:ass-mu-2}, the functional estimator $M_n(u)$ admits a central limit theorem with parametric rate $\sqrt{n}$.
Our main result may be formulated as follows.

\begin{theorem}\label{thm:CLT-integrated}
Suppose that \eqref{eqn:GnG}, \eqref{eqn:pvar-finite}, \eqref{eqn:ergodic}, \eqref{eqn:ass-f}, \eqref{eqn:ass-mu-1}, and \eqref{eqn:ass-mu-2} hold, and let $p\in[1,4)$.
Suppose that $\tau_n\ll n$, and that $L_n$ satisfies $L_n \gg \log(n)^{1+a}$ for some $a\in(0,1)$, and that $L_n \ll n^{\frac{1}{p}-\frac{1}{4}}$, $L_n\ll n^\frac{1}{2p}$. 
Then, as $n\to\infty$,
\begin{align*}
	\sqrt{n}\left(M_n(u) - \frac{1}{n}\sum_{t=\tau_n+L_n}^{\lfloor un\rfloor} f(\mu^n_\frac{t}{n}) \right) &\wconv B\left(\int_0^u Df(\mu_v) \Sigma(v) Df(\mu_v)^T\,dv\right) =M(u), \\
	\Sigma(v) &= \sum_{h=-\infty}^\infty \Cov \left(G(v, \beps_0), G(v, \beps_h)\right) \in\R^{d\times d},
\end{align*}
where $B$ denotes a standard Brownian motion.
The weak convergence holds in the Skorokhod space $D[0,1]$.
If $p\in[1,2)$, it also holds that
\begin{align*}
	\sqrt{n}\left(M_n(u) - \int_{\frac{\tau_n+L_n}{n}\wedge u}^u f(\mu^n_v)\, dv \right) &\wconv M(u).
\end{align*}

\end{theorem}

Just as in Theorem \ref{thm:CLT-linear}, the Cramer-Wold device may be used to extend Theorem \ref{thm:CLT-integrated} to the multivariate setting where $f$ takes values in $\R^{d^*}$. 
The functional weak convergence then holds in the product space $(D[0,1])^{d^*}$.
If $p<2$ and $\tau_n\ll n^\frac{1}{2}$, we may alternatively choose $\int_0^u f(\mu_v^n)\, dv$ as centering term, such that $\sqrt{n}[M_n(u) - \int_0^u f(\mu_v^n)\, dv)]$ has the asymptotic distribution given in Theorem \ref{thm:CLT-integrated}.

\begin{remark}

	The suitable choice of the lag parameter $L$ depends on the strength of the dependency of the time series. 
	On the one hand, the bias term $I_n^2$ grows polynomially with $L$, i.e.\ $I_n^2 = \mathcal{O}_P(L_n^{\min(p,2)} n^{-2/\max(p,2)})$, see Lemma \crossref{lem:I2-L}{C.5} in the supplement.
	This suggests to choose the lag as small as possible. 
	On the other hand, $L$ needs to be large enough such that $X_{t,n}$ and $\hat{\mu}_{t-L,n}$ decouple, in order for $I_n^3$ to be asymptotically unbiased.
	In particular, in view of assumption \ref{eqn:ergodic}, we need to ensure that $\rho^{L_n} \ll 1/\sqrt{n}$.
	To ensure that both, $I_n^2$ and the bias of $I_n^3$, are negligible, a conservative choice is $L_n = c \log(n)^2$ for some factor $c$.
	The sensitivity of our methodology with respect to $L$ is assessed by simulations in Section \ref{sec:MC}.

\end{remark}

\begin{remark}\label{rem:Ln}
	In contrast to Theorem \ref{thm:CLT-linear}, the central limit theorem of the linearized estimator $M_n(u)$ requires $p<4$, i.e.\ the kernel $G_n$ needs to be more regular.
	This restriction is due to the bias incurred by the lag $L$.
	In particular, we can only ensure that $I_n^2 = o_P(1/\sqrt{n})$ if $p<4$, see Lemma \crossref{lem:I2-L}{C.5} in the supplement.
	On the other hand, the criticality of $p=2$ occurs because $|\frac{1}{n}\sum_{t=1}^n f(\mu^n_{\frac{t}{n}}) - \int_0^1 f(\mu^n_u)\, du| = \mathcal{O}(n^{-\frac{1}{p}} \|\mu^n\|_{p-var})$, see Lemma \crossref{lem:I1}{C.4} in the supplement.
	Note that both issues are not present in the linear case of Theorem \ref{thm:CLT-linear}.
	 
\end{remark}

The integrand of the asymptotic variance of the limit process $M(u)$ corresponds to the long-run variance under the local, stationary model $G(u,\beps_t)$. 
In particular, a direct application of the delta method shows that 
\begin{align}
	\sqrt{n}\left[f(\textstyle{\frac{1}{n} \sum_{t=1}^n G(u, \beps_t))} - f(\mu_u)\right] \wconv \mathcal{N}\left(0, Df(\mu_u)^T \Sigma(u) Df(\mu_u) \right). \label{eqn:CLT-stationary} 
\end{align}
If our model is indeed stationary, Theorem \ref{thm:CLT-integrated} shows that $M_n(1)$ has the same asymptotic distribution as \eqref{eqn:CLT-stationary}.
In this sense, accounting for the nonstationarity does not increase the asymptotic variance.

To perform feasible inference based on Theorem \ref{thm:CLT-integrated}, we need to handle the unknown asymptotic variance process. 
A consistent estimator may be constructed via blocked subsampling, similar to the suggestion of \cite{Carlstein1986}.

\begin{theorem}\label{thm:Q-estimation}
	Let the conditions of Theorem \ref{thm:CLT-integrated} hold for some $q>4$, and $p\in[1,4)$. 
	Choose some $b_n\to\infty$ such that $b_n\ll n^{\frac{2}{3 \max(p,2)}}$, $b_n \ll n^\frac{q-4}{2q+4}$.
	Then, as $n\to\infty$,
	\begin{align*}
	Q_n(u) 
	&= \frac{1}{n} \sum_{t=\tau_n+L_n}^{\lfloor n u \rfloor-b_n} \frac{1}{b_n} \left[ Df(\hat{\mu}_{t-L,n})\sum_{i=1}^{b_n}  (X_{t+i,n} - \hat{\mu}_{t-L,n})  \right]^2\\
	&\pconv Q(u)=\int_0^u Df(\mu_v) \Sigma(v) Df(\mu_v)^T\, dv.
	\end{align*}
	The convergence holds uniformly in $u\in[0,1]$ since $Q_n$ is monotone.
\end{theorem}

Theorem \ref{thm:Q-estimation} is a special case of the slightly more general Theorem \crossref{thm:bootstrap-new}{C.7} in the appendix.
Note that the upper bound on $b_n$ reduces to $b_n\ll n^\frac{1}{3}$ if $q\geq 10$. 

The estimator $Q_n(u)$ may be used to perform inference based on $M_n(u)$ via the following multiplier bootstrap scheme, similar to \cite{Zhou2013}. 

\begin{theorem}\label{thm:bootstrap}
	Let the conditions of Theorem \ref{thm:Q-estimation} hold for some $q>4$.
	Let $Y_t \sim\mathcal{N}(0,1)$ be iid standard normal random variables, independent of the $\epsilon_i$, and define the process
	\begin{align*}
	\widehat{M}_n(u) = \frac{1}{\sqrt{n}} \sum_{t=\tau_n+L_n}^{\lfloor n u \rfloor-b_n} Y_t \left[  \frac{1}{\sqrt{b_n}}\sum_{i=1}^{b_n} Df(\hat{\mu}_{t-L,n})(X_{t+i,n} - \hat{\mu}_{t-L,n})  \right],\qquad u\in[0,1].
	\end{align*}
	Then the conditional distribution of $\widehat{M}_n$ given $\mathbb{X}_n = (X_{1,n},\ldots, X_{n,n})$ converges weakly in the Skorokhod space to $M(u)$ in probability, where $M(u)$ is the limit process from Theorem \ref{thm:CLT-integrated}.
\end{theorem}

It can also be shown that the bootstrap consistency of Theorem \ref{thm:bootstrap} holds under weaker rate constraints on $\hat{\mu}_{t,n}$, replacing the rate $o(\sqrt{n})$ by $o(n^{1-\kappa})$ for some $\kappa>0$.
However, a smaller value of $\kappa$ requires stronger conditions on $b_n$, see Theorem \crossref{thm:bootstrap-new}{C.7} in the appendix.
The local smoother $\hat{\mu}_{t,n}^{NW}$ is still consistent in the non-smooth case, where only $\|\mu\|_{p-var}<\infty$, although at a slower rate.
Hence, the latter estimator may still be utilized for consistent variance estimation.
This is of particular interest for applications to change-point tests, as described in the following section, where the bootstrap procedure is still consistent under various alternative hypotheses.

\section{Change-point detection}\label{sec:cp}

A major motivation to perform inference for the integrated parameter $F_n(u)$  resp.\ $F(u) = \int_0^u f(\mu_v)\, dv$ is that the estimator $M_n(u)$ may be used to test for change-points.
Our framework lends itself to test the hypothesis
\begin{align}
	H_0: f(\mu^n_{u})=f(\mu^n_{0}) \text{ for all $u\in[0,1]$} \quad \leftrightarrow \quad H_1: f(\mu^n_u) \neq f(\mu^n_{0})\text{ for some $u\in[0,1]$.} \label{eqn:H0-cp}
\end{align}
To perform a test for this problem, a common approach is to formulate the CUSUM statistic, which in our case reads as 
\begin{align}
	\begin{split}
	T^*_n = \sup_{u\in [u_n,1]}  |T_n(u)|, &\quad \text{where }
	T_n(u)=  M_n(u) - \frac{u-u_n}{1-u_n} M_n(1)\\
	&\quad \text{and }\qquad u_n = \frac{\tau_n+L_n-1}{n}. 
	\end{split}\label{eqn:cusum}
\end{align}

The main result Theorem \ref{thm:CLT-integrated} yields that $T_n(u) \pconv F(u)-uF(1)$ as $n\to\infty$, which is identically zero if the null hypothesis holds.
In this case, Theorem \ref{thm:CLT-integrated} yields that
\begin{align*}
	\sqrt{n} T_n(u) &\wconv  M(u) - u M(1),\\
	\sqrt{n}T^*_n &\wconv T^* = \sup_{u\in[0,1]} |M(u) - uM(1)|,
\end{align*}
where $M(u)$ is the Gaussian limit process.
The limit distribution $T^*$ may be approximated via the bootstrap procedure outlined in Theorem \ref{thm:bootstrap}, i.e.\ by sampling the random variable $\widehat{T}^*_n = \sup_{u\in[0,1]} |\widehat{M}_n(u)-u\widehat{M}_n(1)|$, so that $(\widehat{T}^*_n|\mathbb{X}_n) \wconv T^*$ by virtue of Theorem \ref{thm:bootstrap}.
In particular, denote by $t_\alpha$ the $1-\alpha$ quantile of $T^*$, and by $t_{\alpha,n}$ the $1-\alpha$ quantile of $(\widehat{T}^*_n|\mathbb{X}_n)$.
In practice, the quantile $t_{\alpha,n}$ may be approximated up to arbitrary precision by sampling from the conditional distribution $\hat{T}_n^*$. 
The corresponding test procedure may then be formulated as follows.

\begin{proposition}\label{prop:CUSUM}
	Let the conditions of Theorem \ref{thm:bootstrap} hold, and denote by $t_{\alpha,n}$ the $1-\alpha$ quantile of the conditional distribution $\widehat{T}^*_n|\mathbb{X}_n$.
	If the null hypothesis \eqref{eqn:H0-cp} holds, and if $\Var(M(u_0))>0$ for some $u_0\in(0,1)$, then
	\begin{align*}
		\lim_{n\to\infty} P\left( T_n^* > t_{\alpha,n} \right) = \alpha, \quad n\to\infty.
	\end{align*}
	Hence, rejecting $H_0$ if the test statistic $T_n^*$ exceeds the critical value $t_{\alpha,n}$ leads to a test with nominal size $\alpha\in(0,1)$ asymptotically.
\end{proposition}

Although we focus on the uniform CUSUM test statistic, the functional central limit theorem for the process $M_n(u)$ also enables the consideration of alternative statistics, e.g.\ the MOSUM statistic introduced by \cite{Bauer1978}, see also \cite{Chu1995}, or the Cramér-von Mises statistic $\int_0^1 T_n(u)^2\, du$.

A desirable property of change-point tests is robustness against nuisance changes.
For the quantity $f(\mu^n_u)$ to be non-constant, it is necessary that the local moment $\mu^n_u$ changes.
It is thus tempting to instead test the null hypothesis $H_0^*: \mu^n_u\equiv \mu^n_0$, which is methodologically simpler to achieve.
For example, the methods of \cite{Zhou2013} and \cite{Vogt2015} are applicable to test for $H_0^*$.
However, this approach bears the risk to falsely detect a change although $f(\mu^n_u)$ remains constant. 
For example, it might happen that the variance of a time series is non-constant, while the autocorrelation structure remains constant, as studied by \cite{Dette2018}.
Furthermore, \cite{Schmidt2020} tests for homoscedasticity with a non-constant mean function.
A related approach is presented by \cite{Demetrescu2018}. 
By design, our test is only sensitive to changes in the quantity $f(\mu^n_u)$.

Another type of nuisance change might occur in the parameters which are not explicitly described by the local moment function $\mu^n_u$.
For example, when testing for changes in the mean of a heteroscedastic time series, the variance is a nuisance parameter not contained in the vector $\mu^n_u$, but relevant for statistical inference, see \cite{Gorecki2018} and \cite{Pesta2018}. 
We account for this type of nonstationarity by working in a locally stationary framework which allows not only for heteroscedasticity, but also for a varying dependency structure.
This has also been suggested by \cite{Zhou2013}, who designs a corresponding test for changes in the mean.
The recent articles \cite{Vogt2015}, \cite{Dette2018}, and \cite{Cui2020}, also employ a locally stationary model.

Test statistics for the change point problem usually need to be standardized by an estimator of their asymptotic variance.
However, variance estimators designed for the stationary case might be inconsistent under the alternative, resulting in a loss of power, see \citep{Juhl2009,Shao2010} and the discussion therein.
In contrast, our bootstrap procedure is consistent under the alternative where $f(\mu^n_u)$ is not constant and potentially discontinuous, see Theorem \ref{thm:bootstrap} and the discussion thereafter.
Moreover, we may investigate the behavior of our test statistic under local alternatives at rate $1/\sqrt{n}$.

\begin{proposition}\label{prop:power}
	Suppose that the conditions of Theorem \ref{thm:CLT-integrated} hold, and assume furthermore that $\mu^n_u = \mu_u + \frac{1}{\sqrt{n}} \delta_u$ for some function $u\mapsto \delta_u\in\R^d$, such that $f(\mu_u) = f(\mu_0)$ and $\|\delta\|_{p-var}<\infty$.
	Then 
	\begin{align*}
		\sqrt{n} T^*_n &\wconv \sup_{u\in[0,1]} | T(u) + \Delta(u) |,\\
		\Delta(u) &= \int_0^u Df(\mu_v) \delta_v\, dv - u \int_0^1 Df(\mu_v)\delta_v\, dv,
	\end{align*} 
	where $T(u) = M(u) - u M(1)$, and $M(u)$ is the limit process from Theorem \ref{thm:CLT-integrated}.
	
\end{proposition}

Note that under the local alternative of Proposition \ref{prop:power}, the simple estimator $\hat\mu_{t,n}^{NW}$ still satisfies \eqref{eqn:ass-mu-1} and \eqref{eqn:ass-mu-2} if $u\mapsto \mu_u$ is smooth, see Proposition \ref{prop:NW}.
Hence, Theorem \ref{thm:bootstrap} is still applicable so that the bootstrap is consistent, and the CUSUM test with bootstrapped critical values has non-trivial power against local alternatives in direction $\delta_v$, given that $\Delta(u)\not\equiv 0$.
We also point out that the test has power not only against abrupt changes, but also against changes which occur gradually in time.

The proposed procedure allows for a unified treatment of change-point tests for a wide range of parameters of interest, as demonstrated by the examples below. 
Previously, suitable test statistics have been constructed individually for these problems, while our results show that they may be treated in a rather generic way.

\subsection{Changes in autocorrelation}\label{sec:autocorrelation}

The dependency structure of time series is commonly described in terms of their autocovariance function.
It is thus natural to test the latter for structural stability, as suggested by \cite{Berkes2009}.
They construct a CUSUM test based on the partial sums which form the empirical autocovariance estimator at fixed lag, and derive limit theorems under the assumption of stationarity.
A method to detect changes without fixing the lag is mentioned by \cite[Example 3]{Steland2019}.
The case of a nonparametric mean function is investigated by \cite{Li2013}, and multiple change-points are studied by \cite{Preuss2015}.
The non-stationary case is investigated by \cite{Killick2013}, although without a rigorous analysis of the type I error.

Alternatively, in the same univariate setting, \cite{Dette2018} test whether the autocorrelation $\Cor(X_{t,n}, X_{t-h,n})$ remains constant, for some fixed $h>0$.
This problem is more involved, since it requires standardization by the marginal variances, which are an additional nuisance quantity.
They allow the marginal variance to be non-constant and estimate it non-parametrically, in order to standardize the observations.
Furthermore, \cite{Dette2018} study a nonlinear, locally stationary specification of the underlying time series.
Hence, they account for potential nonstationarity under the null hypothesis.

We may formulate the problem to test for constant autocorrelations in our general framework.
To this end, we set $Y_{t,n} = G_n(\frac{t}{n},\beps_t) = (X_{t,n}, X_{t-h,n}, X_{t,n}^2, X_{t-h,n}^2, X_{t,n} X_{t-h,n})$ for $t=1,\ldots, n$, assuming for simplicity that $X_{r,n} = \tilde{G}_n(0,\beps_r)$ for $r\leq 0$.
Now set $f:\R^5\to\R, x \mapsto (x_5 - x_1x_2)/\sqrt{ (x_3-x_1^2)(x_4-x_2^2)}$, so that $\Cor(X_{t,n}, X_{t-h,n}) = f(\E Y_{t,n}) = f(\mu^n_{t/n})$.
Again, assumptions \eqref{eqn:GnG}-\eqref{eqn:ergodic} are a direct consequence of the corresponding properties of $X_{t,n}$.
The function $f$ is bounded on any compact set $K\subset \{ x\in\R^5 : (x_3-x_1^2)>0, (x_4-x_2^2)>0  \}$.
Thus, if $\Var(X_{t,n}) >2\delta>0$, then \eqref{eqn:ass-f} holds for $\delta>0$ as well.
We may thus construct $M_n(u)$ based on this time series $Y_{t,n}$ and function $f$ to obtain an estimator for the integrated autocorrelation $F_n(u)=\int_0^u \Cor(G_n(v, \beps_0), G_n(v, \beps_{-h}))\, dv$.
The corresponding CUSUM statistic satisfies \eqref{eqn:cusum}.

The resulting CUSUM statistic is similar to the statistic suggested by \cite{Dette2018}.
However, our framwork allows for many potential choices of $\hat{\mu}_{t,n}$, while \cite{Dette2018} only consider a special case.
Moreover, our assumptions regarding the regularity of $G_n$ are weaker.

\subsection{Further examples}\label{sec:kurtosis}

The kurtosis of a random variable $X$ is defined as $\Kurt(X) = \E( X-\E X)^4 / \Var(X)^2$. 
For a univariate time series $X_{t,n}$, let $Y_{t,n} = (X_t, X_t^2, X_t^3, X_t^4)$. 
Then $\Kurt(X_{t,n})$ can be written as a function of $\E(Y_{t,n})$. 
In particular, 
\begin{align*}
	\text{Kurt}(X_{t,n}) = f(\E Y_{t,n}) = \frac{\E X_{t,n}^4 - 4\E X_{t,n} \E X_{t,n}^3 + 6(\E X_{t,n}^2)^2 - 3(\E X_{t,n})^4}{\E X_{t,n}^2 - (\E X_{t,n})^2}
\end{align*}
If $\Var(X_{t,n})>c>0$, then $f$ satisfies \eqref{eqn:ass-f} so that our results are applicable, and the CUSUM statistic \eqref{eqn:cusum} in combination with the bootstrap procedure yields a feasible change-point test.
To the best of our knowledge, the proposed method is the first test for structural stability of the marginal kurtosis.

In a similar way, we may consider the skewness $\text{Skew}(X) = \E (X-\E X)^3 / \Var(X)^\frac{3}{2}$ of a random variable provided that the variance of $X$ is bounded away from zero, or the coefficient of variation $\text{CV}(X) = \sqrt{\Var(X)} / \E(X)$ if the expectation of $X$ is bounded away from zero.
A further example which may be cast in our framework are time-varying autoregressive models, as presented in example \ref{ex:VAR}, where the coefficients may be identified in terms of finitely many autocovariances by means of the Yule-Walker equations.
Our methodology could thus be used to test for changes in the second autoregressive component of a univariate tvAR model.

In the supplement, we also discuss change point tests for the marginal variance (Section \crossref{sec:variance}{B.1}) and for the coefficients of a linear regression model (Section \crossref{sec:regression}{B.2}).

\section{Finite sample performance}\label{sec:MC}

To assess the finite sample performance of our proposed change-point test, we evaluate its size and power properties via simulations.
We consider the locally stationary autoregressive process $X_{t,n}$ given by 
\begin{align}
	X_{t,n} = a(\tfrac{t}{n}) X_{t-1,n} + \sigma(\tfrac{t}{n}) \eta_t(\tfrac{t}{n}). \label{eqn:X-mc}
\end{align}
The innovations $\eta_t(\frac{t}{n})$ are chosen as independent, zero-mean random variables having a symmetrized Gamma distribution with shape parameter $\alpha(\frac{t}{n})$, standardized to unit variance. 
We use
\begin{align*}
	\sigma(u) = 0.5 + |\sin(2\pi u)|, \qquad \alpha(u) = \begin{cases}
		1,& u\leq 0.7,\\ 2,& u> 0.7,
	\end{cases}
\end{align*}
and for $a(u)$, either of the three functions 
\begin{align*}
	a_0(u) = 0.2,\quad a_1(u) = 0.2+\tfrac{u}{2},\quad a_2(u) = 0.2+\tfrac{u}{10}.
\end{align*}
We want to test for stability of the lag-1 autocorrelation, which is equivalent to the stability of the autoregressive coefficient $a$.
To this end, we apply the change-point test presented in Section \ref{sec:autocorrelation} , in combination with the local estimator $\hat{\mu}_{t,n}^{NW}$.

As described in Remark \ref{rem:Ln}, the lag parameter $L$ should be chosen just big enough such that the functional dependence measure at lag $L$ is negligible.
We choose $L_n = c \log(n)^2$ and analyze the effect of the factor $c$ below.
Once $L_n$ is specified, the smoothing bandwidth $k=k_n$ may be determined via cross-validation, by minimizing the prediction error
\begin{align*}
	\Lambda(k)=\sum_{t=1}^{n-L} \|\hat{\mu}^{NW,k}_{t,n} - X_{t+L,n}\|^2,
\end{align*}
where $\hat{\mu}^{NW,k}_{t,n}$ denotes the local average with bandwidth $k$.
In our simulations, we consider bandwidths from the interval $k\in [n^{0.35}, n^{0.75}]$.
Then, a natural choice for the offset is $\tau_n=k_n$.
Finally, we choose the window size for the bootstrap procedure as $b_n=L_n$.
Thus, $L_n$ is the only parameter to be chosen manually.
While a data-driven choice of the lag parameter $L_n$ is desirable, deriving a corresponding method is out of scope of this article. 

To find the critical value of the CUSUM test, for each individual sample, we resample $M=10^3$ independent realizations based on the bootstrap approximation of Theorem \ref{thm:bootstrap}. 
Equivalently, we may use the bootstrap samples to compute an approximate p-value for the CUSUM test statistic. 
When assessing the size of the CUSUM test, we set $a=a_0$, and for the power analysis, we set $a=a_1$ respectively $a=a_2$.

\begin{table}[bt]
	
	\centering
	\footnotesize
	\begin{tabular}{r|ccc|ccc|ccc|ccc|ccc|}
		   & \multicolumn{3}{c|}{$L_n = \lceil\log(n)^2\rceil$} & \multicolumn{3}{c|}{$L_n = \lceil\tfrac{1}{2}\log(n)^2\rceil$} & \multicolumn{3}{c|}{$L_n =\lceil\tfrac{1}{5} \log(n)^2\rceil$} & \multicolumn{3}{c|}{$L_n = \lceil\tfrac{1}{10}\log(n)^2\rceil$} & \multicolumn{3}{c|}{DWZ} \\ 
		 & $H_0$& $H_1$& $H_2$& $H_0$& $H_1$& $H_2$& $H_0$& $H_1$& $H_2$& $H_0$& $H_1$& $H_2$& $H_0$& $H_1$& $H_2$\\ \midrule \midrule  
		$n=100$ & 0.26 & 0.19 & 0.28 & 0.04 & 0.03 & 0.05 & 0.04 & 0.03 & 0.04 & 0.05 & 0.04 & 0.04 & 0.50 & 0.54 & 0.49 \\ 
		500 & 0.02 & 0.03 & 0.02 & 0.02 & 0.13 & 0.03 & 0.04 & 0.37 & 0.05 & 0.05 & 0.49 & 0.07 & 0.18 & 0.67 & 0.22 \\ 
		1000 & 0.01 & 0.18 & 0.01 & 0.02 & 0.55 & 0.04 & 0.04 & 0.81 & 0.09 & 0.06 & 0.89 & 0.11 & 0.13 & 0.91 & 0.21 \\ 
		5000 & 0.02 & 1.00 & 0.20 & 0.04 & 1.00 & 0.33 & 0.07 & 1.00 & 0.42 & 0.08 & 1.00 & 0.46 & 0.07 & 1.00 & 0.47 \\ 
		10000 & 0.03 & 1.00 & 0.56 & 0.06 & 1.00 & 0.68 & 0.08 & 1.00 & 0.74 & 0.09 & 1.00 & 0.76 & 0.07 & 1.00 & 0.78 \\   \bottomrule
	\end{tabular}
	\caption{Size and power of the bootstrap-based CUSUM test for constant autocorrelation, with nominal level $10\%$. Reported values are based on $5000$ independent samples of the test statistic.}
	\label{tab:MC-ar}
\end{table}

\begin{figure}[tb]
	\centering
	\includegraphics[width =  \textwidth]{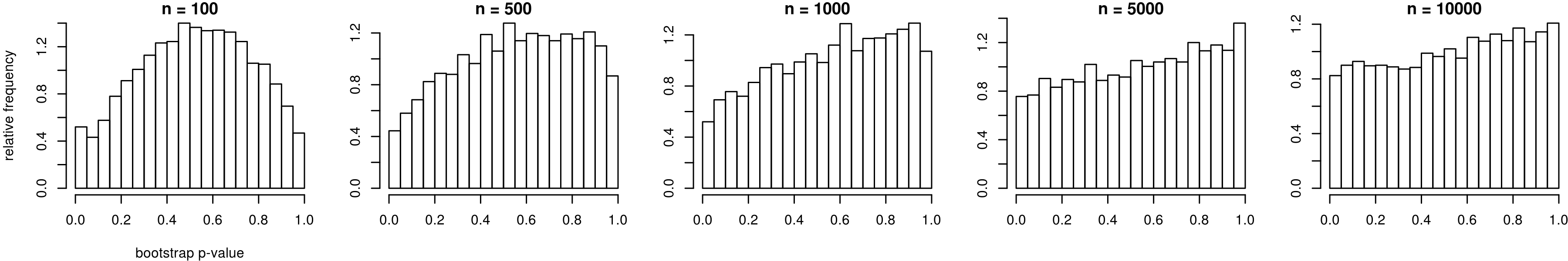}
	\caption{Distribution of bootstrap-based p-values for the CUSUM test for constant autocorrelation, under the null hypothesis, with lag parameter $L_n=\frac{1}{10}\log(n)^2$. The p-values are computed based on $M=10^3$ bootstrap samples, and the histograms are based on $10^4$ independent samples of the test statistic.}
	\label{fig:MC-ar}
\end{figure}

Table \ref{tab:MC-ar} presents the size and power of the proposed test in the presented example, for various values of $n$ and different choices $L_n$. 
The power values correspond to the alternatives $H_1$ based on $a_1$, and $H_2$ based on $a_2$.
We find that our test is rather conservative, i.e.\ the type-I error is actually smaller than the nominal level, and this conservativeness vanishes asymptotically as $n$ increases.
In particular, our test does not falsely detect a structural break even though the nuisance parameter is non-constant.
On the other hand, the method consistently detects deviations from the null hypothesis, as demonstrated by the increasing power against the alternative.
Moreover, it is found that the smallest lag value $L_n = \lceil\frac{1}{10} \log(n)^2\rceil$ yields the best size approximation.
Note that the differences in power for various choices of $L_n$ may be partially explained by the different test sizes. 
For the latter choice of $L_n$, we also depict the distribution of the simulated p-values in Figure \ref{fig:MC-ar}. 
The p-values should ideally be uniformly distributed.
Indeed, for large sample size $n$, the accuracy of the p-values increases, in line with our theoretical results. 

For comparison, we also implement the change point test for constant lag-1 autocorrelation proposed by Dette, Wu, and Zhou \citep{Dette2018}.
The corresponding size and power are also presented in Table \ref{tab:MC-ar}, labeled as DWZ.
In small samples, the latter test achieves a higher power, which may be explained by a correspondingly higher rate of false positives.
For large sample sizes, the power is similar to our proposed test, showing that our broadly applicable method is competitive against the specialized test of \cite{Dette2018}.

\begin{table}[bt]
	\centering
	
	\small
	\begin{tabular}{cc|cc|cc|cc|cc|cc|}
		&& \multicolumn{2}{c|}{$n = 100$} & \multicolumn{2}{c|}{$n = 500$} & \multicolumn{2}{c|}{$n = 1000$} & \multicolumn{2}{c|}{$n = 5000$} & \multicolumn{2}{c|}{$n = 10000$} \\ 
		&& $M_n$ & $\tilde{M}_n$ & $M_n$ & $\tilde{M}_n$ & $M_n$ & $\tilde{M}_n$ & $M_n$ & $\tilde{M}_n$ & $M_n$ & $\tilde{M}_n$ \\ \midrule \midrule  
		\multirow{2}{*}{$H_0$} &error & 0.172 & 0.120 & 0.053 & 0.058 & 0.031 & 0.039 & 0.012 & 0.018 & 0.008 & 0.013 \\ 
		&bias  & 0.071 & -0.114 & 0.022 & -0.054 & 0.011 & -0.036 & -0.001 & -0.016 & -0.002 & -0.012 \\ \midrule
		\multirow{2}{*}{$H_1$} &error & 0.242 & 0.197 & 0.061 & 0.093 & 0.035 & 0.066 & 0.011 & 0.031 & 0.008 & 0.023 \\ 
		&bias  & 0.163 & -0.197 & 0.039 & -0.093 & 0.020 & -0.066 & 0.002 & -0.031 & -0.001 & -0.023 \\ \midrule
		\multirow{2}{*}{$H_2$} &error & 0.180 & 0.134 & 0.053 & 0.062 & 0.032 & 0.044 & 0.012 & 0.020 & 0.008 & 0.014 \\ 
		&bias  & 0.084 & -0.131 & 0.025 & -0.060 & 0.012 & -0.043 & -0.000 & -0.019 & -0.002 & -0.014 \\ \bottomrule						
	\end{tabular}
	\caption{Mean absolute error and bias of $M_n(1)$ and $\tilde{M}_n(1)$ as estimators of $\int_0^u a_i(u)\, du$, $i=0,1,2$, with $L_n = \frac{1}{10}\log(n)^2$ for the linearized estimator. All values are based on $5000$ simulations.}
	\label{tab:MC-ar-est}
\end{table}

We also assess the quality of $M_n(1)$ as an estimator of $\int_0^1 a_i(u)\, du$, $i=0,1,2$, in comparison to the plug-in estimator $\tilde{M}_n(1)$. 
Table \ref{tab:MC-ar-est} presents the mean absolute errors of both estimators, as well as their corresponding bias.
Except for the smallest sample size $n=100$, the proposed linearized estimator performs better than the simple plug-in estimator.
In particular, the linearization greatly decreases the bias of the estimator. 

We also assess the finite sample performance of a change point test for the coefficients of a linear regression model. 
The simulation results are presented in Section \crossref{sec:sim-ols}{B.3} of the supplement.

\section{Empirical illustration}\label{sec:empirical}

To demonstrate the use of our results in practice, we study an application to high-frequency financial data.
In particular, we study the price $p_t$ of the german mid-cap stock-index MDAX on April 4, 2016, from 9:00-15:30, at a sampling frequency of $1$ second.
The data is available as a free sample from the data shop of Deutsche Börse, and part of the supplementary material of this article.
We study the log-returns $d_t = \log(p_{t})-\log(p_{t-1})$, $t=1,\ldots, n$, with sample size $n=23400$.

While many models for asset prices imply uncorrelated returns, empirical research suggests that autocorrelation may be non-zero, especially at high sampling frequencies \citep{Hansen2006}.
This dependence structure is typically attributed to the microstructure of the market, e.g.\ rounding effects or bid-ask rebounds. 
Recently, \cite{Andersen2016} studied the intraday returns of the NASDAQ100 stocks and found evidence for autocorrelation which is not only non-zero, but also non-constant.

Using the framework laid out in Section \ref{sec:autocorrelation} above, we may rigorously perform asymptotic inference for the local autocorrelation $\Cor(d_t, d_{t-1})$.
To this end, we use the local estimator $\hat{\mu}_{t,n}^{NW}$ and choose its bandwidth via cross-validation as in Section \ref{sec:MC}, with $L_n = \lceil \log(n)^2/10\rceil$.

\begin{figure}[tb]
	\centering
	\includegraphics[width=0.8\textwidth]{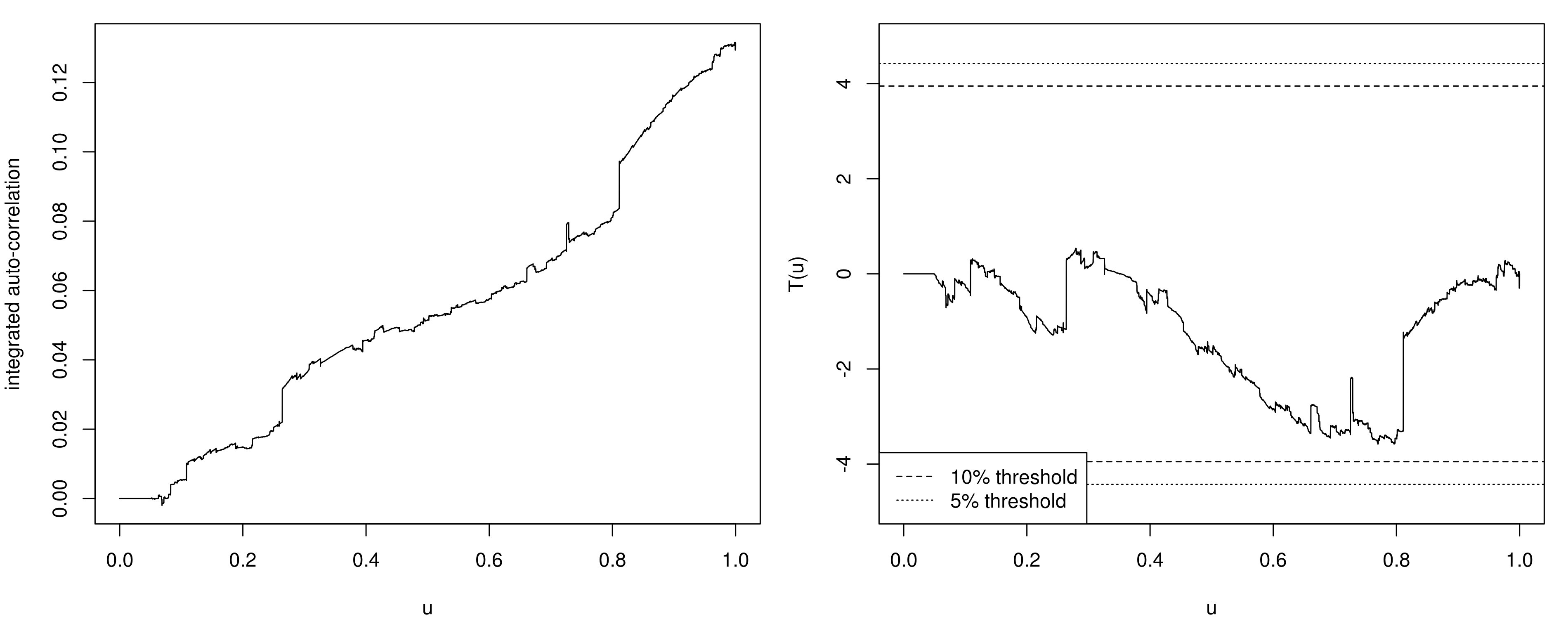}
	\caption{Left: the estimator $M_n(u)$ of the integrated autocorrelation. Right: the corresponding CUSUM process $\bar{T}_n(u)$. The critical thresholds are based on the bootstrap approximation with $10^4$ bootstrap samples.}
	\label{fig:MDAX-autocor}
\end{figure}

The functional estimator $M_n(u)$ of the lag-1 autocorrelation is depicted in Figure \ref{fig:MDAX-autocor} (left).
First, we observe that $M_n(u)$ is roughly increasing, which indicates that the lag-1 autocorrelation is positive on average.
Indeed, the average lag-1 autocorrelation is estimated as $M_n(1) = 0.1314$, with asymptotic standard deviation $\sqrt{Q_n(1)} = 0.0209$.
Moreover, visual inspection of $M_n(u)$ suggests that the slope is varying, which corresponds to a non-constant autocorrelation.
To test this hypothesis rigorously, we perform the CUSUM test suggested in Section \ref{sec:cp}.
The right panel of Figure \ref{fig:MDAX-autocor} shows the CUSUM process $T_n(u)$, and the critical thresholds for a significance level of $10\%$ and $5\%$, respectively. 
The critical values are obtained using the bootstrap approximation of Theorem \ref{thm:bootstrap}, with $10^4$ bootstrap samples. 
The bootstrap-based p-value of the CUSUM test statistic is $0.167$, based on $10^4$ bootstrap samples.
Note that the sample size $n$ is rather large, and the simulation results of Section \ref{sec:MC} suggest that the bootstrap approximation is satisfactory for this regime.
Hence, we find that the variation of the lag-1 autocorrelation is not significant.
That is, based on the bootstrapped CUSUM test, we may not reject the null hypothesis of constant lag-1 autocorrelations for this particular dataset at a significance level of $5\%$.

\begin{figure}[tb]
	\centering
	\includegraphics[width=0.8\textwidth]{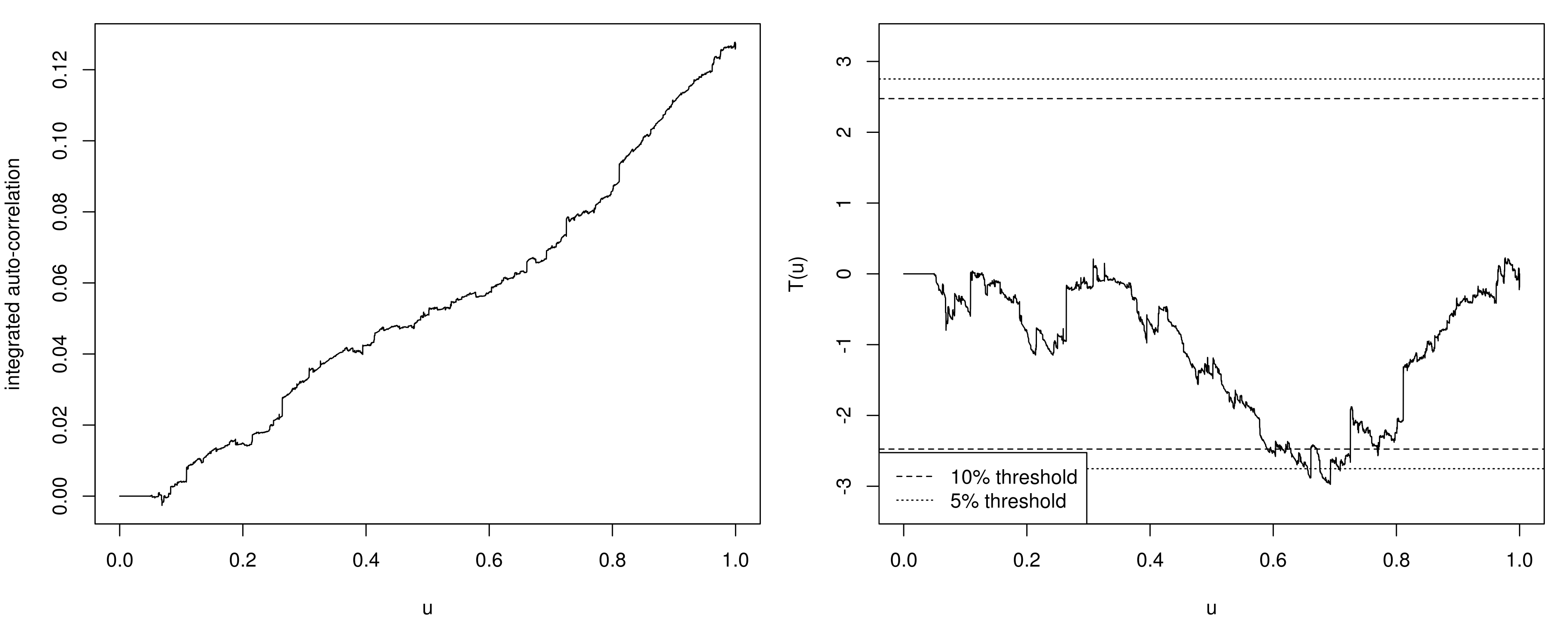}
	\caption{Left: the estimator $M_n(u)$ of the integrated autocorrelation fo the transformed log returns $\tilde{d}_t$. Right: the corresponding CUSUM process $\bar{T}_n(u)$. The critical thresholds are based on the bootstrap approximation with $10^4$ bootstrap samples.}
	\label{fig:MDAX-autocor-bounded}
\end{figure}

The visible discontinuities of the path of $M_n(u)$ in Figure \ref{fig:MDAX-autocor}, suggest that the estimator is influenced by few very large price changes $d_t$.
While our bootstrap procedure automatically accounts for this, the resulting large variance decreases the power of the change point test.
To reduce the effect of the heavy tails of the log returns, we repeat our analysis for the transformed increments $\tilde{d}_t = \arctan(d_t /\gamma )$.
We choose $\gamma=10^{-4}$, which corresponds to the average size of $d_t$ and leads to a unimodal distribution of transformed returns (not depicted).
The estimator $M_n(u)$ of the lag-1 autocorrelation of the transformed returns $\tilde{d}_t$ is depicted in Figure \ref{fig:MDAX-autocor-bounded}, as well as the corresponding CUSUM process.
The p-value of the CUSUM test is $0.027$, hence the hypothesis of constant lag-1 autocorrelation is rejected for the series $\tilde{d}_t$ at a significance level of $5\%$.
Although the autocorrelations of $d_t$ and $\tilde{d}_t$ are not the same parameters, both may be interpreted similarly in the present application.
In particular, our findings support the claim of \cite{Andersen2016} that the serial correlation of intraday log returns is non-constant for the present data set.
In contrast to \cite{Andersen2016}, our change point test does not assess whether the autocorrelation changes its sign.

We also analyze the variance, mean, and kurtosis of $d_t$ and $\tilde{d}_t$ as outlined in Sections \crossref{sec:variance}{B.1} and \ref{sec:kurtosis}.
For $d_t$, the CUSUM tests for constant variance, mean, and kurtosis yield the p-values $0.081$, $0.020$, and $0.285$, respectively.
For $\tilde{d}_t$, the respective p-values are $0.000$, $0.0484$, and $0.000$, respectively.
In combination with our statistical results on the autocorrelation, this provides strong evidence for nonstationarity of $\tilde{d}_t$ and $d_t$.

Many models for asset returns at very high frequencies describe the observed price as the sum of two latent components: the fundamental price $p_t^*$, and the so-called microstructure noise $\eta_t$, such that $p_t = p_t^*+\eta_t$.
The fundamental price is typically modeled as a semimartingale, and inference for this component needs to account for the microstructure effects, see e.g.\ \cite{jacod2009preaveraging}.
The microstructure noise is typically assumed to be independent, or dependent but stationary \citep{Hansen2006, ait2011ultra}.
Nonstationary dependent noise is considered by \cite{jacod2017statistical}, but such that the autocorrelation of the microstructure is constant.
The approach  we pursue in the present paper does not distinguish between the fundamental price and the microstructure effects.
Nevertheless, our empirical findings may motivate the investigation of microstructure models which allow for a nonstationary dependence structure.

\newpage
\appendix
\section{Further remarks}\label{sec:discussion}

\subsection{Alternative definition of local stationarity}\label{sec:dahlhaus}
	The model described in Section \ref{sec:model} may also be compared to the definition of locally-stationary processes introduced by \cite{Dahlhaus2017}.
	They require that for each $u\in[0,1]$, there exists a stationary process $X_t(u)$ such that 
	\begin{enumerate}[(i)]
		\item $\|X_{t,n} - X_t(\frac{t}{n})\|_{L_q} \leq C n^{-\alpha}$, and
		\item $\|X_t(u) - X_t(v)\|_{L_q}\leq C|u-v|^{\alpha}$, for some $C>0$ and $\alpha\in(0,1]$. 
	\end{enumerate}
	If the time series $X_{t,n}$ and $X_t(u)$, $u\in[0,1]$, are $\beps_t$-measurable, we may represent them as $X_{t}(u) = G(u,\beps_t)$, and $X_{t,n} = G_{n}(\frac{t}{n},\beps_t)$, for kernels $G$ and $G_n$ which are measurable w.r.t.\ $\beps_t$.
	Without loss of generality, we may suppose that $u\mapsto G_n(u,\beps)$ is a left-continuous, piecewise constant mapping with finitely many break points at $t/n$, for $t=1,\ldots, n$. 
	Then \begin{align*}
		\|G_n\|_{p-var} &= \left(\sum_{t=2}^n \|G_n(\tfrac{t}{n},\beps_t) - G_n(\tfrac{t-1}{n}, \beps_0)\|^p_{L_q}\right)^\frac{1}{p} \\
		&\leq 2\left(\sum_{t=1}^{n} \|G_n(\tfrac{t}{n},\beps_0) - G(\tfrac{t}{n},\beps_0)\|_{L_q}^p\right)^\frac{1}{p} + \|G\|_{p-var}.
	\end{align*}
	Conditions (i) and (ii) imply that $\|G_n\|_{p-var}$ is bounded for any $p\geq\frac{1}{\alpha}$.
	Hence, provided that the time series is $\beps_t$ measurable, our framework contains the model of \cite{Dahlhaus2017} as a special case.
	For instance, we may describe temporally varying auto-regressive processes, as described in Example \ref{ex:VAR}.

\subsection{Rates of convergence of the pilot estimator $\hat{\mu}_t^n$}\label{sec:rates}

Property \eqref{eqn:ass-mu-1} is a requirement on the rate of convergence of $\hat{\mu}_{t,n}$,  in a form routinely studied in nonparametric statistics, see e.g.\ \cite[Ch.\ 9]{VandeGeer2010}.
The achievable rate of convergence in nonparametric regression depends in particular on the smoothness of the function $u\mapsto\mu^n_u$. 
In our model, the only regularity assumption on $\mu^n_u$ is \eqref{eqn:pvar-finite}, which implies that $\|\mu^n\|_{p-var}< \infty$. 
Thus, empirical process theory suggests an achievable rate of convergence of $\mathcal{O}(n^\frac{p}{2+p})$ for $p\in[1,2)$.
This rate can be derived by combining general results for least squares regression with subgaussian errors \citep[Thm.\ 9.1]{VandeGeer2010} with entropy bounds for the class of functions of bounded p-variation \citep[Cor.\ 3.7.50]{Gine2016}.
Note that this rate is in line with the requirement \eqref{eqn:ass-mu-1}  if $p\in[1,2)$. 
Furthermore, we point out that this rate of convergence matches the rate under the assumption of Hölder continuity with exponent $\beta = \frac{1}{p}$. 
While the Hölder-continuous case may be treated via local smoothing, the generic approach to achieve this rate of convergence under the assumption of finite $p$-variation is via empirical risk minimization, which in particular uses the whole sample.  
However, we additionally require that $\hat{\mu}_{t,n}$ is $\beps_{t}$-measurable, rendering this approach infeasible.
While it might be possible to construct a suitable online estimator which achieves the desired rate of convergence for functions of bounded p-variation, this is out of the scope of this article.
See for example \cite{baby2019online} and \cite{raj2020non} for the case $p=1$ based on iid observations. 

Another alternative is to formulate a parametric estimator of $\mu^n_u$, which will typically satisfy \eqref{eqn:ass-mu-1} with the stronger rate $\sum_{t=1}^n \|\hat{\mu}_{t,n} - \mu^n_{t/n}\|^2=\mathcal{O}_P(1)$.  
For example, this approach is feasible if $\mu^n_u$ is piecewise constant with a single breakpoint.
However, in order to apply a parametric estimator, additional assumptions need to be imposed on the underlying time series.
If these parametric assumptions do not hold, any inference based on the corresponding asymptotic results will be flawed. 
This encourages the use of nonparametric estimators for the local moment function $\mu^n_u$.
	
\section{More examples of change point problems}\label{sec:supp-CP}

\subsection{Changes in variance}\label{sec:variance}

For a univariate time series $X_{t,n}$, various researches have designed tests for constancy of the variances $\Var(X_{t,n})$, starting with the investigation of asset returns by \cite{Hsu1974} and the corresponding methodology of \cite{Wichern1976}.
A test based on cumulative sums of squared observations has been suggested by \cite{Inclan1994}, and \cite{Chen1997} suggest an alternative procedure based on the Schwarz information criterion.
The approach of \cite{Inclan1994} has been generalized by \cite{Lee2001} to dependent processes, and by \cite{Aue2009} to the multivariate case.
A common shortfall of many procedures is that they require constancy of the mean $\E (X_{t,n})$, which is a nuisance parameter in the present situation. 
More recent work studies the case of non-constant mean by using a suitable nonparametric estimator thereof, see \cite{Gao2019} and \cite{Schmidt2020}.
Note that the latter references still require stationarity of all remaining nuisance quantities, such as autocorrelations and higher order moments.

An alternative test for homoscedasticity can be formulated using our results.
Based on the univariate time series $X_{t,n} = \tilde{G}_n(\frac{t}{n},\beps_t)$, we define the vector valued time series $Y_{t,n}=G_n(\frac{t}{n}, \beps_t) = (X_{t,n}, X_{t,n}^2)$ with mean $\mu^n_u = \E G_n(u, \beps_0)$.
Then $\Var(X_{t,n}) = \E X_{t,n}^2 - (\E X_{t,n})^2 = f(\mu^n_{t/n})$, for $f:\R^2\to \R, f(a,b) = b-a^2$.
It is straight-forward to check that if $\tilde{G}_n$ satisfies assumptions \eqref{eqn:GnG}-\eqref{eqn:ergodic} for $q>8$, then $G_n$ satisfies the assumptions for $\frac{q}{2}>4$.
Moreover, $f$ and all its derivatives are bounded on compacts, such that \eqref{eqn:ass-f} holds if $\mu_u =\lim_{n\to\infty} \mu^n_u $ is bounded. 
Thus, the asymptotic result \eqref{eqn:cusum} holds and yields a test for the hypothesis $H_0: \Var(X_{t,n}) \equiv \text{const}$, allowing for non-stationary mean under very mild regularity conditions on dependency structure of $X_{t,n}$.

\subsection{Changes in regression coefficients}\label{sec:regression}

Suppose we observe a univariate time series $Z_{t,n}$ and a $d$-dimensional vector time series $X_{t,n}$, such that $Z_{t,n} = \beta_{t,n}^T X_{t,n} + \eta_{t,n}$. 
The noise values $\eta_{t,n}$ satisfy $\E(\eta_{t,n})=0$ and $\Cov(X_{t,n}, \eta_{t,n})=0$.
We want to test whether $\beta_{t,n}$ is constant, or varies with $t$. 
This problem has been studied, among others, by \cite{Horvath1995}, \cite{Horvath2004}, and \cite{Aue2008} for uncorrelated innovations $\eta_{t,n}$.
Correlated regression errors are treated, for example, by \cite{Robbins2016}.
In this setting, the hypothesis $H_0: \beta_{t,n} \equiv \beta_{1,n}$ may be formulated in the form \eqref{eqn:H0-cp}, since 
\begin{align*}
	\beta_{t,n} = \Cov(X_{t,n})^{-1} \Cov(X_{t,n}, Z_{t,n}),
\end{align*}
where $\Cov(X_{t,n})\in\R^{d\times d}$ and $\Cov(X_{t,n}, Z_{t,n}) \in\R^{d\times 1}$.
Thus, $\beta_{t,n} = f(\E Y_{t,n})$, for 
\begin{align*}
	Y_{t,n} = (X_{t,n}, Z_{t,n}, X_{t,n} X_{t,n}^T, X_{t,n} Z_{t,n}) \in \R^{2d+1+d^2},
\end{align*}
where the matrix $X_{t,n} X_{t,n}^T$ is interpreted as a $d^2$-dimensional vector.
Assumptions \eqref{eqn:GnG}-\eqref{eqn:ergodic} are a direct consequence of the corresponding properties of $X_{t,n}$ and $\eta_{t,n}$.
Moreover, $f$ satisfies \eqref{eqn:ass-f} if we ensure that the smallest eigenvalue of $\Cov(X_{t,n})$ admits a uniform lower bound.

Note that $\beta_{t,n}$ is a $d$-dimensional parameter. 
It is straight-forward to extend the CUSUM statistic $T_n^*$ to the multivariate setting, e.g.\ by replacing the absolute value by an arbitrary vector norm.
Alternatively, one might test for changes in a single coordinate of $\beta_{t,n}$.

\subsection{Further simulation results}\label{sec:sim-ols}

To further assess the finite sample performance of our proposed procedure, we study the regression model
\begin{align*}
	Z_{t,n} &= \beta(\tfrac{t}{n})^T W_{t,n} + X_{t,n},
\end{align*}
for regression coefficients $\beta(u)\in\R^2$, and noise process $X_{t,n}$ as in equation \eqref{eqn:X-mc} of the article, using $\sigma(u)$ and $\alpha(u)$ as specified, and autoregression coefficient $a(u) = a_1(u)$.
The cofactors $W_{t,n} \sim \mathcal{N}(0, \Sigma(\frac{t}{n}))$ are independent, bivariate normal random vectors, independent of $X_{t,n}$, with covariance matrix
\begin{align*}
	\Sigma(u) = A(u)^T A(u),\quad \text{for}\quad
	A(u) = \begin{pmatrix}
		1 & 2+|\sin(2\pi u)| \\ 0 & 1
	\end{pmatrix}^2.
\end{align*}
Here, we want to test for structural stability of the first regression coefficient $\beta(u)_1$, and we apply the CUSUM test as outlined in Section \ref{sec:regression}.
For our simulations, we employ the model with regression coefficients $\beta(u) = \beta^0(u) \equiv (1,2)^T$ as null hypothesis, and $\beta(u) = \beta^1(u) = (1+u, 2+u^2)^T$, $\beta^2(u) = (1+\frac{u}{3}, 2+\frac{u^2}{3})$ as alternatives, i.e.\ a gradual change.

\begin{figure}[tb]
	\includegraphics[width = \textwidth]{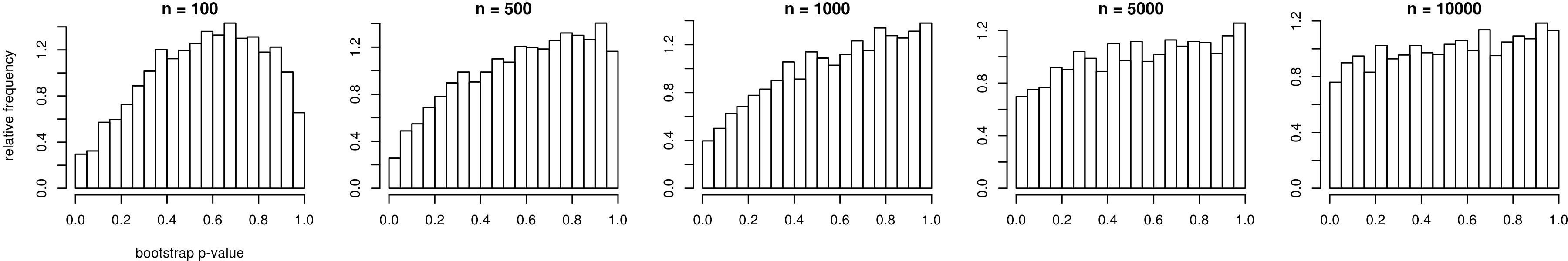}
	\caption{Distribution of bootstrap-based p-values for the CUSUM test for a constant regression coefficient, under the null hypothesis, with lag parameter $L_n=\frac{1}{10}\log(n)^2$. The p-values are computed based on $M=10^3$ bootstrap samples, and the histograms are based on $5000$ independent samples of the test statistic.}
	\label{fig:MC-ols}
\end{figure}

\begin{table}[bt]
	\centering
	\footnotesize
	\begin{tabular}{r|ccc|ccc|ccc|ccc|}
		& \multicolumn{3}{c|}{$L_n = \lceil\log(n)^2\rceil$} & \multicolumn{3}{c|}{$L_n = \lceil\tfrac{1}{2}\log(n)^2\rceil$} & \multicolumn{3}{c|}{$L_n =\lceil\tfrac{1}{5} \log(n)^2\rceil$} & \multicolumn{3}{c|}{$L_n = \lceil\tfrac{1}{10}\log(n)^2\rceil$}  \\ 
		& $H_0$& $H_1$& $H_2$ & $H_0$& $H_1$& $H_2$ & $H_0$& $H_1$& $H_2$ & $H_0$& $H_1$ & $H_2$ \\ \midrule \midrule  
		$n=100$ & 0.16 & 0.14 & 0.15 & 0.08 & 0.06 & 0.08 & 0.03 & 0.03 & 0.03 & 0.03 & 0.03 & 0.03 \\ 
		500     & 0.01 & 0.02 & 0.01 & 0.01 & 0.06 & 0.01 & 0.02 & 0.16 & 0.03 & 0.04 & 0.21 & 0.04 \\ 
		1000    & 0.01 & 0.10 & 0.01 & 0.02 & 0.28 & 0.04 & 0.04 & 0.41 & 0.07 & 0.04 & 0.45 & 0.09 \\ 
		5000    & 0.03 & 0.96 & 0.20 & 0.05 & 0.99 & 0.28 & 0.06 & 0.99 & 0.33 & 0.07 & 0.99 & 0.35 \\ 
		10000   & 0.05 & 1.00 & 0.46 & 0.06 & 1.00 & 0.53 & 0.07 & 1.00 & 0.58 & 0.08 & 1.00 & 0.58 \\   \bottomrule
	\end{tabular}
	\caption{Size and power of the bootstrap-based CUSUM test for a constant regression coefficient, with nominal level $10\%$. Reported values are based on $5000$ independent samples of the test statistic.}
	\label{tab:MC-ols}
\end{table}

The distribution of the bootstrap-based p-values under null hypothesis is depicted in Figure \ref{fig:MC-ols}.
Values of the test's size and power are presented in Table \ref{tab:MC-ols}, for different choices of the lag parameter $L_n$.
Just as for the autocorrelation, we find that our proposed test for the regression coefficient is conservative, and that the size approximation is better for smaller $L_n$.
As the sample size increases, the distribution of the p-values approaches the desired uniform distribution, and the size of the test tends towards the nominal level.
This demonstrates the robustness of our test against non-stationary nuisance parameters.
Furthermore, the CUSUM test consistently detects the nonstationarity of the regression coefficient as the sample size increases.

\section{Technical proofs}\label{sec:proofs}

\subsection{Equivalence of the physical dependence measure}
\begin{proposition}\label{prop:ergodic}
	Let $\beps_t, \beps^*_{t,j}, \tilde{\beps}_{t,j}$, be as in Section \ref{sec:model}.
	Let $G:\R^\infty\to\R^d$ be a measurable function, where we endow $\R^\infty$ with the $\sigma$-Algebra generated by all finite projections.
	If $\|G(\beps_0)\|_{L_q}<\infty$ for some $q\geq 1$, then
	\begin{align*}
		\left\| G(\beps_0) - G(\beps_{0,j}^*) \right\|_{L_q} 
		&\leq \sum_{s=0}^\infty \|G(\beps_0) - G(\tilde{\beps}_{0,j+s})\|_{L_q}.
	\end{align*} 
\end{proposition}
\begin{proof}[Proof of Proposition \ref{prop:ergodic}]
	Consider the martingale $G_t = \E(G(\beps_0)| \epsilon_{0},\ldots, \epsilon_{-t+1})$.
	Since $\|G(\beps_0)\|_{L_q}<\infty$, the martingale convergence theorem guarantees that that $\|G_t - G(\beps_0)\|_{L_q}\to 0$ as $t\to\infty$.
	Now note that $G_t-G(\beps_0) \deq G_t - G(\beps_{0,t}^*)$, such that we also obtain $\|G_t-G(\beps_{0,t}^*)\|_{L_q}\to 0$ as $t\to\infty$ and thus $\|G(\beps_{0,t}^*) - G(\beps_0)\|_{L_q}\to 0$ as $t\to\infty$.
	Hence, for any $t\in\N$, and any $j\in\N$,
	\begin{align*}
		\| G(\beps_0) - G(\beps_{0,j}^*) \|_{L_q} 
		&\leq \|G(\beps_0) - G(\beps_{0,j+t}^*)\|_{L_q} + \sum_{s=1}^t \|G(\beps_{0,j+s}^*) - G(\beps_{0,j+s-1}^*)\|_{L_q}.
	\end{align*}
	But $\|G(\beps_{0,j+s}^*) - G(\beps_{0,j+s-1}^*)\|_{L_q} = \|G(\beps_0) - G(\tilde{\beps}_{0,j+s-1})\|_{L_q}$.
	Thus, letting $t\to\infty$, we find that
	\begin{align*}
		\| G(\beps_0) - G(\beps_{0,j}^*) \|_{L_q} \leq \sum_{s=0}^\infty \|G(\beps_0) - G(\tilde{\beps}_{0,j+s})\|_{L_q}.
	\end{align*}
\end{proof}

Applying Proposition \eqref{prop:ergodic} under assumption \eqref{eqn:ergodic}, we find that
\begin{align*}
	\| G(\beps_t) - G(\beps_{t,j}^*) \|_{L_q} \leq \sum_{s=j}^\infty C_G \rho^j \leq \frac{C_G}{1-\rho} \rho^j. 
\end{align*}

\subsection{Asymptotics of partial sums}
\begin{proof}[Proof of Theorem \ref{thm:CLT-linear}]
We exploit the geometric decay of the dependence measure \eqref{eqn:ergodic} to employ a coupling construction.
To this end, we split the $n$ observations into $m=m_n\ll n$ blocks of length $\tilde{L}=\tilde{L}_n=\lfloor n/m\rfloor$ and let $m\to\infty$. 
Furthermore, let $r=r_n\ll \tilde{L}_n$ with $r_n\to\infty$ be a smaller block length.
It turns out that the rates $\tilde{L}_n = \lceil n^\epsilon\rceil, r_n = \lceil n^{\epsilon/2}\rceil$ are appropriate, for some $0 < \epsilon < \min(\frac{1}{2}-\frac{1}{q}, \frac{1}{p} )$.
We decompose
\begin{align}
	S_n(u)
	&= \frac{1}{\sqrt{n}}\sum_{t=1}^{\lfloor nu\rfloor } X_{t,n}-\E X_{t,n} \nonumber\\
	&= \frac{1}{\sqrt{n}} \sum_{j=1}^{\lfloor um \rfloor} \sum_{t=(j-1)\tilde{L}+r}^{j\tilde{L}} \left[ X_{t,n} - \E X_{t,n} \right] 
	+ \frac{1}{\sqrt{n}} \sum_{j=1}^{\lfloor um\rfloor} \sum_{t=(j-1)\tilde{L}+1}^{(j-1)\tilde{L}+r-1} \left[X_{t,n} - \E X_{t,n} \right] \nonumber\\
	&\quad + \frac{1}{\sqrt{n}} \sum_{t=\tilde{L}\lfloor um\rfloor+1 }^{\lfloor nu\rfloor}\left[ X_{t,n} - \E X_{t,n} \right] \nonumber\\
	&= \frac{1}{\sqrt{n}} \sum_{j=1}^{\lfloor um \rfloor} Y_{j,n}  
		+ \frac{1}{\sqrt{n}} \sum_{j=1}^{\lfloor um \rfloor} \tilde{Y}_{j,n} + R_n(u). \label{eqn:Sn-1}
\end{align}
Consider the remainder term $R_n(u)$. 
By assumptions \eqref{eqn:GnG} and \eqref{eqn:pvar-finite}, we have $\|X_{t,n}\|_{L_q}\leq 2C_G$ for all $n,t$.
The union bound and Markov's inequality yield, for any $a>0$,
\begin{align*}
	P\left(\sup_{u\in[0,1]} |R_n(u)| > a\right) 
	&\leq \sum_{k=\tilde{L}+1}^n  P\left( \frac{1}{\sqrt{n}} \sum_{t=k-\tilde{L}}^k |X_{t,n}-\E X_{t,n}| > a \right) \\
	&\leq \frac{n}{\sqrt{n}^{q}} \frac{(2C_G \tilde{L})^q}{a^q}.
\end{align*}
Since $q>2$, this term tends to zero for our choice of $\tilde{L}$, since $\tilde{L}=\tilde{L}_n = n^\epsilon$ for $\epsilon<\frac{1}{2} - \frac{1}{q}$.

In \eqref{eqn:Sn-1}, the random variables $Y_{j,n}$, $j=1,\ldots, m$, may be replaced by independent copies $Y^*_{j,n}$ as follows. 
For each $j=1,\ldots, m$, let $\tilde{\epsilon}_t^j$ be an independent copy of the $\epsilon_t, t\in\Z$.
Then 
\begin{align*}
	Y_{j,n} 
	&= \sum_{t=(j-1)\tilde{L}+r}^{j\tilde{L}} \left[G_n(\tfrac{t}{n}, \epsilon_t,\epsilon_{t-1},\ldots) - \E X_{t,n}\right] \\
	&= \sum_{t=(j-1)\tilde{L}+r}^{j\tilde{L}} \left[G_n(\tfrac{t}{n}, \epsilon_t,\ldots,\epsilon_{(j-1)\tilde{L}+1}, \tilde{\epsilon}^j_{(j-1)\tilde{L}},\tilde{\epsilon}^j_{(j-1)\tilde{L}-1}\ldots) - \E X_{t,n}\right] + \delta_{j,n} \\
	&= Y^*_{j,n} + \delta_{j,n},
\end{align*}
for a random variable $\delta_{j,n}$ with $\|\delta_{j,n}\|_{L_q} \leq C_G \tilde{L}\sum_{k=r}^\infty\rho^k \leq C\, \tilde{L} \rho^{r}$, by virtue of \eqref{eqn:ergodic}.
Note that the $Y^*_{j,n}$ are independent and satisfy $Y^*_{j,n}\sim Y_{j,n}$, for $j=1,\ldots, m$.
Analogously, we may construct independent copies $\tilde{Y}^*_{j,n}$ of $\tilde{Y}_{j,n}$, $j=1,\ldots, m$, such that $\|\tilde{Y}_{j,n} - \tilde{Y}^*_{j,n}\|_{L_q} = \|\tilde{\delta}_{j,n}\|_{L_q} \leq C\, \tilde{L} \rho^r$.
Hence,
\begin{align*}
	S_n (u)
	&= \frac{1}{\sqrt{n}} \sum_{j=1}^{\lfloor um\rfloor} Y^*_{j,n}  
	+ \frac{1}{\sqrt{n}} \sum_{j=1}^{\lfloor um\rfloor} \tilde{Y}^*_{j,n}
	+ \tilde{R}_n(u) + R_n(u).
\end{align*}
The remainder term $\tilde{R}_n(u)$ satisfies 
\begin{align*}
	\|\sup_{u\in[0,1]} \tilde{R}_n(u)\|_{L_q} \leq \frac{1}{\sqrt{n}} \sum_{j=1}^m \left(\|\delta_{j,n}\|_{L_q} + \|\tilde{\delta}_{j,n}\|_{L_q}\right) \leq C \, \frac{\tilde{L} m}{\sqrt{n}} \rho^r ,
\end{align*} 
which is asymptotically negligible since $r=r_n=n^{\epsilon/2}$ and $\epsilon>0$, and $\tilde{L}_n m_n \leq n$.

We now establish a central limit theorem for the term $S_n^*(u) = \tfrac{1}{\sqrt{n}} \sum_{j=1}^{\lfloor um\rfloor} Y^*_{j,n}$. 
This result will also apply to the term $\tfrac{1}{\sqrt{n}} \sum_{j=1}^m \tilde{Y}^*_{j,n}$ with a different rate of convergence, and we will be able to conclude that the latter is negligible.
To this end, consider
\begin{align*}
	\Var(Y_{j,n}^*) 
	&= \sum_{t,s=(j-1)\tilde{L}+r}^{j\tilde{L}} \Cov \left[ G_n(\tfrac{t}{n}, \beps_t), G_n(\tfrac{s}{n}, \beps_s) \right] \\
	&= \sum_{t,s=(j-1)\tilde{L}+r}^{j\tilde{L}} \Cov \left[ G_n(\tfrac{t}{n}, \beps_t), G_n(\tfrac{t}{n}, \beps_s) \right] + g_n^{s,t},
\end{align*}
for some $|g_n^{s,t}|\leq C \|G_n(\tfrac{t}{n}, \beps_s)-G_n(\tfrac{s}{n},\beps_s)\|_{L_2}$, using the $L_2(P)$ boundedness of $G$.
Furthermore, condition \eqref{eqn:ergodic} implies that $|\Cov \left[ G_n(\tfrac{t}{n}, \beps_t), G_n(\tfrac{t}{n}, \beps_s)\right]| \leq C \rho^{|t-s|}$.
Hence, it can be checked that
\begin{align*}
	&\quad \sum_{s,t=(j-1)\tilde{L}+r}^{j\tilde{L}} \Cov \left[ G_n(\tfrac{t}{n}, \beps_t), G_n(\tfrac{t}{n}, \beps_s) \right] \\
	&= \bar{\delta}_{j,n} + \sum_{t=(j-1)\tilde{L}+r}^{j\tilde{L}} \sum_{s=-\infty}^\infty \Cov \left[ G_n(\tfrac{t}{n}, \beps_t), G_n(\tfrac{t}{n}, \beps_s) \right] \\
	&= \mathcal{O}(1) + \sum_{t=(j-1)\tilde{L}+r}^{j\tilde{L}} \sigma_n^2(\tfrac{t}{n})
\end{align*}
for some $|\bar{\delta}_{j,n}| \leq C$, and $\sigma_n(u) = \sum_{h=-\infty}^\infty \Cov(G_n(u, \beps_h), G_n(u,\beps_0)$.
Thus, for any $u\in[0,1]$, 
\begin{align}
	\sum_{j=1}^{\lfloor um \rfloor} \Var\left( \frac{Y^*_{j,n}}{\sqrt{n}} \right) 
	&= \frac{1}{n} \sum_{j=1}^{\lfloor um \rfloor} \bar{\delta}_{j,n} + \frac{1}{n}\sum_{j=1}^{\lfloor um \rfloor} \sum_{t=(j-1)\tilde{L}+r}^{j\tilde{L}} \sigma_n^2(\tfrac{t}{n}) +  \frac{1}{n}\sum_{j=1}^{\lfloor um \rfloor}\sum_{t,s=(j-1)\tilde{L}+1}^{j\tilde{L}} g_{s,t} \nonumber\\
	&= V_n^1(u) + V_n^2(u) + V_n^3(u). \label{eqn:Sn-var-1}
\end{align}
The term $V_n^1(u)$ tends to zero uniformly in $u$ because the $\bar{\delta}_{j,n}$ are bounded, and $m_n\ll n$.
The third term may be bounded as 
\begin{align*}
	\sup_{u\in[0,1]} |V_n^3(u)| 
	& \leq \frac{1}{n} \sum_{j=1}^m \sum_{t,s=(j-1)\tilde{L}+1}^{j\tilde{L}} C \|G_n(\tfrac{t}{n}, \beps_0)-G_n(\tfrac{s}{n},\beps_0)\|_{L_2} \\
	&\leq \frac{C}{n} \sum_{t=1}^n \sup_{s=1,\ldots, n, \left|s-t\right|\leq \tilde{L}} \|G_n(\tfrac{t}{n}, \beps_0)-G_n(\tfrac{s}{n},\beps_0)\|_{L_q}\\ 
	&\leq \frac{\tilde{L} C}{n} \sum_{t=1}^n \|G_n(\tfrac{t}{n}, \beps_0)-G_n(\tfrac{t-1}{n},\beps_0)\|_{L_q} \\
	&\leq \frac{\tilde{L} C}{n} n^{1-\frac{1}{p}} \left(\sum_{t=1}^n \|G_n(\tfrac{t}{n}, \beps_0)-G_n(\tfrac{t-1}{n},\beps_0)\|^p_{L_q}\right)^\frac{1}{p}\\
	&\leq \tilde{L}_n C n^{-\frac{1}{p}} \|G_n\|_{p-var},
\end{align*}
which tends to zero since $\|G_n\|_{p-var} \leq C_G$ by assumption, and $\tilde{L}_n n^{-\frac{1}{p}}\to 0$.

To treat the term $V_n^2(u)$, we show that $\|\sigma^2\|_{p-var}<\infty$.
First, we find that for $u,v\in[0,1]$ and any $h\in \Z$,
\begin{align*}
	&\quad \left| \Cov(G_n(u,\beps_h), G_n(u, \beps_0) ) - \Cov(G_n(v,\beps_h), G_n(v, \beps_0) ) \right| \\
	&\leq \left|  \Cov(G_n(u,\beps_h), G_n(u, \beps_0) - G_n(v, \beps_0) )  \right| + \left|  \Cov(G_n(u,\beps_h)-G_n(v,\beps_h), G_n(v, \beps_0) )  \right| \\
	&= \left|  \Cov(G_n(u,\beps_h), G_n(u, \beps_0) - G_n(v, \beps_0) )  \right| + \left|  \Cov(G_n(v,\beps_{-h}), G_n(u, \beps_0) - G_n(v, \beps_0) )  \right|.
\end{align*}
Both summands may be treated identically using \eqref{eqn:ergodic}.
In particular, assuming w.l.o.g.\ $h\geq 0$, 
\begin{align*}
	&\quad \left|  \Cov(G_n(u,\beps_h), G_n(u, \beps_0) - G_n(v, \beps_0) )  \right|\\
	&= \left|  \Cov(G_n(u,\beps_{h,h}^*), G_n(u, \beps_0) - G_n(v, \beps_0) )  \right| \\
	&\qquad + \left|  \Cov(G_n(u,\beps_h) - G_n(u, \beps^*_{h,h}), G_n(u, \beps_0) - G_n(v, \beps_0) )  \right| \\
	&\leq 0 + C\rho^{|h|} \|G_n(u)-G_n(v)\|_{L_2}.
\end{align*}
Hence,
\begin{align}
	|\sigma_n^2(u)-\sigma_n^2(v)| 
	&\leq \sum_{h=-\infty}^\infty  \left| \Cov(G_n(u,\beps_h), G_n(u, \beps_0) ) - \Cov(G_n(v,\beps_h), G_n(v, \beps_0) ) \right| \nonumber \\
	&\leq \sum_{h=-\infty}^\infty C\rho^{|h|} \|G_n(u)-G_n(v)\|_{L_2} \nonumber\\
	&\leq C \|G_n(u)-G_n(v)\|_{L_2}, \label{eqn:sigma-pvar}
\end{align}
which implies that $\|\sigma^2_n\|_{p-var} \leq C \|G_n\|_{p-var}\leq C \cdot C_G$.
Moreover, we have that $\sigma_n^2(u)\to \sigma^2(u)$, because the same argument as above yields
\begin{align}
	|\sigma_n^2(u) - \sigma^2(u)| \leq \sum_{h=-\infty}^\infty C \rho^{|h|} \|G_n(u,\beps_0)-G(u,\beps_0)\|_{L_2} \quad \to 0. \label{eqn:sigma-limit}
\end{align}
Hence, we find that $u\mapsto \sigma_n^2(u)$ is bounded and of bounded $p$-variation, uniformly in $n$.
The same holds for $\sigma^2(u)$.
This suffices to establish convergence of the Riemann sum $V_n^2(u)$ in \eqref{eqn:Sn-var-1}, since
\begin{align*}
	&\quad \sup_{u\in[0,1]}\left|V_n^2(u) - \int_0^u \sigma_n^2(v)\, dv \right| \\
	&= \sup_{u\in[0,1]}\left| \frac{1}{n} \sum_{j=1}^{\lfloor um\rfloor }\sum_{t=(j-1)\tilde{L}+r}^{j\tilde{L} } \sigma_n^2(\tfrac{t}{n}) - \int_0^u \sigma_n^2(v)\, dv \right| \\
	&= \sup_{u\in[0,1]}\left|\frac{1}{n} \sum_{t=1}^{\lfloor un\rfloor } \sigma_n^2(\tfrac{t}{n}) - \int_0^\frac{\lfloor un\rfloor}{n} \sigma_n^2(v)\, dv\right| + \mathcal{O}\left(\frac{mr}{n}\right) + \mathcal{O}\left(\frac{\tilde{L}}{n}\right) \\
	&= \sup_{u\in[0,1]}\left|\sum_{t=1}^{\lfloor un\rfloor } \int_{\frac{t-1}{n}}^{\frac{t}{n}} [\sigma_n^2(\tfrac{t}{n}) - \sigma_n^2(v)]\, dv\right| + \mathcal{O}\left(\frac{mr}{n}\right) + \mathcal{O}\left(\frac{\tilde{L}}{n}\right) + \mathcal{O}\left(\frac{1}{n}\right)\\
	&\leq \frac{1}{n}\sum_{t=1}^n \sup_{\frac{t-1}{n} \leq v < w \leq \frac{t}{n}} |\sigma_n^2(v)-\sigma_n^2(w)| + \mathcal{O}\left(\frac{mr}{n}\right) + \mathcal{O}\left(\frac{\tilde{L}}{n}\right) \\
	&\leq n^{-\frac{1}{p}} \|\sigma_n^2\|_{p-var} + \mathcal{O}\left(\frac{mr}{n}\right) + \mathcal{O}\left(\frac{\tilde{L}}{n}\right).
\end{align*}
By our choice of $\tilde{L}$ and $r$, in particular $r/\tilde{L} \asymp \frac{mr}{n}\to 0$, the latter term tends to zero. 
Moreover, $\int_0^u \sigma_n^2(v)\, dv \to \int_0^u \sigma^2(v)\, dv$ uniformly in $u$ by the dominated convergence theorem.
We have thus shown that 
\begin{align*}
	\Var \left( \frac{1}{\sqrt{n}} \sum_{j=1}^{\lfloor um \rfloor} Y^*_{j,n}\right) 
	&= \int_0^u \sigma^2(v)\, dv + o(1),
\end{align*}
uniformly in $u$.
On the other hand, following the same steps, we can show that 
\begin{align*}
	\Var\left( \frac{1}{\sqrt{n}} \sum_{j=1}^{\lfloor um \rfloor} \tilde{Y}_{j,n} \right) 
	&= o(1) + \frac{1}{n} \sum_{j=1}^{\lfloor um \rfloor} \sum_{t=(j-1)\tilde{L}+1}^{(j-1)\tilde{L}+r-1} \sigma_n^2(\tfrac{t}{n}) = o(1) + \mathcal{O}\left(\frac{mr}{n}\right),
\end{align*}
which tends to zero by our choice of $m,r$. 
Doob's inequality thus yields that
\begin{align*}
	\sup_{u\in [0,1]} \frac{1}{\sqrt{n}} \sum_{j=1}^{\lfloor um\rfloor} \tilde{Y}_{j,n} \to 0,
\end{align*}
in probability, as $n\to\infty$. 
Hence, $S_n(u) = \tfrac{1}{\sqrt{n}} \sum_{j=1}^{\lfloor um\rfloor} Y^*_{j,n} + o(1)$, uniformly in $u$.

To establish a central limit theorem, we verify Lyapunov's condition.
By our assumptions on $G_n$, in particular convergence and bounded $p$-variation in $L_q(P)$, we know that $\|G_n(u)\|_{L_q}\leq C$ for some constant $C$, uniformly in $n$ and $u$.
Hence, 
\begin{align}
	\sum_{j=1}^m \left\| Y^*_{j,n} / \sqrt{n}  \right\|_{L_q}^q 
	&\leq \sum_{j=1}^m \frac{(C\tilde{L})^q}{n^\frac{q}{2}} 
	= \mathcal{O} ( \tilde{L}^{q-1} / n^{q/2-1}) 
	= \mathcal{O} ( \tilde{L}n^{\frac{2-q}{2q-2}})^\frac{1}{q-1} , \label{eqn:Lyapunov}
\end{align}
since $m\tilde{L}/n \to 1$.
This term tends to zero since $\tilde{L}_n n^{\frac{2-q}{2q-2}} \leq \tilde{L}_n n^{\frac{1}{2}-\frac{1}{q}} \to 0$ by our choice of $\tilde{L}_n$.
Thus, the functional central limit theorem is applicable, see e.g.\ \citep[Thm. VIII.3.33]{Jacod2003}.
We obtain
\begin{align*}
	\frac{1}{\sqrt{n}}\sum_{j=1}^{\lfloor um\rfloor} Y^*_{j,n} \wconv B_{\int_0^u \sigma^2(v)\, dv}.
\end{align*}
\end{proof}

\subsection{Properties of the local smoother $\hat{\mu}_{t,n}^{NW}$}

Proposition \ref{prop:NW} is a consequence of the following two Lemmas.

%UNIFORM CONVERGENCE OF SPOT ESTIMATOR
\begin{lemma}\label{lem:uniform}
	Let \eqref{eqn:ergodic} hold for some $q>2$.  
	Let $\tau_n\to\infty$ and choose $k=k_n$ such that $k_n \gg n^{\frac{2}{q}}$.
	Then, as $n\to \infty$,
	\begin{align*}
		\sup_{t=\tau_n, \ldots, n} \left\| \hat{\mu}^{NW}_{t,n} - \frac{1}{k\wedge t} \sum_{i=(t-k)\vee 1}^{t} \mu_{\frac{i}{n}}^n  \right\| \pconv 0.
	\end{align*}
\end{lemma}
\begin{proof}[Proof of Lemma \ref{lem:uniform}]
The union bound and Markov's inequality yield, for any $a>0$,
\begin{align}
	&\quad P\left( \sup_{t=\tau_n, \ldots, n} \left\| \hat{\mu}^{NW}_{t,n} - \frac{1}{k\wedge t} \sum_{i=(t-k)\vee 1}^{t} \mu_{\frac{i}{n}}^n  \right\| > a \right) \nonumber\\
	&= P\left( \sup_{t=\tau_n, \ldots, n}\left\|\frac{1}{k\wedge t} \sum_{i=(t-k)\vee 1}^t \left[X_{i,n} - \E X_{i,n}\right] \right\| > a\right) \nonumber\\
	&\leq a^{-q} \sum_{t=\tau_n}^{n}   \left\| \frac{1}{k\wedge t} \sum_{i=(t-k)\vee 1}^t \left[X_{i,n} - \E X_{i,n}\right] \right\|^q_{L_q}. \label{eqn:union-bound}
\end{align}
To bound the latter $L_q$ norm, we apply Theorem 1 of \cite{Liu2013}.
The latter result is formulated for stationary time series of the type $X_t = G(\beps_t)$. 
However, the stationarity is only strictly required for the last equation of the proof therein, where we may replace $\|X_1\|_{L_q}$ by $s_q = \sup_n\sup_{u\in[0,1]}\|G_n(u, \beps_0)\|_{L_q}$, which is bounded by virtue of \eqref{eqn:ergodic} for $j=0$. 
At all other occasions in the proof of \cite{Liu2013}, only the quantity $\theta_{j,q} = \|G(\beps_0) - G(\tilde{\beps}_{0,j})\|_{L_q}$ is relevant, which may be replaced in our nonstationary setting by \begin{align*}
	\theta^*_{j,q} = \sup_n\sup_{u\in[0,1]} \|G_n(u,\beps_0) - G_n(u,\tilde{\beps}_{0,j})\|_{L_q},\qquad q\geq 2.
\end{align*}
By assumption \eqref{eqn:ergodic}, we have $\theta^*_{j,q} \leq C \rho^j$.
Hence, for our case, the bound of \cite{Liu2013} reads as follows:
For all $k$ and all $1\leq s < t\leq n$ such that $t-s=k$,
\begin{align*}
	&\quad \left\| \sum_{i=s}^t \left[X_{i,n} - \E X_{i,n}\right] \right\|_{L_q} \\
	&\leq k^{1/2} \left[\frac{87\,q}{\log q} \sum_{j=1}^k \theta^*_{j,2} + 3(q-1)^{1/2} \sum_{j=k+1}^\infty \theta^*_{j,q} + \frac{29\, q}{\log q} s_2\right] \\
	&\quad + k^{1/q} \left[ \frac{87 \, q(q-1)^{1/2}}{\log q} \sum_{j=1}^k j^{1/2-1/q} \theta^*_{j,q} + \frac{29\, q}{\log q} s_q \right] \\
	&\leq k^{1/2} \left(s_q + \sum_{j=1}^\infty \theta^*_{j,q} \right) \left( \frac{87\,q}{\log q} +  3(q-1)^{1/2} + \frac{58\, q}{\log q} + \frac{87 \, q(q-1)^{1/2}}{\log q} \right) \\
	&\leq C k^{1/2},
\end{align*}
for a factor $C$ which does not depend on $k$ nor $n$. 
For the latter two inequalities, we used that $\theta^*_{j,2}\leq \theta^*_{j,q}$ and $s_2\leq s_q$ for $q\geq 2$, and the geometric decay of $\theta^*_{j,q}$.
Thus,
\begin{align*}
	\sum_{t=\tau_n}^n  \left\| \frac{1}{k\wedge t} \sum_{i=(t-k)\vee 1}^t \left[X_{i,n} - \E X_{i,n}\right] \right\|^q_{L_q}
	&\leq C \sum_{t=\tau_n}^n (k\wedge t)^{-q/2} \\
	&\leq C \left(n k^{-q/2} + \tau_n^{1-q/2}\right).
\end{align*}
This term tends to zero if $q>2$.
We hence obtain the uniform convergence via the union bound \eqref{eqn:union-bound}.
\end{proof}

%MSE CONSISTENCY OF MU_HAT
\begin{lemma}\label{lem:I2}
	Let \eqref{eqn:ergodic} hold for some $q\geq 2$.
	Suppose that $u\mapsto\mu_u^n$ is $\beta$-Hölder continuous with Hölder constant $C$ independent of $n$.
	Then
	\begin{align*}
		 \frac{1}{n} \sum_{t=1}^{n} \left\| \hat{\mu}^{NW}_{t,n} - \mu_{\frac{t}{n}}^n \right\|^2 
		 = \mathcal{O}_P\left(\left(\frac{k}{n}\right)^{2\beta} + \frac{\log n}{k}\right).
	\end{align*}
	If, instead of Hölder continuity, we assume that $\|\mu^n\|_{p-var}<C<\infty$ for some $p\in [1,2)$, then
	\begin{align*}
		\frac{1}{n} \sum_{t=1}^{n} \left\| \hat{\mu}^{NW}_{t,n} - \mu_{\frac{t}{n}}^n\right\|^2 
		= \mathcal{O}_P\left(\frac{k}{n} + \frac{\log n}{k}\right).
	\end{align*}
	Furthermore, if $\mu^n_{u} = \mu^{n,1}_{u} + \mu^{n,2}_{u}$ such that $u\mapsto \mu^{n,1}_{u}$ is $\beta$-Hölder continuous with Hölder constant $C$ independent of $n$, and $\|\mu^{n,2}\|_{p-var}\leq a_n$, then
	\begin{align*}
		\frac{1}{n} \sum_{t=1}^{n} \left\| \hat{\mu}^{NW}_{t,n} - \mu_{\frac{t}{n}}^n\right\|^2 
		= \mathcal{O}_P\left(\frac{a_n k}{n} + \left(\frac{k}{n}\right)^{2\beta} + \frac{\log n}{k}\right).
	\end{align*}
\end{lemma}
\begin{proof}[Proof of Lemma \ref{lem:I2}]
	Assume without loss of generality that the $X_{i,n}$ are scalar, i.e.\ $d=1$.
	Otherwise, we may treat each component individually.
	We have
	\begin{align}
		\begin{split}
			\frac{1}{n} \sum_{t=1}^n  \E \|\hat\mu^{NW}_{t,n}-\mu_{\frac{t}{n}}^n\|^2
			&\leq \frac{2}{n} \sum_{t=1}^{n} \E \left\| \frac{1}{k\wedge t} \sum_{i=(t-k)\vee 1}^t \left[X_{i,n} - \E X_{i,n}\right] \right\|^2 \\
			&\quad + \frac{2}{n} \sum_{t=1}^{n} \left\| \frac{1}{k\wedge t} \sum_{i=(t-k)\vee 1}^t \left[\mu_{\frac{i}{n}}^n - \mu_{\frac{t}{n}}^n\right] \right\|^2.
		\end{split} \label{eqn:I2-decomp}
	\end{align}
	The Hölder continuity of $\mu_u^n$ yields 
	\begin{align*}
		\frac{1}{n} \sum_{t=1}^{n} \left\| \frac{1}{k\wedge t} \sum_{i=(t-k)\vee 1}^t \left[\mu^n_{\frac{i}{n}} - \mu^n_{\frac{t}{n}}\right] \right\|^2
		&\leq C  \left(\frac{k}{n}\right)^{2\beta}.
	\end{align*}
	Alternatively, if we only assume $\|\mu^n\|_{p-var}<\infty$,
	\begin{align*}
		 \frac{1}{n} \sum_{t=1}^{n} \left\| \frac{1}{k\wedge t} \sum_{i=(t-k)\vee 1}^t \left[\mu^n_{\frac{i}{n}} - \mu^n_{\frac{t}{n}}\right] \right\|^2 
		&\leq \frac{1}{n}  \sum_{t=1}^n \sup_{s \in \left[\frac{t-k}{n}\vee 0, t\right]} \|\mu^n_{t/n} - \mu^n_s \|^2 \\ 
		&\leq \frac{1}{n} (2\|\mu\|_\infty)^{2-p}  \sum_{t=1}^n \sup_{s \in \left[\frac{t-k}{n}\vee 0, t\right]} \|\mu^n_{t/n} - \mu^n_s \|^p \\
		&\leq \frac{k}{n} (2\|\mu^n\|_\infty)^{2-p}  \|\mu^n\|_{p-var}^p
	\end{align*}
	Note that $\|\mu^n\|_\infty \leq |\mu^n_0| + \|\mu^n\|_{p-var}$, which is bounded.
	The third case is obtained by a combination of the two previous bounds.
	
	Regarding the first term in \eqref{eqn:I2-decomp}, note that \eqref{eqn:ergodic} implies $|\Cov(G_n(u, \beps_i), G_n(v, \beps_j))| \leq C\rho^{|i-j|}$, for any $u,v\in[0,1]$ and any $n$.
	Hence,
	\begin{align*}
		 \E \left\| \frac{1}{k\wedge t} \sum_{i=(t-k)\vee 1}^t \left[X_{i,n} - \E X_{i,n}\right] \right\|^2 
		&=\Var \left( \frac{1}{k\wedge t} \sum_{i=(t-k)\vee 1}^t G_n(\tfrac{i}{n}, \beps_i) \right)\\
		&= \frac{1}{(k\wedge t)^2}\sum_{i,j=(t-k)\vee 1}^t \Cov\left( G_n(\tfrac{i}{n}, \beps_i), G_n(\tfrac{j}{n}, \beps_j) \right) \\
		&\leq \frac{1}{k\wedge t} \sum_{j=1}^\infty C \rho^j.
	\end{align*}
	Therefore, 
	\begin{align*}
		&\quad \frac{1}{n} \sum_{t=1}^{n} \E \left\| \frac{1}{k\wedge t} \sum_{i=(t-k)\vee 1}^t \left[X_{i,n} - \E X_{i,n}\right] \right\|^2 \\
		&\leq \frac{C}{n} \sum_{t=1}^{n}  \frac{1}{k\wedge t} \\
		&\leq \frac{C}{k} + \frac{C}{n} \sum_{t=1}^k \frac{1}{t} \\
		&\leq \frac{C}{k} + \frac{C \log k}{n} 
		\quad\leq\quad \frac{2C \log n}{k},
	\end{align*}
	completing the proof.
\end{proof}

Of course, Lemma \ref{lem:uniform} and Lemma \ref{lem:I2} also hold with $\mu_u^n$ replaced by its limit $\mu_u$ if we use the samples $X_{t} = G(\tfrac{t}{n},\beps_t)$ instead of $X_{t,n}=G_n(\tfrac{t}{n},\beps_t)$.

\begin{proof}[Proof of Proposition \ref{prop:NW}]
	Clearly, Lemma \ref{lem:I2} implies that \eqref{eqn:ass-mu-1} holds.
	Regarding \eqref{eqn:ass-mu-2}, note that $\sup_{u\in[0,1]} |\mu^n_u - \mu_u| \to 0$ by virtue of \eqref{eqn:GnG}.
	Thus, for any $\delta>0$, we have that $\frac{1}{k\wedge t} \sum_{i=(t-k)\vee 1}^t \mu^n_{\frac{i}{n}} \in \mathcal{M}_{\delta/2}$ for $n$ sufficiently large. 
	Together with Lemma \ref{lem:uniform}, this implies \eqref{eqn:ass-mu-2}.
\end{proof}

\subsection{Functional central limit theorem for $M_n(u)$}

Recall the decomposition $M_n(u) = I_n^1(u) + I_n^2(u) + I_n^3(u)$, where
\begin{align*}
	I_n^1(u) &= \frac{1}{n} \sum_{t=\tau_n+L_n}^{\lfloor nu \rfloor} f\left( \mu_{\frac{t}{n}}^n \right) , \\
	I_n^2(u) &= \frac{1}{n} \sum_{t=\tau_n+L_n}^{\lfloor nu \rfloor} \left( \mu^n_{\frac{t}{n}} - \hat{\mu}_{t-L,n} \right)^T\frac{D^2f\left(\tilde{\mu}_{t,n}\right)}{2} \left( \mu^n_{\frac{t}{n}} - \hat{\mu}_{t-L,n} \right), \\
	I_n^3(u) &=  \frac{1}{n} \sum_{t=\tau_n+L_n}^{\lfloor nu \rfloor} Df(\hat{\mu}_{t-L,n}) \left( X_{t,n}-\mu^n_{\frac{t}{n}} \right).
\end{align*}

The first term $I_n^1(u)$ is a discretized version of the integral $\int_0^u f(\mu^n_v)\, dv$, and its convergence depends on the regularity of the function $u\mapsto \mu^n_u$. 
Assumption \eqref{eqn:pvar-finite} implies $\|\mu^n\|_{p-var}<C_G$, which may be used to treat the term $I_n^1(u)$ as follows.

%INTEGRAL DISCRETIZATION
\begin{lemma}\label{lem:I1}
	For any function $u\mapsto g_u\in\R^d$, it holds that
	\begin{align*}
		\sup_{u\in[0,1]}\left\|\frac{1}{n} \sum_{t=1}^{\lfloor un\rfloor} g_{\frac{t}{n}} - \int_{0}^{\frac{\lfloor un\rfloor}{n}} g_v\, dv \right\| &\leq n^{-\frac{1}{p}} \|g\|_{p-var}.
	\end{align*}
	
	Furthermore, if \eqref{eqn:GnG} and \eqref{eqn:pvar-finite} hold for $p\geq 1$, then there exists $N>0$ such that for all $n\geq N$,
	\begin{align}
		\sup_{u\in [0,1]} \left|I_n^1(u) - \int_{\frac{\tau_n+L_n}{n}}^u f(\mu_v^n)\, dv \right|
		&\leq C_f \|\mu^n\|_{p-var} n^{-\frac{1}{p}},\label{eqn:I1-bound-cutoff} \\
		\sup_{u\in [0,1]} \left|I_n^1(u) - \int_0^u f(\mu_v^n)\, dv \right|
		&\leq C_f \|\mu^n\|_{p-var} n^{-\frac{1}{p}} + C_f \frac{\tau_n+L_n}{n}. \label{eqn:I1-bound}
	\end{align}
\end{lemma}
\begin{proof}[Proof of Lemma \ref{lem:I1}]
	It holds that
	\begin{align*}
		\left\|\frac{1}{n} \sum_{t=1}^{\lfloor un\rfloor} g_{\frac{t}{n}} - \int_{0}^{\frac{\lfloor un\rfloor}{n}} g_v\, dv \right\|
		\leq \int_{0}^{1} \left\| g_v - g_{\frac{\lceil vn\rceil}{n}}\right\|\, dv 
		&\leq \frac{1}{n} \sum_{t=1}^n  \sup_{v\in[\frac{t-1}{n}, \frac{t}{n}]} \left\|g_v - g_{\frac{t}{n}}\right\| \\
		&\leq \frac{n^{1-\frac{1}{p}}}{n} \left(\sum_{t=1}^n  \sup_{v\in[\frac{t-1}{n}, \frac{t}{n}]} \left\|g_v - g_{\frac{t}{n}}\right\|^p\right)^\frac{1}{p} \\
		&\leq n^{-\frac{1}{p}} \|g\|_{p-var}.
	\end{align*}
	
	Regarding the term $I_n^1(u)$, note that $\mu_u^n\in\mathcal{M}_\delta$ for all $u\in[0,1]$, and $n$ large enough.
	Then the mean value theorem and the boundedness of $D f$ on $\mathcal{M}_\delta$ yield $|f(\mu^n_u)-f(\mu^n_v)| \leq C_f \|\mu^n_u-\mu^n_v\|$ for any $u,v\in[0,1]$, such that $\|f(\mu^n)\|_{p-var}\leq C_f \|\mu^n\|_{p-var}$.
	Thus, for $n$ large enough and for any $u$ such that $\lfloor nu\rfloor > \tau_n+L_n$,
	\begin{align*}
		\left|I_n^1(u) - \int_{\frac{\tau_n+L_n}{n}}^{\frac{\lfloor nu\rfloor}{n}} f(\mu_v^n) \, dv\right|
		& \leq C_f n^{-\frac{1}{p}} \|\mu^n\|_{p-var}.
	\end{align*}
	Moreover, since $f(\mu^n_v)$ is bounded,
	\begin{align*}
		\left|\int_0^{\frac{\tau_n+L_n}{n}} f(\mu^n_v)\, dv\right| \leq C_f \frac{\tau_n+L_n}{n},\qquad \left|\int_{\frac{\lfloor nu\rfloor}{n}}^{u} f(\mu^n_v)\, dv\right| \leq  \frac{C_f}{n}\to 0,
	\end{align*} 
	completing the proof.
\end{proof}

In particular, for $(\tau_n+L_n) \ll n^\frac{1}{2}$ and $p<2$, we find that 
\begin{align*}
	\sup_{u\in [0,1]} \left|I_n^1(u) - \int_0^u f(\mu_v^n)\, dv \right| = o(1/\sqrt{n}).
\end{align*}

The term $I_n^2(u)$ may be treated directly by our assumptions on the nonparametric estimator $\hat{\mu}_{t,n}$, and using additionally the regularity of $u\mapsto \mu^n_u$.

%MSE bound
\begin{lemma}\label{lem:I2-L}
	Let \eqref{eqn:GnG}, \eqref{eqn:pvar-finite}, \eqref{eqn:ass-f}, \eqref{eqn:ass-mu-1}, and \eqref{eqn:ass-mu-2} hold, for some $p\geq 1$.
	Then
	\begin{align*}
		\sup_{u\in[0,1]} |I_n^2(u)| = o_P(1/\sqrt{n}) + \mathcal{O}_P(L_n^{\min(p,2)} n^{-\frac{2}{\max(p,2)}}).
	\end{align*}
	In particular, $\sup_{u\in[0,1]} |I_n^2(u)| = o_P(1/\sqrt{n})$ if $L_n\ll n^{\frac{1}{p}-\frac{1}{4}}$ for $p\geq 2$, and $L_n \ll n^\frac{1}{2p}$ for $p\in[1,2)$. 
\end{lemma} 
\begin{proof}[Proof of Lemma \ref{lem:I2-L}]
	Assumption \eqref{eqn:GnG} implies that for $n$ sufficiently large, $\mu^n_u\in\mathcal{M}_\delta$ for all $u\in[0,1]$.
	Moreover, $\hat{\mu}_{t-L,n}\in\mathcal{M}_\delta$ for all $t\geq \tau_n+L_n$, with probability tending to one. 
	Since $\mathcal{M}_\delta$ is convex, we find that $\tilde{\mu}_{t,n}\in\mathcal{M}_\delta$ for all $t\geq \tau_n+L_n$, with probability tending to one.
	On this event, it holds that
	\begin{align*}
		|I_n^2(u)| 
		&\leq \frac{2C_f}{n} \sum_{t=\tau_n+L_n}^{\lfloor nu\rfloor} \| \hat{\mu}_{t-L,n} - \mu^n_{\frac{t-L}{n}}\|^2 + \|\mu^n_{\frac{t-L}{n}} - \mu^n_{\frac{t}{n}}\|^2 \\
		&\leq  o_P(1/\sqrt{n}) + \frac{2C_f C_G}{n} \sum_{t=\tau_n+L_n}^{\lfloor nu\rfloor} \|\mu^n_{\frac{t-L}{n}} - \mu^n_{\frac{t}{n}}\|^2.
	\end{align*}
	
	Moreover, if $p<2$, boundedness of $\|\mu_t^n\|$ yields
	\begin{align*}
		\frac{1}{n} \sum_{t=\tau_n+L_n}^n \|\mu^n_{\frac{t-L}{n}} - \mu^n_{\frac{t}{n}}\|^2
		&\leq \frac{C}{n} \sum_{t=\tau_n+L_n}^n \|\mu^n_{\frac{t-L}{n}} - \mu^n_{\frac{t}{n}}\|^p \\
		&\leq \frac{C}{n} \sum_{t=\tau_n+L_n}^n \left(\sum_{i=1}^{L_n}\|\mu^n_{\frac{t-i+1}{n}} - \mu^n_{\frac{t-i}{n}}\|\right)^p \\
		&\leq \frac{C}{n} \sum_{t=\tau_n+L_n}^n L_n^{p-1} \sum_{i=1}^{L_n}\|\mu^n_{\frac{t-i+1}{n}} - \mu^n_{\frac{t-i}{n}}\|^p \\
		&\leq \frac{L_n^p}{n} \sum_{t=1}^n \|\mu^n_{\frac{t}{n}} - \mu^n_{\frac{t-1}{n}}\|^p 
		\qquad = L_n^p n^{-1} \|\mu^n\|_{p-var}^p.
	\end{align*}
	If $p>2$, we may apply Hölder's inequality to obtain
	\begin{align*}
		\frac{1}{n} \sum_{t=\tau_n+L_n}^n \|\mu^n_{\frac{t-L}{n}} - \mu^n_{\frac{t}{n}}\|^2
		&\leq 	\frac{1}{n} n^{1-\frac{2}{p}}\left(\sum_{t=\tau_n+L_n}^n \|\mu^n_{\frac{t-L}{n}} - \mu^n_{\frac{t}{n}}\|^p\right)^\frac{2}{p} \\
		&= \left(\frac{1}{n}\sum_{t=\tau_n+L_n}^n \|\mu^n_{\frac{t-L}{n}} - \mu^n_{\frac{t}{n}}\|^p\right)^\frac{2}{p}
		&\leq L_n^2 n^{-\frac{2}{p}} \|\mu^n\|_{p-var}^2.
	\end{align*}
	This yields the claimed upper bound.
\end{proof}

The last term $I_n^3(u)$ admits a functional central limit theorem with scaling factor $\sqrt{n}$, as shown in the following Lemma. 
Its proof follows the same ideas as the proof of Theorem \ref{thm:CLT-linear}, but additional care is needed when treating the random but consistent factor $Df(\hat{\mu}_{t,n})\approx Df(\mu^n_{t/n})$.

\begin{lemma}\label{lem:I3}
Suppose that \eqref{eqn:GnG}, \eqref{eqn:pvar-finite}, \eqref{eqn:ergodic}, \eqref{eqn:ass-f}, \eqref{eqn:ass-mu-1}, and \eqref{eqn:ass-mu-2} hold for some $p\geq 1$.
Suppose that $L_n$ satisfies $L_n \gg \log(n)^{1+a}$ for some $a\in(0,1)$, and that $L_n \ll n^{\frac{1}{p}}$. 
Then, as $n\to\infty$,
\begin{align*}
	\sqrt{n}I_n^3(u) &\wconv B\left(\int_0^u Df(\mu_v) \Sigma(v) Df(\mu_v)^T\,dv\right), \\
	\Sigma(v) &= \sum_{h=-\infty}^\infty \Cov \left(G(v, \beps_0), G(v, \beps_h)\right) \in\R^{d\times d},
\end{align*}
where $B$ denotes a standard Brownian motion.
The weak convergence holds in the Skorokhod space $D[0,1]$.
\end{lemma}
\begin{proof}[Proof of Lemma \ref{lem:I3}]
	As in the proof of Theorem \ref{thm:CLT-linear}, we split the sum into blocks of size $\tilde{L}=\tilde{L}_n = \lfloor n^\epsilon \wedge \frac{L_n}{2}\rfloor$, and $r=r_n=\lfloor\tilde{L}_n^{1-\epsilon} \rfloor$, for some arbitrarily small $\epsilon\in(0,1)$ to be chosen in the sequel. 
	We assume $n$ to be sufficiently large such that $1\leq r_n < \tilde{L}_n\leq \frac{L_n}{2}$.
	Denote by $m=m_n=\lfloor n/\tilde{L}_n \rfloor$ the number of blocks. 
	
	\underline{Step (i):}
	As a first step, we match the upper and lower limits of the sum $I_n^3(u)$ with the block size $L=L_n$. 
	Define
	\begin{align*}
		\bar{I}_n^3(u) 
		&= \frac{1}{n} \sum_{t=\left\lceil \frac{\tau_n+L_n}{\tilde{L}}\right\rceil \tilde{L}+1}^{\lfloor um\rfloor \tilde{L} }  Df(\hat{\mu}_{t-L,n}) \left( X_{t,n}-\mu_{\frac{t}{n}}^n \right).
	\end{align*}
	By virtue of Lemma \ref{lem:uniform} and \eqref{eqn:ass-f}, $\|Df(\hat{\mu}_{t,n})\| \leq C_f$ for all $t$, with probability tending to one, such that
	\begin{align*}
		\left| I_n^3(u)- \bar{I}_n^3(u)\right| 
		&= \mathcal{O}_P \left( \frac{2\tilde{L}_n\, C_f}{n} \sup_{t=1,\ldots, n} \left|X_{t,n}-\E X_{t,n}\right|\right)  .
	\end{align*}
	Since $\|X_{t,n}\|_{L_q}\leq C<\infty$, the union bound and Markov's inequality yield the upper bound $\sup_{t=1,\ldots, n} \left|X_{t,n}-\E X_{t,n}\right| = \mathcal{O}_P(n^{\frac{1}{q}})$.
	Using $\tilde{L}_n \leq n^\epsilon$, we obtain
	\begin{align*}
		\sup_{u\in[0,1]} \left| I_n^3(u) - \bar{I}_n^3(u) \right| & = \mathcal{O}_P\left(   n^{\epsilon+\frac{1}{q}-1}  \right).
	\end{align*}
	Since $q>2$, we may choose $\epsilon>0$ sufficiently small such that the latter term is of order $o_P(1/\sqrt{n})$. 
	
	\underline{Step (ii):}
	Next, we introduce $\chi_{t,n} = \mathds{1}(\hat{\mu}_{t,n}\in\mathcal{M}_\delta)$, and
	\begin{align*}
		Y_{j,n} &= \sum_{t=(j-1) \tilde{L}+r+1 }^{j\tilde{L}} \chi_{t-L,n} Df(\hat{\mu}_{t-L,n})(X_{t,n}-\E X_{t,n}), \\
		\tilde{Y}_{j,n} &= \sum_{t=(j-1)\tilde{L}+1 }^{(j-1)\tilde{L}+r} \chi_{t-L,n} Df(\hat{\mu}_{t-L,n})(X_{t,n}-\E X_{t,n}),\\
		\tilde{I}_n^3(u) &= \frac{1}{n}\sum_{j=\left\lceil \frac{\tau_n+L_n}{\tilde{L}}\right\rceil +1}^{\lfloor um\rfloor} (Y_{j,n} + \tilde{Y}_{j,n}).
	\end{align*}
	By virtue of assumption \eqref{eqn:ass-mu-2}, we have $\sup_{u\in[0,1]} |\bar{I}_n^3(u) - \tilde{I}_n^3(u)| \to 0$ in probability.
	Now, for each $j$, let $\epsilon_i^{(j)}, i\in\Z$ be an independent copy of the $\epsilon_i$, and define 
	\begin{align*}
		Y^*_{j,n}         & =  \sum_{t=(j-1)\tilde{L}+r+1}^{j\tilde{L}} \chi_{t-L,n}Df(\hat{\mu}_{t-L,n}) \left[G_n\left(\tfrac{t}{n}, \epsilon_{t},\ldots, \epsilon_{(j-1)\tilde{L}+1}, \epsilon_{(j-1)\tilde{L}}^{(j)},\ldots\right)-\E X_{t,n}\right] \\
		& =  \sum_{t=(j-1)\tilde{L}+r+1}^{j\tilde{L}} \chi_{t-L,n}Df(\hat{\mu}_{t-L,n}) Z_{t,n}^j, \\
		\tilde{Y}^*_{j,n} & =  \sum_{t=(j-1)\tilde{L}+1}^{(j-1)\tilde{L}+r} \chi_{t-L,n} Df(\hat{\mu}_{t-L,n}) \left[ G_n\left(\tfrac{t}{n}, \epsilon_{t},\ldots, \epsilon_{(j-1)\tilde{L}+1-r}, \epsilon_{(j-1)\tilde{L}-r}^{(j)},\ldots\right)-\E X_{t,n}\right] \\
		& =  \sum_{t=(j-1)\tilde{L}+1}^{(j-1)\tilde{L}+r} \chi_{t-L,n} Df(\hat{\mu}_{t-L,n}) \tilde{Z}_{t,n}^j.
	\end{align*}
	This construction yields that the  $Y^*_{j,n}$ and the $\tilde{Y}^*_{j,n}$ are martingale differences, since $\hat{\mu}_{t-L,n}$ is measurable w.r.t.\ $\beps_{t-L}$, and since $\tilde{L}_n+r_n<L_n$. 
	It also holds that $\tilde{Z}_{t,n}^j \sim Z_{t,n}^j \sim (X_{t,n} - \E X_{t,n})$.
	Moreover, by virtue of \eqref{eqn:ergodic}, it holds that $\|Y^*_{j,n}-Y_{j,n}\|_{L_q} \leq C \, \tilde{L} \rho^r$ and $\|\tilde{Y}_{j,n}^*-\tilde{Y}_{j,n}\|_{L_q}\leq C\, r \rho^r \leq C\, \tilde{L}_n \rho^r$.
	Now define the martingales
	\begin{align*}
		J_n^1(u) &= \frac{1}{n}\sum_{j=\left\lceil \frac{\tau_n+L_n}{\tilde{L}}\right\rceil +1}^{\lfloor um\rfloor}  Y^*_{j,n},\\
		J_n^2(u) &= \frac{1}{n}\sum_{j=\left\lceil \frac{\tau_n+L_n}{\tilde{L}}\right\rceil +1}^{\lfloor um\rfloor}  \tilde{Y}^*_{j,n},
	\end{align*}
	and note that
	\begin{align*}
		\left\|\sup_{u\in[0,1]}\left| \tilde{I}_n^3(u) - J_n^1(u) - J_n^2(u)\right| \right\|_{L_q}
		&\leq  \frac{C}{n} \sum_{j=1}^{m_n} \tilde{L}_n \rho^{r_n}  \leq C\rho^{r_n} = o(\sqrt{n}),
	\end{align*}
	since $\tilde{L}_n m_n\leq n$, and $r_n\gg \log(n)^{(1+a)(1-\epsilon)}\gg \log(n)^{1+\frac{a}{2}}$, upon choosing $\epsilon$ sufficiently small.
	
	\underline{Step (iii):}
	Define the two sequences of filtrations $\F^1_{k,n} = \sigma(\epsilon_i, \epsilon^{(j)}_i :j\leq  k, i\leq k\tilde{L}_n )$, and $\F^2_{k,n} = \sigma(\epsilon_i, \epsilon^{(j)}_i :j\leq k, i\leq k\tilde{L}_n-r_n )$.
	Then $J^1_n(k/m)$ is an $\F^1_{k,n}$ martingale, and $J^2_n(k/m)$ is a $\F^2_{k,n}$ martingale.
	The predictable quadratic variation of $\sqrt{n} J^1_n$ is given by
	\begin{align*}
		\langle \sqrt{n} J_n^1(u)\rangle 
		&= \frac{1}{n}\sum_{j=\left\lceil \frac{\tau_n+L_n}{\tilde{L}}\right\rceil +1}^{\lfloor um\rfloor} \E\left[ \left( Y^*_{j,n} \right)^2  | \F^1_{j-1,n}\right] \\
		&= \frac{1}{n} \sum_{j=\left\lceil \frac{\tau_n+L_n}{\tilde{L}}\right\rceil +1}^{\lfloor um\rfloor} \sum_{s,t=(j-1)\tilde{L}+r+1}^{j\tilde{L}} \chi_{t-L,n} \chi_{s-L,n} Df(\hat{\mu}_{t-L,n})\Cov(Z^j_{t,n}, Z^j_{s,n})Df(\hat{\mu}_{s-L,n})^T.
	\end{align*}
	We observe that
	\begin{align*}
		&\quad \left|Df(\hat{\mu}_{t-L,n})\Cov(Z^j_{t,n}, Z^j_{s,n})Df(\hat{\mu}_{s-L,n})^T
		- Df(\mu^n_{\frac{t}{n}})\Cov(Z^j_{t,n}, Z^j_{s,n})Df(\mu^n_{\frac{s}{n}})^T\right| \chi_{t-L,n} \chi_{s-L,n} \\
		&\leq \left|Df(\hat{\mu}_{t-L,n})\Cov(Z^j_{t,n}, Z^j_{s,n})\left[Df(\hat{\mu}_{s-L,n}) - Df(\mu^n_{\frac{s}{n}})\right]^T \right|\chi_{t-L,n} \chi_{s-L,n} \\
		&\qquad+ \left| \left[Df(\hat{\mu}_{t-L,n}) - Df(\mu^n_{\frac{t}{n}}) \right]\Cov(Z^j_{t,n}, Z^j_{s,n})Df(\hat{\mu}_{s-L,n})^T  \right|\chi_{t-L,n} \chi_{s-L,n}\\
		&\leq C \rho^{|t-s|}  \left[ \left\|Df(\hat{\mu}_{t-L,n}) - Df(\mu^n_{\frac{t}{n}}) \right\| +  \left\| Df(\hat{\mu}_{s-L,n}) - Df(\mu^n_{\frac{s}{n}}) \right\| \right]\chi_{t-L,n} \chi_{s-L,n} \\
		&\leq C \rho^{|t-s|}  \left[ \left\|\hat{\mu}_{t-L,n} - \mu^n_{\frac{t}{n}} \right\| +  \left\| \hat{\mu}_{s-L,n} - \mu^n_{\frac{s}{n}} \right\| \right],
	\end{align*}
	where the last inequalities are a consequence of \eqref{eqn:ergodic}, and the boundedness of $Df$ and $D^2f$ on $\mathcal{M}_\delta$.
	Moreover, 
	\begin{align}
		&\quad \frac{1}{n}\sum_{j=\left\lceil \frac{\tau_n+L_n}{\tilde{L}}\right\rceil +1}^{\lfloor um\rfloor} \sum_{s,t=(j-1)\tilde{L}+r+1}^{j\tilde{L}} \rho^{|t-s|}  \left[ \left\|\hat{\mu}_{t-L,n} - \mu^n_{\frac{t}{n}} \right\| +  \left\| \hat{\mu}_{s-L,n} - \mu^n_{\frac{s}{n}} \right\| \right] \nonumber\\
		&= \frac{2}{n}\sum_{j=\left\lceil \frac{\tau_n+L_n}{\tilde{L}}\right\rceil +1}^{\lfloor um\rfloor} \sum_{s,t=(j-1)\tilde{L}+r+1}^{j\tilde{L}} \rho^{|t-s|}   \left\|\hat{\mu}_{t-L,n} - \mu^n_{\frac{t}{n}} \right\| \nonumber\\
		&\leq \frac{C}{n} \sum_{j=\left\lceil \frac{\tau_n+L_n}{\tilde{L}}\right\rceil +1}^{\lfloor um\rfloor} \sum_{t=(j-1)\tilde{L}+r+1}^{j\tilde{L}}  \left\|\hat{\mu}_{t-L,n} - \mu^n_{\frac{t}{n}} \right\| \nonumber\\
		&\leq \frac{C}{n} \sum_{t=\tau_n}^{n}  \left\|\hat{\mu}_{t-L,n} - \mu^n_{\frac{t}{n}} \right\| \nonumber\\
		&\leq \frac{C}{\sqrt{n}} \sqrt{\sum_{t=\tau_n}^{n}  \left\|\hat{\mu}_{t-L,n} - \mu^n_{\frac{t}{n}} \right\|^2}. \label{eqn:I3-cov-det}
	\end{align}
	In the proof of Lemma \ref{lem:I2-L}, we have shown that $\frac{1}{n}\sum_{t=\tau_n}^{n}  \left\|\hat{\mu}_{t-L,n} - \mu^n_{\frac{t}{n}} \right\|^2 = o_p(1/\sqrt{n}) + \mathcal{O}_P(L_n^{\min(p,2)} n^{-\frac{2}{\max(p,2)}})$, which tends to zero because $L_n\ll n^\frac{1}{p}$ by assumption.
	Hence, \eqref{eqn:I3-cov-det} tends to zero as $n\to\infty$, and we thus obtain 
	\begin{align}
		\langle \sqrt{n} J_n^1(u)\rangle 
		&= o_P(1) + \sum_{j=\left\lceil \frac{\tau_n+L_n}{\tilde{L}}\right\rceil +1}^{\lfloor um\rfloor} \sum_{s,t=(j-1)\tilde{L}+r+1}^{j\tilde{L}}  Df(\mu^n_{\frac{t}{n}})\Cov(Z^j_{t,n}, Z^j_{s,n})Df(\mu^n_{\frac{s}{n}})^T \label{eqn:J-quadratic-1}
	\end{align}
	where the $o_P(1)$ term vanishes uniformly in $u\in[0,1]$. 
	Note that we may omit the indicators $\chi_{t-L,n}$ asymptotically by virtue of assumption \eqref{eqn:ass-mu-2}.
	Analogously, we obtain
	\begin{align*}
		\langle \sqrt{n} J_n^2(u)\rangle 
		&= o_P(1) + \frac{1}{n} \sum_{j=\left\lceil \frac{\tau_n+L_n}{\tilde{L}}\right\rceil +1}^{\lfloor um\rfloor} \sum_{s,t=(j-1)\tilde{L}+1}^{(j-1)\tilde{L}+r}  Df(\mu^n_{\frac{t}{n}})\Cov(Z^j_{t,n}, Z^j_{s,n})Df(\mu^n_{\frac{s}{n}})^T.
	\end{align*}
	
	\underline{Step (iv):}
	We further simplify the asymptotic expression for $\langle \sqrt{n} J_n^1(u)\rangle$. 
	In \eqref{eqn:J-quadratic-1}, we want to replace $\mu^n_\frac{s}{n}$ by $\mu^n_\frac{t}{n}$. 
	Observe that
	\begin{align*}
		&\quad \left|\frac{1}{n} \sum_{j=\left\lceil \frac{\tau_n+L_n}{\tilde{L}}\right\rceil +1}^{\lfloor um\rfloor} \sum_{s,t=(j-1)\tilde{L}+r+1}^{j\tilde{L}}  Df(\mu^n_{\frac{t}{n}})\Cov(Z^j_{t,n}, Z^j_{s,n})\left[Df(\mu^n_{\frac{s}{n}}) - Df(\mu^n_{\frac{t}{n}})  \right]^T \right| \\
		&\leq \frac{C}{n} \sum_{t=1}^n \sum_{s=t-\tilde{L}}^{t+\tilde{L}} \rho^{|s-t|}\left\| Df(\mu^n_{\frac{s}{n}})- Df(\mu^n_{\frac{t}{n}})  \right\| \\
		&\leq \frac{C}{n} \sum_{t=1}^n \max_{|s-t|\leq \tilde{L}} \left\| \mu^n_{\frac{s}{n}}- \mu^n_{\frac{t}{n}}  \right\| \\
		&\leq \frac{C\tilde{L}_n}{n} \sum_{t=2}^n \left\| \mu^n_{\frac{t}{n}} - \mu^n_{\frac{t-1}{n}} \right\| \\
		&\leq C \tilde{L}_n n^{-\frac{1}{p}} \left( \sum_{t=2}^n \left\| \mu^n_{\frac{t}{n}} - \mu^n_{\frac{t-1}{n}} \right\|^p\right)^\frac{1}{p} \quad \leq C \tilde{L}_n n^{-\frac{1}{p}} \|\mu^n\|_{p-var}.
	\end{align*}
	By \eqref{eqn:pvar-finite}, and since $\tilde{L}_n \ll n^\epsilon$ the latter term is of order $\mathcal{O}(n^{\epsilon-\frac{1}{p}})=o(1)$ if $\epsilon$ is small enough.
	Hence,
	\begin{align}
		&\quad \langle \sqrt{n} J_n^1(u)\rangle \nonumber \\
		&= o_P(1) + \frac{1}{n} \sum_{j=\left\lceil \frac{\tau_n+L_n}{\tilde{L}}\right\rceil +1}^{\lfloor um\rfloor} \sum_{s,t=(j-1)\tilde{L}+r+1}^{j\tilde{L}}  Df(\mu^n_{\frac{t}{n}})\Cov(Z^j_{t,n}, Z^j_{s,n})Df(\mu^n_{\frac{t}{n}})^T \nonumber \\
		&= o_P(1) + \frac{1}{n} \sum_{j=\left\lceil \frac{\tau_n+L_n}{\tilde{L}}\right\rceil +1}^{\lfloor um\rfloor} \sum_{t=(j-1)\tilde{L}+r+1}^{j\tilde{L}} Df(\mu^n_{\frac{t}{n}}) \left[\sum_{s=(j-1)\tilde{L}+r+1}^{j\tilde{L}} \Cov(Z^j_{t,n}, Z^j_{s,n})\right]Df(\mu^n_{\frac{t}{n}})^T. \label{eqn:J-quadratic-2}
	\end{align}
	
	The covariance term may be simplified as follows. 
	Define
	\begin{align*}
		\Sigma_n(u) &= \sum_{h=-\infty}^\infty \Cov(G_n(u, \beps_h), G_n(u, \beps_0)) \quad \in \R^{d\times d},\qquad u\in[0,1].
	\end{align*}
	Note that \eqref{eqn:ergodic} implies that $\|\Cov(G_n(u, \beps_h), G_n(u, \beps_0))\| \leq C \rho^{|h|}$ for some constant $C$, where $\|\cdot\|$ denotes an arbitrary matrix norm. 
	In particular, $\Sigma_n(u)$ is well defined since the series is finite.
	
	Now, for any $t\geq (\tau_n+L_n)$, let $\delta_{t,n} = \sup_{|s-t|\leq \tilde{L}} \|G_n(\frac{t}{n},\beps_0) - G_n(\frac{s}{n}, \beps_0)\|_{L_q}$.
	Since $q>2$, we find that for $(j-1)\tilde{L}+1 \leq t \leq j\tilde{L}$,
	\begin{align*}
		&\quad \left\|  \sum_{t=(j-1)\tilde{L}+r+1}^{j\tilde{L}} \left[ \Sigma_n(\tfrac{t}{n}) -  \sum_{s=(j-1)\tilde{L}+r+1}^{j\tilde{L}} \Cov(Z^j_{t,n}, Z^j_{s,n}) \right] \right\|\\
		& = \left\|  \sum_{t=(j-1)\tilde{L}+r+1}^{j\tilde{L}} \left[ \Sigma_n(\tfrac{t}{n}) -  \sum_{s=(j-1)\tilde{L}+r+1}^{j\tilde{L}}  \Cov(G_n(\tfrac{t}{n}, \beps_{t-s}), G_n(\tfrac{s}{n}, \beps_0)) \right] \right\|\\
		& \leq \left\|  \sum_{t=(j-1)\tilde{L}+r+1}^{j\tilde{L}} \left[ \Sigma_n(\tfrac{t}{n}) -  \sum_{s=(j-1)\tilde{L}+r+1}^{j\tilde{L}}  \Cov(G_n(\tfrac{t}{n}, \beps_{t-s}), G_n(\tfrac{t}{n}, \beps_0)) \right] \right\| + \sum_{t=(j-1)\tilde{L}+r+1}^{j\tilde{L}} C \tilde{L} \delta_{t,n}\\
		& = \left\|  \sum_{t=(j-1)\tilde{L}+r+1}^{j\tilde{L}} \sum_{\substack{[s\leq(j-1)\tilde{L}+r]\\ \vee\, [s>j\tilde{L}]}}  \Cov(G_n(\tfrac{t}{n}, \beps_{t-s}), G_n(\tfrac{t}{n}, \beps_0)) \right\| + C \sum_{t=(j-1)\tilde{L}+r+1}^{j\tilde{L}} \tilde{L} \delta_{t,n}\\
		& \leq   \sum_{t=(j-1)\tilde{L}+r+1}^{j\tilde{L}}\sum_{\substack{[s\leq(j-1)\tilde{L}+r]\\ \vee\, [s>j\tilde{L}]}}  C \rho^{|t-s|}  + C\sum_{t=(j-1)\tilde{L}+r+1}^{j\tilde{L}} \tilde{L} \delta_{t,n} \\
		& \leq  C \sum_{t=(j-1)\tilde{L}+r+1}^{j\tilde{L}}  [ \rho^{|t-(j-1)\tilde{L}|} + \rho^{|t-j\tilde{L}|}]  + C \sum_{t=(j-1)\tilde{L}+r+1}^{j\tilde{L}} \tilde{L} \delta_{t,n} \\
		& \leq  C  + C\sum_{t=(j-1)\tilde{L}+r+1}^{j\tilde{L}} \tilde{L} \delta_{t,n}.
	\end{align*}
	Therefore,
	\begin{align*}
		&\quad \left\| \langle \sqrt{n} J_n^1(u)\rangle - \frac{1}{n} \sum_{j=\left\lceil \frac{\tau_n+L_n}{\tilde{L}}\right\rceil +1}^{\lfloor um\rfloor} \sum_{t=(j-1)\tilde{L}+r+1}^{j\tilde{L}} Df(\mu^n_{\frac{t}{n}}) \Sigma_n(\tfrac{t}{n}) Df(\mu^n_{\frac{t}{n}})^T \right\| \\
		&\leq o_P(1) + C \frac{\tilde{L}_n}{n} \sum_{t=\tau_n+L_n}^n \delta_{t,n} + C \frac{m_n}{n} \\
		&\leq o_P(1) + C \frac{\tilde{L}_n^2}{n} \sum_{t=\tau_n+L_n}^n \|\mu^n_{\frac{t}{n}} - \mu^n_{\frac{t-1}{n}} \| + C \frac{m_n}{n}\\
		&\leq o_P(1) +  C \tilde{L}_n^2 n^{-\frac{1}{p}} \|\mu^n\|_{p-var} + C \frac{m_n}{n}.
	\end{align*}
	The latter term tends to zero upon choosing $\epsilon$ small enough, since $\tilde{L}_n = \mathcal{O}(n^\epsilon)$.
	Note also that $m_n/n\to 0$.
	Finally, since $r_n \ll \tilde{L}_n$, we have shown that 
	\begin{align}
		\langle \sqrt{n} J_n^1(u)\rangle 
		&= o_P(1) + \frac{1}{n} \sum_{j=\left\lceil \frac{\tau_n+L_n}{\tilde{L}}\right\rceil +1}^{\lfloor um\rfloor} \sum_{t=(j-1)\tilde{L}+r+1}^{j\tilde{L}} Df(\mu^n_{\frac{t}{n}}) \Sigma_n(\tfrac{t}{n}) Df(\mu^n_{\frac{t}{n}})^T \nonumber\\
		&= o_P(1) + \frac{1}{n}  \sum_{t=1}^{\lfloor un\rfloor} Df(\mu^n_{\frac{t}{n}}) \Sigma_n(\tfrac{t}{n}) Df(\mu^n_{\frac{t}{n}})^T, \label{eqn:J-quadratic-3}
	\end{align}
	where the last step holds because $\tau_n+L_n\ll n$ and due to the boundedness of $Df(\mu^n_u)$ and $\Sigma_n(u)$, which is in turn implied by \eqref{eqn:ergodic} and \eqref{eqn:pvar-finite} for $q>2$.
	Note also that the $o_P(1)$ term in \eqref{eqn:J-quadratic-3} vanishes uniformly in $u\in[0,1]$.
	
	Analogously, we may show that
	\begin{align*}
		\langle \sqrt{n} J_n^2(u)\rangle 
		&= o_P(1) + \frac{1}{n} \sum_{j=\left\lceil \frac{\tau_n+L_n}{\tilde{L}}\right\rceil +1}^{\lfloor um\rfloor} \sum_{t=(j-1)\tilde{L}+1}^{(j-1)\tilde{L}+r} Df(\mu^n_{\frac{t}{n}}) \Sigma_n(\tfrac{t}{n}) Df(\mu^n_{\frac{t}{n}})^T  \longrightarrow 0,
	\end{align*}
	using again the boundedness, and the fact that $r_n\ll \tilde{L}_n$ such that the number of summands is asymptotically negligible. 
	This implies that $\sup_{u\in[0,1]} |J_n^2(u)| \pconv 0$ by the Burkholder-Davis-Gundy inequality for martingales.
	
	Observe that the function $u\mapsto \Sigma_n(u)$ has bounded $p$-variation, see \eqref{eqn:sigma-pvar}.
	Moreover, $u\mapsto Df(\mu^n_u)$ has bounded $p$-variation by assumptions \eqref{eqn:pvar-finite} and \eqref{eqn:ass-f}.
	Since both $\|\Sigma_n(u)\|$ and $\|Df(\mu^n_u)\|$ are bounded, we find that $u\mapsto Df(\mu^n_u)^T \Sigma_n(u) Df(\mu^n_u)$ has bounded $p$-variation.
	Thus, as in Lemma \ref{lem:I1},
	\begin{align*}
		\left|\frac{1}{n}  \sum_{t=1}^{\lfloor un\rfloor} Df(\mu^n_{\frac{t}{n}}) \Sigma_n(\tfrac{t}{n}) Df(\mu^n_{\frac{t}{n}})^T - \int_0^u Df(\mu^n_v)\Sigma_n(v)Df(\mu^n_v)^T\, dv\right| =\mathcal{O}(n^{-\frac{1}{p}}).
	\end{align*}
	Assumption \eqref{eqn:GnG} implies that $\mu^n_u\to \mu_u$ and $\Sigma_n(u)\to \Sigma(u)$ uniformly in $u\in[0,1]$. 
	The latter uniform convergence holds because $\Cov(G_n(u,\beps_0), G_n(u,\beps_h))\to \Cov(G(u,\beps_0), G(u,\beps_h))$ uniformly in $u$, and $\sup_{u\in[0,1]}\|\Cov(G_n(u,\beps_0), G_n(u,\beps_h))\|\leq C \rho^{|h|}$ by virtue of \eqref{eqn:ergodic}.
	Hence,
	\begin{align*}
		\sup_{u\in[0,1]} \left|\langle \sqrt{n} J_n^1(u)\rangle - \int_0^u Df(\mu_v)\Sigma(v)Df(\mu_v)^T\, dv\right| = o_P(1).
	\end{align*}
			
	\underline{Step (v):}
	Having at hand the limit of $\langle \sqrt{n} J_n^1(u)\rangle$, a functional central limit theorem for $\sqrt{n}J_n^1(u)$ holds if we can verify Lyapunov's condition.
	We have, for any $2<\delta<q$,
	\begin{align*}
		&\quad \sum_{j=\left\lceil \frac{\tau_n+L_n}{\tilde{L}}\right\rceil +1}^{m} \E \left[ \|Y^*_{j,n}/\sqrt{n}\|^\delta \,| \F^1_{j-1,n}  \right] \\
		&=  n^{-\frac{\delta}{2}} \sum_{j=\left\lceil \frac{\tau_n+L_n}{\tilde{L}}\right\rceil +1}^{m} \E \left[ \left\| \sum_{t=(j-1)\tilde{L}+r+1}^{j\tilde{L}} \chi_{t-L,n} Df(\hat{\mu}_{t-L,n}) Z_{t,n}^j \right\|^\delta \,\Bigg| \F^1_{j-1,n}  \right] \\
		&\leq C_f^\delta n^{-\frac{\delta}{2}} \sum_{j=\left\lceil \frac{\tau_n+L_n}{\tilde{L}}\right\rceil +1}^{m}  \E \left[ \left(\sum_{t=(j-1)\tilde{L}+r+1}^{j\tilde{L}} \left\|   Z_{t,n}^j \right\|\right)^\delta \,\Bigg| \F^1_{j-1,n}  \right] \\
		& \leq C_f^\delta n^{-\frac{\delta}{2}} \tilde{L}_n^{\delta-1} \sum_{j=\left\lceil \frac{\tau_n+L_n}{\tilde{L}}\right\rceil +1}^{m} \sum_{t=(j-1)\tilde{L}+r+1}^{j\tilde{L}} \E  \left\|  Z_{t,n}^j \right\|^\delta  \\
		& \leq C_f^\delta n^{-\frac{\delta}{2}} \tilde{L}_n^{\delta-1} \sum_{t=1}^{n} \E  \left\| G_n(\tfrac{t}{n}, \beps_0) \right\|^\delta  \\
		&= \mathcal{O}\left(n^{1-\frac{\delta}{2}} \tilde{L}_n^{\delta-1}\right),
	\end{align*}
	where we used that $\|G_n(u,\beps_0)\|_{L_q}$ is uniformly bounded.
	By choosing $\epsilon$ sufficiently small, and using $\tilde{L}_n\leq n^\epsilon$, the latter term tends to zero.
	Hence, the functional central limit theorem for martingales, e.g.\ \citep[Thm. VIII.3.33]{Jacod2003}, yields
	\begin{align*}
		\sqrt{n} J_n^1(u) \wconv B\left(\int_0^u Df(\mu_s)\Sigma(s)Df(\mu_s)^T \, ds\right).
	\end{align*}
	In the previous steps, we have established that $\sup_{u\in[0,1]} |I_n^3(u) - J_n^1(u)| = o_P(\sqrt{n})$, which yields the desired result.
\end{proof}

Theorem \ref{thm:CLT-integrated} is now a direct consequence of Lemmas \ref{lem:I1}, \ref{lem:I2-L}, and \ref{lem:I3}.

\subsection{Bootstrap inference}

Theorem \ref{thm:Q-estimation} is a special case of the following more general theorem about consistent estimation of the limiting variance. 
In contrast to the central limit theorem for $M_n(u)$, here, we only require the weaker assumption
\begin{align}
\sum_{t=\tau_n}^n \|\hat{\mu}_{t,n} - \mu^n_{\frac{t}{n}}\|^2  &= o(n^{1-\kappa}), \tag{C.5*}\label{eqn:ass-mu-1-weak}
\end{align} 
for some $\kappa\in(0,\frac{1}{2}]$.

\begin{theorem}\label{thm:bootstrap-new}
	Let the conditions of Theorem \ref{thm:CLT-integrated} hold for some $q>4$, replacing \eqref{eqn:ass-mu-1} by \eqref{eqn:ass-mu-1-weak}. 
	Choose some $b_n\to\infty$ such that 
	\begin{align*}
		 b_n\ll n^{\frac{2}{3} \frac{1}{\max(p,2)}}, \quad b_n \ll n^\frac{\kappa q-2}{q+2}.
	\end{align*}
	Then, as $n\to\infty$,
	\begin{align*}
		Q_n(u) 
		&= \frac{1}{n} \sum_{t=\tau_n+L_n}^{\lfloor n u \rfloor-b_n} \frac{1}{b_n} \left[ Df(\hat{\mu}_{t-L,n})\sum_{i=1}^{b_n}  (X_{t+i,n} - \hat{\mu}_{t-L,n})  \right]^2\\
		&\pconv Q(u) = \int_0^u Df(\mu_v) \Sigma(v) Df(\mu_v)^T\, dv.
	\end{align*}
	The convergence holds uniformly in $u\in[0,1]$ since $Q_n$ is monotone.
\end{theorem}

\begin{proof}[Proof of Theorem \ref{thm:bootstrap-new}]
	\underline{Step (i):}
	As a first step, we are concerned with the estimators $\hat{\mu}_{t-L, n}$.
	We decompose 
	\begin{align*}
		&\quad \sum_{i=1}^{b_n} Df(\hat{\mu}_{t-L,n}) (X_{t+i,n} - \hat{\mu}_{t-L,n})  \\
		&= \sum_{i=1}^{b_n} \Bigg\{Df(\mu^n_\frac{t}{n})  \left[G_n(\tfrac{t}{n},\beps_{t+i}) - \mu^n_{t,n}\right] +  \left[Df(\hat{\mu}_{t-L,n}) -Df(\mu^n_\frac{t}{n}) \right] \left[ G_n(\tfrac{t}{n}, \beps_{t+i}) - \mu^n_{\frac{t}{n}} \right] \\
		&\qquad + Df(\hat{\mu}_{t-L,n})\left[ G_n(\tfrac{t+i}{n}, \beps_{t+i}) - G_n(\tfrac{t}{n}, \beps_{t+i}) \right]
		+ Df(\hat{\mu}_{t-L,n})\left[ \mu^n_{\frac{t}{n}} - \hat{\mu}_{t-L,n} \right] \Bigg \}.
	\end{align*}
	With probability tending to one, it holds that $\hat{\mu}_{t-L,n}\in\mathcal{M}_\delta$ for all $t=\tau_n+L_n,\ldots, n$ by virtue of \eqref{eqn:ass-mu-2}. 
	On this event, we may exploit the regularity and boundedness of $f$ such that 
	\begin{align*}
		&\quad \left|\sum_{i=1}^{b_n} Df(\hat{\mu}_{t-L,n}) (X_{t+i,n} - \hat{\mu}_{t-L,n})  -  \sum_{i=1}^{b_n} Df(\mu^n_\frac{t}{n}) \left[G_n(\tfrac{t}{n},\beps_{t+i}) - \mu^n_{t,n}\right] \right| \\
		&\leq C \sum_{i=1}^{b_n} Z_{t,i,n} \|\hat{\mu}_{t-L,n} - \mu^n_{\frac{t}{n}}\| + \|  G_n(\tfrac{t+i}{n}, \beps_{t+i}) - G_n(\tfrac{t}{n}, \beps_{t+i}) \|\quad =: \Delta_{t,n}, 
	\end{align*}
	for $Z_{t,i,n} = \|G_n(\tfrac{t}{n}, \beps_{t+i}) - \mu^n_{\frac{t}{n}} \| + 1$.
	Now, via the inequality $|a^2-b^2|\leq |a-b| |a+b|$, we obtain
	\begin{align*}
		&\Bigg|\frac{1}{n} \sum_{t=\tau_n+L_n}^{\lfloor n u \rfloor-b_n} \frac{1}{b_n}\left[\sum_{i=1}^{b_n}Df(\hat{\mu}_{t-L,n}) (X_{t+i,n} - \hat{\mu}_{t-L,n}) \right]^2 \\
		&\qquad- \frac{1}{n} \sum_{t=\tau_n+L_n}^{\lfloor n u \rfloor-b_n} \frac{1}{b_n}\left[\sum_{i=1}^{b_n} Df(\mu^n_\frac{t}{n}) \left[G_n(\tfrac{t}{n},\beps_{t+i}) - \mu^n_{t,n}\right] \right]^2\Bigg| \\
		&\leq  \frac{1}{n} \frac{1}{b_n} \sum_{t=\tau_n+L_n}^{\lfloor n u \rfloor-b_n} \Delta_{t,n} \left[ \Delta_{t,n} + \left|\sum_{i=1}^{b_n} Df(\mu^n_\frac{t}{n}) \left[G_n(\tfrac{t}{n},\beps_{t+i}) - \mu^n_{t,n}\right] \right| \right] \\
		&\leq \frac{1}{n} \frac{1}{b_n} \sum_{t=\tau_n+L_n}^{\lfloor n u \rfloor-b_n} \Delta_{t,n}^2 
		+  \sqrt{\frac{1}{n} \frac{1}{b_n} \sum_{t=\tau_n+L_n}^{\lfloor n u \rfloor-b_n} \Delta_{t,n}^2} \sqrt{\frac{1}{n} \sum_{t=\tau_n+L_n}^{\lfloor n u \rfloor-b_n} \frac{1}{b_n}\left[\sum_{i=1}^{b_n} Df(\mu^n_\frac{t}{n})\left[G_n(\tfrac{t}{n},\beps_{t+i}) - \mu^n_{t,n}\right] \right]^2}.
	\end{align*}
	It thus suffices to show that
	\begin{align}
		\frac{1}{n} \frac{1}{b_n} \sum_{t=\tau_n+L_n}^{\lfloor n u \rfloor-b_n} \Delta_{t,n}^2 \pconv 0, \label{eqn:bootstrap-delta}
	\end{align}
	and
	\begin{align}
		\frac{1}{n} \sum_{t=\tau_n+L_n}^{\lfloor n u \rfloor-b_n} \frac{1}{b_n}\left[\sum_{i=1}^{b_n}Df(\mu^n_\frac{t}{n}) \left[G_n(\tfrac{t}{n},\beps_{t+i}) - \mu^n_{t,n}\right] \right]^2 \pconv \int_0^u Df(\mu_v)^T\Sigma(v) Df(\mu_v)\, dv.\label{eqn:bootstrap-int}
	\end{align}
	
	\underline{Step (ii):}
	Regarding the term \eqref{eqn:bootstrap-delta}, we find that
	\begin{align}
		&\quad \frac{1}{n} \frac{1}{b_n} \sum_{t=\tau_n+L_n}^{\lfloor n u \rfloor-b_n} \Delta_{t,n}^2 \nonumber \\
		&\leq \frac{C}{n} \frac{1}{b_n} \sum_{t=\tau_n+L_n}^{\lfloor n u \rfloor-b_n} \left\{ \left[\sum_{i=1}^{b_n} Z_{t,i,n} \|\hat{\mu}_{t-L,n} - \mu^n_{\frac{t}{n}}\|\right]^2 + \left[ \sum_{i=1}^{b_n} \|  G_n(\tfrac{t+i}{n}, \beps_{t+i}) - G_n(\tfrac{t}{n}, \beps_{t+i}) \|\right]^2 \right\} \nonumber \\
		&\leq  \frac{C}{n} \sum_{t=\tau_n+L_n}^{\lfloor n u \rfloor-b_n} \left\{ \|\hat{\mu}_{t-L,n} - \mu^n_{\frac{t}{n}}\|^2 \sum_{i=1}^{b_n} |Z_{t,i,n}|^2  + \sum_{i=1}^{b_n} \|  G_n(\tfrac{t+i}{n}, \beps_{t+i}) - G_n(\tfrac{t}{n}, \beps_{t+i}) \|^2\right\} \nonumber \\
		&\leq  \frac{C}{n} \sum_{t=\tau_n+L_n}^{\lfloor n u \rfloor-b_n} \left\{ \|\hat{\mu}_{t-L,n} - \mu^n_{\frac{t}{n}}\|^2 b_n \max_{i=1,\ldots, b_n}|Z_{t,i,n}|^2  + \sum_{i=1}^{b_n} \left[ \sum_{j=1}^i\|  G_n(\tfrac{t+j}{n}, \beps_{t+i}) - G_n(\tfrac{t+j-1}{n}, \beps_{t+i}) \|\right]^2\right\} \nonumber \\
		\begin{split}
		&\leq \frac{C b_n}{n} \max_{\substack{t=\tau_n+L_n,\ldots, n\\ i=1,\ldots, b_n}} |Z_{t,i,n}|^2 \sum_{t=\tau_n+L_n}^{\lfloor n u \rfloor-b_n} \|\hat{\mu}_{t-L,n} - \mu^n_{\frac{t}{n}}\|^2  \\
		&\quad + \frac{C b_n}{n} \sum_{t=\tau_n+L_n}^{\lfloor n u \rfloor-b_n} \sum_{i,j=1}^{b_n} \|  G_n(\tfrac{t+j}{n}, \beps_{t+i}) - G_n(\tfrac{t+j-1}{n}, \beps_{t+i}) \|^2.
		\end{split} \label{eqn:bootstrap-delta-1}
	\end{align}
	The second term may be bounded by noting that, for $p\leq 2$,
	\begin{align*}
		&\quad \frac{ b_n}{n} \sum_{t=\tau_n+L_n}^{\lfloor n u \rfloor-b_n} \sum_{i,j=1}^{b_n} \E\|  G_n(\tfrac{t+j}{n}, \beps_{t+i}) - G_n(\tfrac{t+j-1}{n}, \beps_{t+i}) \|^2 \\
		&\leq \frac{ b_n^3}{n} \sum_{t=\tau_n+L_n}^{\lfloor n u \rfloor}  \|  G_n(\tfrac{t}{n}, \beps_{0}) - G_n(\tfrac{t-1}{n}, \beps_{0}) \|_{L_2}^2 \\
		&\leq \frac{ b_n^3}{n} \sum_{t=1}^{n}  \|  G_n(\tfrac{t}{n}, \beps_{0}) - G_n(\tfrac{t-1}{n}, \beps_{0}) \|_{L_q}^p \qquad \leq \frac{ b_n^3}{n} \|G_n\|_{p-var}^p.
	\end{align*}
	In the second inequality, we used that $q>2$, $p<2$, and that $\|G_n(u, \beps_0)\|_{L_q}$ is uniformly bounded due to \eqref{eqn:pvar-finite}.
	If $p>2$, we find that
	\begin{align*}
		\sum_{t=1}^{n}  \|  G_n(\tfrac{t}{n}, \beps_{0}) - G_n(\tfrac{t-1}{n}, \beps_{0}) \|_{L_q}^2 
		&\leq n^{1-\frac{2}{p}} \left(\sum_{t=1}^{n}  \|  G_n(\tfrac{t}{n}, \beps_{0}) - G_n(\tfrac{t-1}{n}, \beps_{0}) \|_{L_q}^p \right)^\frac{2}{p} \\
		&\leq n^{1-\frac{2}{p}} \|G_n\|_{p-var}^2,
	\end{align*}
	such that the second term in \eqref{eqn:bootstrap-delta-1} is of order $\mathcal{O}(b_n^3n^{-\frac{2}{\max(p,2)}})$.
	
	Concerning the first term in \eqref{eqn:bootstrap-delta-1}, we note that 
	\begin{align*}
		\max_{\substack{t=\tau_n+L_n,\ldots, n}} |Z_{t,i,n}| 
		&\leq |Z_{1,i,n}| + \sum_{t=2}^n |Z_{t,i,n} - Z_{t-1,i,n}|, \\ 
		\left\|\max_{\substack{t=\tau_n+L_n,\ldots, n}} |Z_{t,i,n}|\right\|_{L_q} 
		&\leq 1 + 2\|G_n(\tfrac{1}{n}, \beps_{t+i})\|_{L_q} + 2\sum_{t=2}^n \|G(\tfrac{t}{n},\beps_{t+i}) - G(\tfrac{t-1}{n},\beps_{t+i})\|_{L_q} \\
		&\leq 1+ 2 C_G + n^{1-\frac{1}{p}} \|G_n\|_{p-var} \\
		&\leq 2n^{1-\frac{1}{p}} C_G + 1.
	\end{align*}
	Since $q>2$, this yields
	\begin{align*}
		\left\|\max_{\substack{t=\tau_n+L_n,\ldots, n}} |Z_{t,i,n}|^2\right\|_{L_{\frac{q}{2}}}  = \left\|\max_{\substack{t=\tau_n+L_n,\ldots, n}} |Z_{t,i,n}|\right\|_{L_{q}}^2
		&\leq C\,n^{2-\frac{2}{p}},
	\end{align*}
	for some finite $C$.
	An alternative upper bound, which is sharper for $p>2$, is given by
	\begin{align*}
		\left\|\max_{\substack{t=\tau_n+L_n,\ldots, n}} |Z_{t,i,n}|^2\right\|_{L_{\frac{q}{2}}} 
		&\leq \sum_{t=1}^n \|Z_{t,i,n}\|_{L_q}^2 \leq C n,
	\end{align*}
	such that 
	\begin{align*}
		\left\|\max_{\substack{t=\tau_n+L_n,\ldots, n}} |Z_{t,i,n}|^2\right\|_{L_{\frac{q}{2}}} 
		&\leq  C n^{2 - \frac{2}{\min(p,2)}}.
	\end{align*}
	
	Thus, the union bound and Markov's inequality yield that 
	\begin{align*}
		\max_{i=1,\ldots, b_n} \max_{t=1,\ldots, n} |Z_{t,i,n}|^2 =\mathcal{O}_P( |n^{2-\frac{2}{\min(p,2)}} b_n|^{\frac{2}{q}}).
	\end{align*}
	Furthermore, as in the proof of Lemma \ref{lem:I2-L}, 
	\begin{align*}
		\sum_{t=\tau_n+L_n}^{\lfloor n u \rfloor-b_n} \|\hat{\mu}_{t-L,n} - \mu^n_{\frac{t}{n}}\|^2 
		&\leq o_P(n^{1-\kappa}) + \mathcal{O}_P(L_n^{\min(p,2)} n^{1-\frac{2}{\max(2,p)}}).
	\end{align*}
	Plugging this into \eqref{eqn:bootstrap-delta-1}, we have thus shown that
	\begin{align*}
		\frac{1}{n} \frac{1}{b_n} \sum_{t=\tau_n+L_n}^{\lfloor n u \rfloor-b_n} \Delta_{t,n}^2 
		&= \frac{b_n}{n} \mathcal{O}_P(|n^{2-\frac{2}{\min(p,2)}}b_n|^\frac{2}{q}) \left[o_P(n^{1-\kappa}) + \mathcal{O}_P(L_n^{\min(p,2)} n^{1-\frac{2}{\max(p,2)}}) \right] + \mathcal{O}_P\left(b_n^3 n^{ -\frac{2}{\max(p,2)} }\right) 
	\end{align*}
	For $p\leq 2$, this bound reduces to 
	\begin{align*}
		\frac{1}{n} \frac{1}{b_n} \sum_{t=\tau_n+L_n}^{\lfloor n u \rfloor-b_n} \Delta_{t,n}^2 
		&= \mathcal{O}_P \left( n^{\frac{4}{q}(1-\frac{1}{p})-\kappa} b_n^{1+\frac{2}{q}}   \right)
		+ \mathcal{O}_P \left( n^{\frac{4}{q}(1-\frac{1}{p})-\kappa} b_n^{1+\frac{2}{q}} n^{\kappa-1}L_n^p   \right)
		+ \mathcal{O}_P\left(b_n^3 n^{ -1 }\right) \\
		&\leq \mathcal{O}_P \left( n^{\frac{2}{q}-\kappa} b_n^{1+\frac{2}{q}}   \right)
		+ \mathcal{O}_P \left( n^{\frac{2}{q}-\kappa} b_n^{1+\frac{2}{q}} n^{\kappa-1}L_n^p   \right)
		+ \mathcal{O}_P\left(b_n^3 n^{ -1 }\right),
	\end{align*}
	which tends to zero by the assumptions $b_n\ll n^{-\frac{1}{3}}$, $b_n\ll n^{\frac{\kappa q -2}{q+2}}$, and $L_n \ll n^\frac{1}{2p}$.
	For $p>2$, we find that
	\begin{align*}
		\frac{1}{n} \frac{1}{b_n} \sum_{t=\tau_n+L_n}^{\lfloor n u \rfloor-b_n} \Delta_{t,n}^2 
		&= \mathcal{O}_P\left( n^{\frac{2}{q}-\kappa} b_n^{1+\frac{2}{q}} \right) 
		+ \mathcal{O}_P\left( n^{\frac{2}{q}-\kappa} b_n^{1+\frac{2}{q}} n^{\kappa-1} L_n^2\right)
		+ \mathcal{O}_P\left(b_n^3 n^{ -\frac{2}{p} }\right),
	\end{align*}
	which also tends to zero as $n\to\infty$.

	\underline{Step (iii):}
	It remains to study the term \eqref{eqn:bootstrap-int}.
	Introduce the random variables
	\begin{align*}
		\zeta_{t,n}  &= \frac{1}{b_n}\left[\sum_{i=1}^{b_n} Df(\mu^n_\frac{t}{n}) \left[G_n(\tfrac{t}{n},\beps_{t+i}) - \mu^n_{t,n}\right] \right]^2,
	\end{align*}
	such that we are led to study the term $\frac{1}{n} \sum_{t=\tau_n+L_n}^{\lfloor n u \rfloor-b_n} \zeta_{t,n}$.
	To bound the covariances of the $\zeta_{t,n}$, we write $\zeta_{t,n} = \zeta_{t,n}(\beps_{t+b_n})$, such that \eqref{eqn:ergodic} may be exploited.
	Let $b\geq 0$ be an integer, and recall that
	\begin{align*}
		\beps^*_{t+b_n, b+ b_n} &=(\epsilon_{t+b_n},\ldots, \epsilon_{t-b+1}, \epsilon^*_{t-b},\ldots) \in\R^\infty, \\
		\beps^*_{t+i, b +i} &= (\epsilon_{t+i},\ldots, \epsilon_{t-b+1}, \epsilon^*_{t-b },\ldots) \in\R^\infty,
	\end{align*}
	where the $\epsilon^*_j$ are an independent copy of the $\epsilon_j$.
	Then
	\begin{align*}
		&\quad \|\zeta_{t,n}(\beps_{t+b_n}) - \zeta_{t,n}(\beps^*_{t+b_n, b+b_n})\|_{L_2} \\
		&=\frac{1}{b_n} \left\| \left[\sum_{i=1}^{b_n}Df(\mu^n_\frac{t}{n}) \left[G_n(\tfrac{t}{n},\beps_{t+i}) - \mu^n_{t,n}\right] \right]^2 - \left[\sum_{i=1}^{b_n} Df(\mu^n_\frac{t}{n}) \left[G_n(\tfrac{t}{n},\beps^*_{t+i, b+i}) - \mu^n_{t,n}\right] \right]^2 \right\|_{L_2} \\
		&\leq \frac{C}{b_n} \left\| \sum_{i=1}^{b_n}  \left\|G_n(\tfrac{t}{n},\beps_{t+i}) - G_n(\tfrac{t}{n},\beps^*_{t+i, b+i}) \right\| \sum_{j=1}^{b_n} \left[  \left\|G_n(\tfrac{t}{n},\beps_{t+j})\right\| + \left\|G_n(\tfrac{t}{n},\beps^*_{t+j, b+j}) \right\| \right] \right\|_{L_2}
		\intertext{via the inequality $|x^2-y^2| \leq |x-y||x+y|$,}
		& \leq \frac{C}{b_n}  \sum_{i=1}^{b_n}  \left\|G_n(\tfrac{t}{n},\beps_{t+i}) - G_n(\tfrac{t}{n},\beps^*_{t+i, b+i}) \right\|_{L_4} \left\| \sum_{j=1}^{b_n} \left\|G_n(\tfrac{t}{n},\beps_{t+j})\right\|+ \left\|G_n(\tfrac{t}{n},\beps^*_{t+i, b+i}) \right\| \right\|_{L_4}
		\intertext{by Hölder's inequality,}
		& \leq \frac{C}{b_n}  \sum_{i=1}^{b_n}  \left\|G_n(\tfrac{t}{n},\beps_{t+i}) - G_n(\tfrac{t}{n},\beps^*_{t+i, b+i}) \right\|_{L_q}  \sum_{j=1}^{b_n} \left[\left\|G_n(\tfrac{t}{n},\beps_{t+j})\right\|_{L_q} + \left\|G_n(\tfrac{t}{n},\beps^*_{t+j, b+j}) \right\|_{L_q} \right]\\
		&\leq C b_n \rho^{b}.
	\end{align*}
	In the last step, we used $q>4$, the boundedness of $\|G_n(u,\beps_0)\|_{L_q}$ due to \eqref{eqn:pvar-finite}, and the ergodicity \eqref{eqn:ergodic}.
	
	Moreover, it holds that
	\begin{align*}
		\|\zeta_{t,n}\|_{L_2} 
		&=  \left\| \frac{1}{\sqrt{b_n}} \sum_{i=1}^{b_n} Df(\mu^n_\frac{t}{n}) \left[G_n(\tfrac{t}{n},\beps_{t+i}) - \mu^n_{t,n}\right] \right\|_{L_4}^2 \\
		&\leq \frac{1}{b_n} \left[ \sum_{i=1}^{b_n} \left\| Df(\mu^n_\frac{t}{n}) \left[G_n(\tfrac{t}{n},\beps_{t+i}) - \mu^n_{t,n}\right]  \right\|_{L_4} \right]^2 \\
		&\leq C b_n,
	\end{align*}
	because $\|G_n(\frac{t}{n}, \beps_0)\|_{L_4}$ is bounded by virtue of \eqref{eqn:pvar-finite}, and $Df(\mu^n_{t/n})$ is bounded by virtue of \eqref{eqn:ass-f}.
	
	Now consider $s,t$ such that $s<t$.
	If $|s-t|>b_n$, then $\zeta_{s,n}$ is independent of $\beps^*_{t+b_n, |s-t|+b_n}$, such that
	\begin{align*}
		\left|\Cov(\zeta_{t,n}, \zeta_{s,n})\right| 
		&= \left|\Cov\left(\zeta_{t,n}(\beps_{t+b_n}) - \zeta_{t,n}(\beps^*_{t+b_n, |s-t|+b_n}), \zeta_{s,n}\right) \right| \\
		&\leq \|\zeta_{t,n}(\beps_{t+b_n}) - \zeta_{t,n}(\beps^*_{t+b_n, |s-t|+b_n})\|_{L_2} \|\zeta_{s,n}\|_{L_2}\\
		&\leq C b_n^2 \rho^{|s-t|}.
	\end{align*}
	These findings may be summarized as 
	\begin{align}
		\left|\Cov(\zeta_{t,n}, \zeta_{s,n})\right| \leq C b_n^2 \rho^{(|s-t|-b_n) \vee 0}. \label{eqn:bootstrap-zeta-cov}
	\end{align}
	In particular,
	\begin{align*}
		\Var \left( \frac{1}{n} \sum_{t=\tau_n+L_n}^{\lfloor n u \rfloor-b_n} \zeta_{t,n} \right) 
		&= \frac{1}{n^2}  \sum_{t=\tau_n+L_n}^{\lfloor n u \rfloor-b_n} \Cov(\zeta_{t,n}, \zeta_{s,n}) \\
		&\leq \frac{Cb_n^2}{n^2} \sum_{s,t=1}^n \rho^{(|s-t|-b_n) \vee 0} \\
		&\leq \frac{Cb_n^2}{n^2}  \sum_{t=1}^n \left[b_n + \frac{2}{1-\rho}\right] \quad = \mathcal{O}\left( \frac{b_n^3}{n} \right),
	\end{align*}
	which tends to zero by assumption.
	Thus, $\frac{1}{n} \sum_{t=\tau_n+L_n}^{\lfloor n u \rfloor-b_n} [\zeta_{t,n} - \E \zeta_{t,n}] \pconv 0$ as $n\to\infty$.
	
	\underline{Step (iv):}
	The expected value of $\zeta_{t,n}$ may be computed as
	\begin{align*}
		&\quad \E \zeta_{t,n} \\ 
		&= Df(\mu^n_{\frac{t}{n}}) \left[ \frac{1}{b_n} \sum_{i,j=1}^{b_n} \Cov\left(G_n(\tfrac{t}{n}, \beps_i), G_n(\tfrac{t}{n}, \beps_j) \right) \right] Df(\mu^n_{\frac{t}{n}})^T \\
		&= Df(\mu^n_{\frac{t}{n}}) \left[ \frac{1}{b_n} \sum_{i=1}^{b_n} \sum_{j=-\infty}^\infty \Cov\left(G_n(\tfrac{t}{n}, \beps_i), G_n(\tfrac{t}{n}, \beps_j) \right) \right] Df(\mu^n_{\frac{t}{n}})^T + \mathcal{O}\left( \frac{1}{b_n} \sum_{i=1}^{b_n} [\rho^i + \rho^{b_n-i}] \right) \\
		&= Df(\mu^n_{\frac{t}{n}}) \left[\sum_{h=-\infty}^\infty \Cov\left(G_n(\tfrac{t}{n}, \beps_h), G_n(\tfrac{t}{n}, \beps_0) \right) \right] Df(\mu^n_{\frac{t}{n}})^T + \mathcal{O}(\tfrac{1}{b_n}) \\
		&= Df(\mu^n_{\frac{t}{n}}) \Sigma_n(\tfrac{t}{n}) Df(\mu^n_{\frac{t}{n}})^T + \mathcal{O}(\tfrac{1}{b_n}),
	\end{align*}
	using the geometric ergodicity \eqref{eqn:ergodic} in the second step.
	Note that the remainder term is of order $1/b_N$, uniformly in $t$ and $n$, because $\|Df(\mu^n_{\frac{t}{n}})\|\leq C_f$ for all $t$, and $n$ large enough.
	
	Just as in the proof of Lemma \ref{lem:I3}, step (iv) therein, we find that 
	\begin{align*}
		\left|\frac{1}{n} \sum_{t=\tau_n+L_n}^{\lfloor n u \rfloor-b_n}  \E \zeta_{t,n} - \int_0^u Df(\mu^n_{v}) \Sigma_n(v) Df(\mu^n_{v})^T\, dv\right| \to 0,
	\end{align*}
	as $n\to\infty$.
	Moreover, since $\|G_n(\frac{t}{n},\beps_0) - G(\frac{t}{n},\beps_0)\|_{L_2}\to 0$ as $n\to\infty$, uniformly in $t$, the dominated convergence theorem yields that 
	\begin{align*}
		\left\|\Sigma(\tfrac{t}{n}) - \sum_{h=-\infty}^\infty \Cov\left(G_n(\tfrac{t}{n}, \beps_h), G_n(\tfrac{t}{n}, \beps_0) \right) \right\| \to 0,
	\end{align*}
	uniformly in $t$.
	Furthermore, $\|Df(\mu^n_{t/n}) - Df(\mu_{t/n})\|\to 0$ uniformly in $t$, such that 
	\begin{align*}
		\frac{1}{n}\sum_{t=\tau_n+L_n}^{\lfloor n u \rfloor-b_n}\E \zeta_{t,n}
		&= \int_0^u Df(\mu^n_{v}) \Sigma_n(v) Df(\mu^n_{v})^T\, dv + o(1) \\
		&= \int_0^u Df(\mu_{v}) \Sigma(v) Df(\mu_{v})^T \, dv + o(1).
	\end{align*}
	Together with the variance bound derived in step (iii), this establishes \eqref{eqn:bootstrap-int} and thus completes the proof.	
\end{proof}

Theorem \ref{thm:Q-estimation} is a consequence of Theorem \ref{thm:bootstrap-new} for $\kappa=\frac{1}{2}$.
Regarding the bootstrap consistency, Theorem \ref{thm:bootstrap}, note that $Q_n(u)$ is precisely the variance process of $\widehat{M}_n(u)$.

\begin{proof}[Proof of Theorem \ref{thm:bootstrap}]
	Conditionally on $\mathbb{X}_n$, $\widehat{M}_n(u)$ is a zero mean Gaussian process with independent increments, and variance function $\Var(\widehat{M}_n(u)|\mathbb{X}_n) = Q_n(u)$ for $Q_n$ as in Theorem \ref{thm:bootstrap-new}.
	It may hence be represented as $\widehat{M}_n(u) = B(Q_n(u))$ for a standard Brownian motion $B$ which is independent of $\mathbb{X}$.
	The uniform convergence of $Q_n(u)\to Q(u) = \int_0^u Df(\mu_v) \Sigma(v) Df(\mu_v)^T\, dv$ thus yields $\sup_{u\in[0,1]} | B(Q_n(u))-B( Q(u))| \pconv 0$.
	This yields weak convergence of $\widehat{M}_n(u)$ in probability.
	That is, for any probability metric $d$ which metricizes weak convergence on the Skorokhod space, we have $d(\mathcal{L}(M), \mathcal{L}(\widehat{M}_n | \mathds{X}_n))\pconv 0$, where $\mathcal{L}(M)$ denotes the measure induced by the random element $M$. 
\end{proof}

\subsection{Properties of the tvVAR example}

To show existence of the tvVAR process of Example \ref{ex:VAR}, we want to show that $\|\prod_{j=1}^i A(u-\frac{j-1}{n})\|\leq C \rho^i$, where $\|\cdot\|$ denotes the Euclidean operator norm. 
If $\|A(u)\|\leq \rho$ for all $u\in[0,1]$, this conclusion is obvious. 
However, the latter condition is too restrictive for many applications.
For example, consider the univariate autoregressive process of order $m$ given by $X_{t,n} = \sum_{j=1}^m a_j(\frac{t}{n}) X_{t-j,n} + \epsilon_t$.
This process may be represented as a tvVAR process with state vector $Y_{t,n} = (X_{t,n},\ldots, X_{t-m+1,n})$, and autoregressive matrix
\begin{align}
	A(u) = \begin{pmatrix}
		a_1(u) & a_2(u) & a_3(u) &  \ldots & a_m(u) \\
		1 & 0 & 0 & \ldots & 0 \\
		0 & 1 & 0 & \ldots & 0 \\
		\vdots & \vdots & \vdots & \vdots &\vdots \\
		 0 &  \ldots & 0 & 1 & 0
	\end{pmatrix} \in \R^{m\times m}. \label{eqn:ar-matrix}
\end{align} 
A typical condition for the stability of $X_{t,n}$ is that the autoregressive polynomial $1-\sum_{j=1}^m a_j(u)z^j$ has no (complex) roots of absolute value smaller than $\frac{1}{\rho_0}>1$, see e.g.\ \cite[Sec.\ 3.1.3]{Giraud2015}.
Equivalently, the spectral radius of $A(u)$ is bounded by $\rho_0<1$. 
On the other hand, we have $\|A(u)\|\geq 1$ for all $u\in[0,1]$, irrespective of the coefficients $a_j(u)$.

In view of the previous example, we do not restrict the Euclidean operator norm to be smaller than one.
Instead, we impose an upper bound on the spectral radius of the matrices $A(u)$, i.e.\ $\rho(A(u))\leq \rho_0 <1$.
This comes at the price of requiring additional regularity of the mapping $u\mapsto A(u)$.

\begin{lemma}\label{lem:eigenvalue-product}
	Let $A_i\in\R^{d\times d}$, $i\in\N$, be a sequence of matrices such that $ \|A_i\| \leq A^* < \infty$, and $\sup_{i\in\N}\rho(A_i)= \rho_0 <1$.
	Assume furthermore that for some $p\geq 1$,
	\begin{align*}
		\|A\|_{p-var} = \left(\sup_{1\leq i_1 < i_2 < \ldots} \sum_{k=1}^\infty \| A_{i_{k+1}}-A_{i_k}\|^p\right)^{\frac{1}{p}} \leq A^* <\infty.
	\end{align*}
	Then for any $\rho>\rho_0$, there exists a $K=K(\rho,\rho_0, A^*)$ such that
	\begin{align*}
		\left\|\prod_{i=s+1}^{t} A_i \right\| \leq K \rho^{t-s},\qquad 0\leq s < t .
	\end{align*}
\end{lemma}
\begin{proof}[Proof of Lemma \ref{lem:eigenvalue-product}]
	We proceed as in the proof of \cite[Lemma 8]{Giraud2015}.
	For any $s,t\in\N_0, s\leq t$, we have
	\begin{align}
		\prod_{i=s+1}^t A_i 
		&= A_{s+1}^{t-s} + \sum_{k=s+1}^{t} A_{s+1}^{(k-1)-s} (A_k - A_{s+1}) \prod_{i=k+1}^t A_i, \nonumber \\
		\left\| \prod_{i=s+1}^t A_i \right\|
		&\leq \left\|A_{s+1}^{t-s}\right\| + (t-s) (A^*)^{t-s-1} \max_{s+1\leq k \leq t} \|A_k-A_{s+1}\|. \label{eqn:prod-norm-1}
	\end{align}
	Here, we employ the convention $\prod_{i=u}^v A_i = I$ if $u>v$, i.e.\ the empty product is the identity matrix.
	
	Now choose some $\rho_1 \in (\rho_0,\rho)$. 
	By virtue of \cite[Lemma 12]{moulines2005recursive}, there exists a universal $C=C(\rho_0, \rho_1, A^*)$ such that $\|A_{s+1}^{t-s}\| \leq C \rho_1^{t-s}$.
	We choose $l = l(\rho_0,\rho_1,A^*)\in\N$ such that $C\rho_1^l \leq \frac{\rho^l}{2}$, and let $\epsilon = \frac{\rho^l}{ 2 l (A^*)^{l-1} }$.
	Split the product into blocks of length $l$, i.e.\ we write
	\begin{align*}
		\prod_{i=s+1}^t A_i 
		&= \left[\prod_{j=1}^{\lfloor \frac{t-s}{l}\rfloor} \left( \prod_{i=s+(j-1)l+1}^{s+jl} A_i \right) \right] \left[ \prod_{i=s+\lfloor \frac{t-s}{l}\rfloor l +1}^t A_i \right], \\
		\left\|\prod_{i=s}^t A_i \right\|
		&\leq \left[\prod_{j=1}^{\lfloor \frac{t-s}{l}\rfloor} \left\| \prod_{i=s+(j-1)l+1}^{s+jl} A_i \right\| \right] (A^*)^l.
	\end{align*}
	Using \eqref{eqn:prod-norm-1}, we may further bound
	\begin{align*}
		\delta_j = \left\| \prod_{i=s+(j-1)l+1}^{s+jl} A_i \right\|
		&\leq 
		\begin{cases}
			\rho^l, & \max_{s+(j-1)l+1\leq k \leq s+jl} \|A_k-A_{s+(j-1)l+1}\| < \epsilon, \\
			(A^*)^l, & \text{otherwise}.
		\end{cases}
	\end{align*}
	There are at most $\kappa = \lceil\frac{(A^*)^p}{\epsilon^p} \rceil$ indices $j$ such that $\delta_j>\rho^l$, because
	\begin{align*}
		&\sum_{j=1}^\infty \mathds{1}\left(\max_{s+(j-1)l+1\leq k \leq s+jl} \|A_k-A_{s+(j-1)l+1}\| \geq \epsilon\right) \\
		&\leq \epsilon^{-p} \sum_{j=1}^\infty \left[\max_{s+(j-1)l+1\leq k \leq s+jl} \|A_k-A_{s+(j-1)l+1}\|^p \right]
		\leq \epsilon^{-p}\|A\|_{p-var}^p \leq \kappa.
	\end{align*}
	In particular,
	\begin{align*}
		\prod_{j=1}^{\lfloor \frac{t-s}{l}\rfloor} \left\| \prod_{i=s+(j-1)l+1}^{s+jl} A_i \right\|
		&\leq \rho^{l \lfloor \frac{t-s}{l}\rfloor} \rho^{-l\kappa} (A^*)^{l\kappa}, \\
		\left\|\prod_{i=s+1}^t A_i \right\|
		&\leq \rho^{t-s} \rho^{-l(\kappa+1)} (A^*)^{l(\kappa+1)} = \rho^{t-s} K,
	\end{align*}
	for some $K=K(\rho, \rho_0, A^*)$.
\end{proof}

In the literature, results similar to Lemma \ref{lem:eigenvalue-product} have been established to show existence of the tvVAR model. 
The existing results are formulated for matrices $A_i = A(t-\frac{i}{p})$ for some function $u\mapsto A(u)$. 
In particular, \cite[Proposition 13]{moulines2005recursive} requires $u\mapsto A(u)$ to be Hölder continuous, and Lemma 8 of \cite{Giraud2015} replaces Hölder continuity by some arbitrarily weak uniform continuity.
Proposition 2.4 of \cite{Dahlhaus2009a} requires $u\mapsto A(u)$ to be of bounded variation.
The bounded $p$-variation imposed in Lemma \ref{lem:eigenvalue-product} is weaker than the assumptions of \cite{moulines2005recursive} and \cite{Dahlhaus2009a}. 
Compared to \cite{Giraud2015}, we allow for discontinuities.
On the other hand, Giraud et al.\ only require uniform continuity, without further restrictions on the modulus of continuity, which also includes some irregular continuous cases where $\|A\|_{p-var}=\infty$ for all $p\geq 1$.

\begin{proposition}\label{prop:VAR}
	Let $G_n(u,\beps_t)$ as in Example \ref{ex:VAR}, with $\sup_u \|A(u)\|<\infty$ and $\sup_u \|B(u)\|<\infty$, and suppose that the matrices $A(u)$ have spectral radius $\rho(A(u))\leq\rho_0<1$ for all $u\in[0,1]$.
	\begin{enumerate}
		\item If $\|A\|_{p-var} < \infty$, $\|B\|_{p-var}<\infty$, and $\|\mu\|_{p-var}<\infty$ for some $p\geq 1$, then there exists some $C=C(A,B,\mu)$ such that $\|G_n\|_{p-var} + \sup_{u\in[0,1]} \|G_n(u,\beps_0)\|_{L_q} \leq C$ for all $n$.
		\item If $u\mapsto A(u)$, $u\mapsto B(u)$, and $u\mapsto \mu(u)$ are $\beta$-Hölder continuous in the Euclidean operator norm $\|\cdot\|$, then there exists $C_H=C_H(A,B,\mu)<\infty$ such that for all $n\in\N$, the mapping $u\mapsto G_n(u,\beps_t)$ is $\beta$-Hölder continuous in $L_q(P)$ with Hölder constant $C_H$.
	\end{enumerate} 
\end{proposition}
\begin{proof}[Proof of Proposition \ref{prop:VAR}]
	Note that $\E \|\epsilon_0\|^q<\infty$ by assumption.
	In the sequel, the factor $C=C(A,B,\mu)$ may vary from line to line.
	For any $u\in[0,1]$, we have
	\begin{align*}
		\|G_n(u,\beps_0)\|_{L_q} 
		&\leq C \sum_{i=0}^\infty \left\| \prod_{j=1}^i A(u-\tfrac{j-1}{n}) \right\| \, \left\|B(u-\tfrac{i}{n})\right\| + \|\mu(u)\| \\
		&\leq C \sum_{i=0}^\infty \rho^i + \|\mu(0)\| + \|\mu\|_{p-var} \quad 
		\leq C(A,B,\mu).
	\end{align*}
	In the last step, we used Lemma \ref{lem:eigenvalue-product} for some $\rho \in(\rho_0,1)$.
	In the Hölder continuous case, Lemma \ref{lem:eigenvalue-product} is applicable because $\beta$-Hölder continuity implies finite $\frac{1}{\beta}$-variation.
	In the case of finite $p$-variation, Lemma \ref{lem:eigenvalue-product} is directly applicable.
	
	We now consider the regularity of $u\mapsto G_n(u,\beps_0)$. 
	For any $u,v\in[0,1]$, we have 
	\begin{align}
		&\quad \|G_n(u,\beps_0) - G_n(v,\beps_0)\|_{L_q} \nonumber \\
		&\leq C \sum_{i=0}^\infty \left\| \left[ \prod_{j=1}^i A\left(u-\tfrac{j-1}{n}\right) \right]  B\left(u-\tfrac{i}{n}\right) - \left[ \prod_{j=1}^i A\left(v-\tfrac{j-1}{n}\right) \right]  B\left(v-\tfrac{i}{n}\right) \right\| + \|\mu(u)-\mu(v)\|  \nonumber  \\
		&\leq C \sum_{i=0}^\infty  \left\| \prod_{j=1}^i  A\left(u-\tfrac{j-1}{n}\right)\right\|  \left\|B\left(u-\tfrac{i}{n}\right) - B\left(v-\tfrac{i}{n}\right) \right\| + \|\mu(u)-\mu(v)\| \nonumber \\
		&\quad + C \sum_{i=0}^\infty \sum_{k=1}^i \left[ \left\| \prod_{\substack{j=1}}^{k-1} A\left(u-\tfrac{j-1}{n}\right) \right\| \left\|\prod_{\substack{j=k+1}}^{i} A\left(v-\tfrac{j-1}{n}\right) \right\| \right] \nonumber \\
		&\qquad\qquad \cdot  \left\| A(u-\tfrac{k-1}{n}) - A(v-\tfrac{k-1}{n}) \right\| \left\| B\left(v-\tfrac{i}{n}\right) \right\| \nonumber  \\
		&\leq C \sum_{i=0}^\infty  \rho^i \left\|B\left(u-\tfrac{i}{n}\right) - B\left(v-\tfrac{i}{n}\right) \right\|
		 + C \sum_{i=0}^\infty \sum_{k=1}^i \rho^{i-1} \left\| A(u-\tfrac{k-1}{n}) - A(v-\tfrac{k-1}{n}) \right\| \nonumber \\
		&\qquad +\|\mu(u)-\mu(v)\|.\label{eqn:tvVAR-increment} 
	\end{align}

	Now if $\mu$, $A$, and $B$ are $\beta$-Hölder continuous with Hölder constants $C_\mu$, $C_A$, and $C_B$, we find that
	\begin{align*}
		&\quad \|G_n(u,\beps_0) - G_n(v,\beps_0)\|_{L_q} \\
		&\leq C \sum_{i=0}^\infty  \rho^i |u-v|^\beta C_B	 + C \sum_{i=0}^\infty \sum_{k=1}^i \rho^{i-1} |u-v|^\beta C_A + C_\mu |u-v|^\beta\\
		&\leq |u-v|^\beta  (C_\mu+C_A + C_B)\, C \left[1+ \sum_{i=0}^\infty \rho^i + \sum_{i=0}^\infty i\, \rho^{i-1} \right],
	\end{align*}
	and the latter series are finite since $\rho<1$.
	
	Alternatively, we may treat the case of finite $p$-variation. 
	Let $u_0 < u_1< \ldots < u_m$.
	Applying the Minkowski inequality to \eqref{eqn:tvVAR-increment}, we find obtain
	\begin{align*}
		&\quad \left[\sum_{l=1}^m \|G_n(u_l, \beps_0) - G_n(u_{l-1},\beps_0)\|_{L_q}^p\right]^\frac{1}{p} \\
		&\leq C \sum_{i=0}^\infty \rho^i \left[ \sum_{l=1}^m \left\|B\left(u_l-\tfrac{i}{n}\right) - B\left(u_{l-1}-\tfrac{i}{n}\right) \right\|^p \right]^{\frac{1}{p}}\\
		&\quad + C \sum_{i=0}^\infty \sum_{k=1}^i \rho^{i-1}  \left[\sum_{l=1}^m \left\| A(u_l-\tfrac{k-1}{n}) - A(u_{l-1}-\tfrac{k-1}{n}) \right\|^p \right]^{\frac{1}{p}} + \|\mu\|_{p-var} \\
		&\leq C \sum_{i=0}^\infty \rho^i  \left\|B\right\|_{p-var} 
		 + C \sum_{i=0}^\infty \sum_{k=1}^i \rho^{i-1}  \left\| A \right\|_{p-var} + \|\mu\|_{p-var} \\
		& \leq C(\|B\|_{p-var} + \|A\|_{p-var}+\|\mu\|_{p-var}).
	\end{align*}
	Since $C=C(A,B,\mu)$, this completes the proof.
\end{proof}

\subsection{Properties of the CUSUM test}

\begin{proof}[Proof of Proposition \ref{prop:CUSUM}]
	By virtue of a result of \cite{lifshits1983absolute}, see the Corollary on p.~606 therein, we find that the distribution of $T^*$ is absolutely continuous, apart from a potential atom at $0$. 
	Since $\Var(M(u_0))>0$ for some $u_0\in(0,1)$, the variance structure of the process $M(u)$ yields that $\Var(M(u_0)-u_0M(1))>0$.
	Due to the Gaussianity of $M(u)$, it holds that $P(|M(u_0) - u_0 M(1)| =0)=0$, such that $P(T^*=0)=0$.
	Hence, the cumulative distribution function $H(x) = P(T^*\leq x)$ is continuous. 
	Denote furthermore by $H_n(x) = P(\widehat{T}_n | \mathds{X}_n)$ the (random) distribution function of the bootstrap approximation. 
	We have to show that $\E \left[H(H_n^{-1}(1-\alpha))\right]\to 1-\alpha$ as $n\to\infty$, where $H^{-1}$ denotes the quantile function corresponding to a cumulative distribution function $H:\R\to[0,1]$.
	
	Suppose to the contrary that $\E \left[H(H_n^{-1}(1-\alpha))\right]\not\to 1-\alpha$, then for some $\epsilon>0$, there exists a subsequence $n_k\to\infty$ as $k\to\infty$ such that $|\E\left[H(H_{n_k}^{-1}(1-\alpha))\right]- (1-\alpha)|>\epsilon$.
	Since $H_{n_k}\wconv H$ in probability, there exists a further subsequence, also denoted by $n_k$, such that $H_{n_k}\wconv H$ almost surely as $k\to\infty$.
	This implies that $H_{n_k}^{-1}(p)\to H^{-1}(p)$, for all $p\in[0,1]$ where $H^{-1}(p)$ is continuous \cite[Lemma 21.2]{van1998asymptotic}.
	Since $H^{-1}$ is left-continuous and has at most countably many discontinuities, we find that
	\begin{align*}
		H^{-1}(1-\alpha)\leq \liminf_{k\to\infty} H_{n_k}^{-1}(1-\alpha) \leq \limsup_{k\to\infty} H_{n_k}^{-1}(1-\alpha) \leq H^{-1}((1-\alpha)+),
	\end{align*}
	where $H^{-1}(p+)=\lim_{q\downarrow p} H^{-1}(q)$ denotes the right-hand limit.
	Since $F$ is continuous, we have $H(H^{-1}(1-\alpha)) = H (H^{-1}((1-\alpha)+)) = 1-\alpha$, such that $H(H^{-1}_{n_k}(1-\alpha))\to 1-\alpha$ almost surely. 
	Via the dominated convergence theorem, we obtain $\E \left[H(H_{n_k}^{-1}(1-\alpha))\right]\to (1-\alpha)$, completing the proof by contradiction.
\end{proof}

\begin{proof}[Proof of Proposition \ref{prop:power}]
	
	It holds that
	\begin{align*}
	T_n(u) &= \left[M_n(u) - \frac{(u-u_n)_+}{1-u_n} M_n(1)\right] + \frac{1}{n}\left[ \sum_{t=\tau_n+L_n}^{\lfloor un\rfloor} f(\mu^n_\frac{t}{n}) - \frac{(u-u_n)_+}{1-u_n} \sum_{t=\tau_n+L_n}^{n} f(\mu^n_\frac{t}{n})\right] \\
	&= \left[M_n(u) - \tfrac{(u-u_n)_+}{1-u_n} M_n(1)\right] + \Delta_n(u).
	\end{align*}
	By virtue of Theorem \ref{thm:CLT-integrated}, $\sqrt{n}[M_n(u) - \frac{(u-u_n)_+}{1-u_n}M_n(1)]\wconv M(u)$.
	Moreover,
	\begin{align*}
	\Delta_n(u) 
	&= \frac{1}{n} \left[ \sum_{t=\tau_n+L_n}^{\lfloor un\rfloor} f(\mu^n_\frac{t}{n}) - \frac{(\frac{\lfloor un\rfloor}{n}-u_n)_+}{1-u_n} \sum_{t=\tau_n+L_n}^{n} f(\mu^n_\frac{t}{n})\right] + \mathcal{O}(\tfrac{1}{n}) \\
	&= \frac{1}{n} \sum_{t=\tau_n+L_n}^{\lfloor un\rfloor} \left[f(\mu_\frac{t}{n}^n) - \frac{1}{n+1-\tau_n-L_n} \sum_{s=\tau_n+L_n}^{n} f(\mu^n_\frac{s}{n})\right] + \mathcal{O}(\tfrac{1}{n}).
	\end{align*}
	Now use that $f(\mu_u^n) = f(\mu_u) + \frac{1}{\sqrt{n}}Df(\mu_u) \delta_u + \mathcal{O}(1/n)$, uniformly in $u\in[0,1]$ by virtue of the boundedness of $f$.
	Since $f(\mu_u) = f(\mu_0)$ for all $u\in[0,1]$, we obtain
	\begin{align*}
	\Delta_n(u) &= \frac{1}{\sqrt{n}}\frac{1}{n} \sum_{t=\tau_n+L_n}^{\lfloor un\rfloor} \left[Df(\mu_\frac{t}{n})\delta_{\frac{t}{n}}  - \frac{1}{n+1-\tau_n-L_n} \sum_{s=\tau_n+L_n}^{n} Df(\mu_\frac{s}{n})\delta_{\frac{s}{n}}\right] + \mathcal{O}(\tfrac{1}{n}), \\
	\sqrt{n} \Delta_n(u) &= o(1) + \frac{1}{n} \sum_{t=1}^{\lfloor un\rfloor} Df(\mu_\frac{t}{n}) \delta_\frac{t}{n} - \frac{u}{n} \sum_{t=1}^{n} Df(\mu_\frac{t}{n}) \delta_\frac{t}{n} \\
	&= o(1) + \Delta(u).
	\end{align*}
	In the last step, we use Lemma \ref{lem:I1} and the fact that $\|u\mapsto Df(\mu_u) \delta_u\|_{p-var}<\infty$.
	Moreover, the convergence $\sqrt{n}\Delta_n(u)\to \Delta(u)$ holds uniformly in $u\in[0,1]$.
	This establishes the convergence of $T_n^*$.
	
\end{proof}

\bibliographystyle{apalike}

\bibliography{nonstationaryCP}

\end{document}